\documentclass[11pt, oneside]{amsart}
\usepackage[marginparwidth=2cm]{geometry}                % See geometry.pdf to learn the layout options. There are lots.
\usepackage{graphicx, slashed}
\usepackage{amssymb, amsthm, mathrsfs, bm}
\usepackage{epstopdf}
\usepackage{enumerate}
\usepackage{color}
\usepackage{hyperref}
\usepackage{bbm}
\numberwithin{equation}{subsection}

\newcommand{\mc}{\mathcal}
\newcommand{\gr}{\nabla}
\newcommand{\grt}{\tilde{\nabla}}

\newcommand{\grh}{\hat{\gr}}
\newcommand{\pa}{\partial}

\newcommand{\hf}{\frac{1}{2}}
\def\d[#1]{\text{d} #1}

\def\pd[#1]{\frac{\partial}{\partial {#1}}}
\def\pdt[#1][#2]{\frac{\partial {#1}}{\partial {#2}}}

\newcommand{\emt}{\mathbb{T}}
\newcommand{\mf}{\mathcal{M}}
\newcommand{\mh}{\mathcal{H}}

\newcommand{\mi}{\mathcal{I}}

\newcommand{\Lu}{\underline{L}}
\newcommand{\Hu}{\underline{H}}

\newcommand{\f}{\frac}
\newcommand{\bu}{\bar{u}}
\newcommand{\bv}{\bar{v}}

\newcommand{\rt}{\tilde{r}}

\newcommand{\fl}{\mathbb{F}}
\newcommand{\abs}[1]{\left|#1 \right|} 
\newcommand{\norm}[2]{\left|\left |#1 \right| \right |_{#2}}

\newcommand{\sr}{\mathscr{R}}
\makeatletter
\newtheorem*{rep@theorem}{\rep@title}
\newcommand{\newreptheorem}[2]{%
	\newenvironment{rep#1}[1]{%
		\def\rep@title{#2 \ref{##1}}%
		\begin{rep@theorem}}%
		{\end{rep@theorem}}}
\makeatother

\newtheorem{prop}{Proposition}
\newreptheorem{prop}{Proposition}

\newtheorem{thm}{Theorem}[section]
\newtheorem{defn}{Definition}[section]
\newtheorem{rem}{Remark}[section]
\newtheorem{lem}{Lemma}[section]

\newtheorem{coro}{Corollary}[section]
\title[Stability of Toroidal AdS-Schwarzschild]{Stability of the Toroidal AdS Schwarzschild Solution in the Einstein--Klein-Gordon System}
\author{Jake Dunn ${}^1$}

\author{Claude Warnick ${}^2$}
\thanks{\vspace{.1cm}\texttt{j.dunn15@imperial.ac.uk};  \texttt{c.m.warnick@maths.cam.ac.uk }\\
	\phantom{1    }\hspace{.05cm} ${}^1$ Dept. of Mathematics, South Kensington Campus, Imperial College London, SW7 2AZ, UK.\vspace{.1cm}\\
	\phantom{1    }\hspace{.05cm} ${}^2$ Centre for Mathematical Sciences, Wilberforce Road, Cambridge CB3 0WA, UK\vspace{.1cm}}

\date{\today \vspace{.1cm}}  
\begin{document}
	\maketitle
\begin{abstract}
We consider the stability of the toroidal AdS-Schwarzshild black holes as solutions of the Einstein--Klein-Gordon system, with Dirichlet or Neumann boundary conditions for the scalar field. Restricting to perturbations that respect the toroidal symmetry we show both orbital and asymptotic stability for the full nonlinear problem, for a range of choices of the Klein-Gordon mass. The solutions we construct with Neumann boundary conditions have a Hawking mass which diverges towards infinity, reflecting the infinite energy of the Klein-Gordon field for perturbations satisfying these boundary conditions.
\end{abstract}

	\tableofcontents

\section{Introduction}
The Einstein--Klein-Gordon system in an asymptotically anti de-Sitter space time is given by 
\begin{equation}\label{EKG}
\begin{split}
R_{\mu\nu}-\hf g_{\mu\nu}R - \frac{3}{l^2}g_{\mu\nu} &= 8\pi T_{\mu\nu},\\
\Box_g \psi - \frac{2a}{l^2}\psi &=0,\\
\gr_\mu\psi\gr_\nu\psi - \hf g_{\mu\nu}\gr_\sigma\psi\gr^\sigma\psi - g_{\mu\nu}\frac{a}{l^2}\psi^2 &= T_{\mu\nu}. 
\end{split}
\end{equation}
Here $g$, $\psi$ are, respectively, the $(3+1)$-dimensional Lorentzian metric and Klein-Gordon field for which we solve. $R_{\mu\nu}$ is the Ricci curvature of $g$, $R$ the scalar curvature, $l$ the AdS radius related to the cosmological constant of the system, $\Lambda$, through the relationship  $\Lambda = \f{-3}{l^2}$, and $-\frac{9}{8}< a< -\frac{5}{8}$ is a negative constant which can be thought of as the mass of the Klein-Gordon equation. It will be convenient to write $a = \frac{1}{2} \kappa^2 - \frac{9}{8}$, where $\kappa\in(0,1)$. 

Upon fixing coordinates, (\ref{EKG}) becomes a quasilinear system of hyperbolic PDE. In contrast to the situation for non-negative cosmological constant, the natural setting to solve this system is that of an initial-boundary value problem \cite{friedrich_einstein_1995,holzegel_self-gravitating_2012,holzegel_einstein-klein-gordon-ads_2013, Enciso:2014lwa}. In addition to the usual Cauchy data specified on a spacelike hypersurface, one is led to impose boundary conditions on a timelike surface `at infinity', $\mi$. Among possible choices of boundary data, one expects to be able to fix the conformal metric at $\mi$ together with a choice of Dirichlet or Neumann boundary conditions for $\psi$.  

The simplest setting to consider this system is under the assumption of spherical symmetry. There are two stationary vacuum ($\psi \equiv 0$) solutions,  the anti-de Sitter spacetime and the Kottler or AdS-Schwarzshild solution. The nonlinear stability of the AdS spacetime against Klein-Gordon perturbations, with Dirichlet boundary conditions for the Klein-Gordon field, has been extensively studied numerically, since the pioneering work of Bison and Rostorowski in \cite{bizon_weakly_2011}. The numerical work in this paper, together with perturbative arguments therein, suggest that there exist initial data arbitrarily close to AdS which collapse to form a black hole in finite time, giving evidence for the instability conjectured by Anderson \cite{Anderson} and Dafermos-Holzegel \cite{DafHol}. In a more rigorous setting, black hole formation has been shown to occur for the related Einstein-null dust model by Moschidis \cite{moschidis_proof_2017,moschidis_einstein--null_2017}. We also note the work of Dold \cite{dold_global_2017} who has established that for the $(1+4)$-dimensional AdS-Eguchi-Hansen solutions, with negative mass,  no horizon may form when  the spacetime is perturbed by data in the Bianchi IX symmetry class.

The spherically symmetric AdS-Schwarzschild black hole has also been well studied in the mathematical literature. In a series of papers, Holzegel \cite{holzegel_massive_2010} and Holzegel-Smulevici \cite{holzegel_decay_2011} established that solutions of the Klein-Gordon equation satisfying Dirichlet boundary conditions decay in time, but at a slow rate: no rate of decay faster than an inverse power of the logarithm can hold uniformly. The reason for this slow decay is a stable trapping mechanism which operates near null infinity. With more general boundary conditions \cite{warnick_massive_2013}, the problem of uniform boundedness of the linear Klein-Gordon field on this background was studied in \cite{holzegel_boundedness_2014}. At the nonlinear level, in \cite{holzegel_stability_2013} an orbital and asymptotic stability result was established for small perturbations of the spherically symmetric AdS-Schwarzschild spacetime, measured with respect to the standard Klein-Gordon norm.

Moving away from spherical symmetry, by varying the geometry of the conformal boundary one may construct a vast array of stationary vacuum spacetimes, with or without a black hole \cite{anderson_non-trivial_2002,chrusciel_non-singular_2007,chrusciel_non-singular_2017-2,chrusciel_non-singular_2017-1,chrusciel_non-singular_2017,chrusciel_non-singular_2018}. A family of solutions of particular interest generalise the AdS-Schwarzschild solutions \cite{kottler_uber_1918,lemos_two-dimensional_1995} (See also \cite{Birmingham:1998nr}). In Schwarzschild-like coordinates, the spacetime manifold is $\mathbb{R}_t \times (r_+, \infty) \times \Sigma$ and the metric takes the form:
\begin{equation}
g= -\left( k-\frac{2M}{r} + \frac{r^2}{l^2} \right) dt^2 + \frac{dr^2}{k-\frac{2M}{r} + \frac{r^2}{l^2}} + r^2 d\sigma_k^2 \label{KotDef}
\end{equation}
where $d\sigma_k^2$ is a fixed metric of constant curvature $k\in\{1,0,-1\}$ on the surface $\Sigma$ which we assume to be compact\footnote{This is not strictly necessary: one could also consider open $\Sigma$. In particular, our results apply equally to planar black holes where perturbations respect the symmetries of the plane.}, $M>0$ is a parameter, and $r_+$ is the largest root of the equation $kr +l^{-2} r^3-2M=0$. The spacetime can be continued in the usual way across the horizon at $r_+$. These spacetimes have been extensively studied in the physics literature due to the conjectured correspondence between gravitational systems with negative cosmological constant, and conformal field theories \cite{maldacena_large_1999, Polchinski:2010hw}

The choice $k=1$ corresponds to a spherical horizon, $k=0$ to a toroidal horizon, and $k=-1$ to a horizon with higher genus. In all cases $\mi$ inherits the geometry of the horizon. In this paper, we shall address the stability of the $k=0$ spacetime, for which $\Sigma$ is a torus and $d\sigma_0$ is the flat metric. The main results we shall establish are summarised as:
\begin{thm}
For a range of $\kappa$, the $k=0$ AdS-Schwarzschild is stable as a solution to the Einstein--Klein-Gordon system with homogeneous Neumann or Dirichlet boundary conditions for $\psi$, against small perturbations respecting the toroidal symmetry. The maximal development of a perturbed initial data set contains an asymptoticaly AdS region, bounded by an event horizon, in which the solution settles down to an AdS-Schwarzschild solution exponentially quickly with respect to an appropriate coordinate.
\end{thm}

This may be thought of as an extension of the results of Holzegel-Smulevici from the spherical to toroidal setting. There are three key ways in which our results differ from \cite{holzegel_stability_2013}. Firstly, the change from $k=1$ to $k=0$ creates some modifications to the equations we study. These modifications can typically be treated as error terms and do not cause any significant analytical difficulty. Secondly, we consider the problem with Neumann boundary conditions for the Klein-Gordon field. This is a much more profound change, since the standard Klein-Gordon energy diverges for perturbations obeying Neumann boundary conditions. In particular, the perturbations we consider are not `small' in the sense of \cite{holzegel_stability_2013}, indeed generically the perturbations are of infinite size in the norms they consider. In order to deal with this, we introduce a renormalised Hawking mass and establish that it satisfies a monotonicity property. This both allows us to consider the more general boundary conditions as well as streamlining some of the arguments. Finally, as a consequence of the improved local existence theory of \cite{holzegel_einstein-klein-gordon-ads_2013} we are able to work at the $H^1$ level of regularity in contrast to the $H^2$ regularity assumed in Holzegel-Smulevici. This of course also allows us to consider perturbations that are large in comparison to those of \cite{holzegel_stability_2013}.

In contrast to the situation for spherically symmetric AdS spacetimes, it is possible to have non-trivial vacuum perturbations which correspond to $r-$dependent deformations of the metric of the toroidal part of the $k=0$ AdS-Schwarzschild black hole. Perturbations similar to this are considered in the physics literature \cite{Chesler:2008hg, Myers:2017sxr}. The enlarged symmetry class within which this places the AdS-Schwarzschild black hole includes the AdS-Soliton metrics, which play an analogous role to the AdS spacetime in the spherically symmetric setting: i.e.\ they are complete vacuum spacetimes with no horizons.  We show that for perturbations such that the tori of symmetry are rectangular, the resulting evolution equations are equivalent to the Einstein--Klein-Gordon problem considered above, with Dirichlet boundary conditions for the scalar field. As a consequence, we establish asymptotic stability for a class of AdS black holes against non-trivial \emph{vacuum} perturbations. 

The paper proceeds as follows, in \S\ref{EKGSection} we reduce the Einstein--Klein-Gordon equations (\ref{EKG}) to a system of PDE by making a symmetry ansatz consistent with the toroidal symmetry and fixing double-null coordinates. We also introduce the renormalised Hawking mass, which will prove crucial in the analysis of the system. In \S \ref{WPS} we give a brief treatment of the wellposedness of the reduced system of equations, and in \S\ref{GU} discuss the issue of geometric uniqueness for this problem, as well as establishing the extension principles that we require in \S\ref{EP}. In \S\ref{EKGSection}--\S\ref{EP}, where we are concerned mainly with setting the problem up, many of the results are minor modifications of previous results known in the spherically symmetric setting. Where this is the case, for brevity we do not provide all of the details of the proofs. 

In \S\ref{IDS} we discuss the class of perturbations that we shall consider, in particular we introduce the norms with respect to which we assume smallness. In \S \ref{OSS} we establish the first main claim of Theorem \ref{MT}, which may be thought of as an orbital stability result, namely that the maximal development contains a region qualitatively similar to the exterior region of the AdS-Schwarzschild black hole. The key to this result is a bootstrap argument making use of the monotonicity of the renormalised Hawking mass. In Section \S \ref{AS} we complete the proof of the main theorem by establishing asymptotic stability, i.e. that the perturbed spacetime asymptotically approaches the AdS-Schwarzschild spacetime exponentially quickly in a suitable sense. We conclude with a discussion of the vacuum case.

 \section{The Einstein--Klein-Gordon system}\label{EKGSection}
 In this section we introduce the symmetry class of metrics we consider and reduce the Einstein--Klein-Gordon system under this ansatz. We also discuss the toroidal AdS Schwarzschild spacetime.
 
%The Einstein--Klein-Gordon system in an asymptotically anti de-Sitter space time is given by 
%\begin{equation}\label{EKG}
%\begin{split}
%R_{\mu\nu}-\hf g_{\mu\nu}R - \frac{3}{l^2}g_{\mu\nu} &= 8\pi T_{\mu\nu},\\
%\Box_g \psi - \frac{2a}{l^2}\psi &=0,\\
%\gr_\mu\psi\gr_\nu\psi - \hf g_{\mu\nu}\gr_\sigma\psi\gr^\sigma\psi - g_{\mu\nu}\frac{a}{l^2}\psi^2 &= T_{\mu\nu}. 
%\end{split}
%\end{equation}
%Here $g$, $\psi$ are the Lorentzian metric and Klein-Gordon field which we are solving for respectively. $R_{\mu\nu}$ is the Ricci curvature, $R$ the scalar curvature, $l$ the AdS radius related the cosmological constant of the system $\Lambda$ through the relationship  $\Lambda = \f{-3}{l^2}$, $a$ is a negative constant which can be thought of as the mass of the Klein-Gordon equation. We define a parameter related to $a$ given by $\kappa= \sqrt{\f{9}{4}+2a}$, we restrict that $\kappa\in(0,1)$.
\subsection{System reduction}~\\ 
Recall that we seek solutions to (\ref{EKG}). We label the spacetime coordinates by $(u,v,x,y)$. From \cite{gowdy_vacuum_1974}, it is known that imposing a global toroidal symmetry on the spacetime enforces the product form $\mf=\mc{Q}^+\times \mathfrak{T}^2$. Where $\mc{Q}^+$ is a two dimensional Lorentzian manifold, and $\mathfrak{T}^2= \mathbb{R}^2/\mathbb{Z}^2$. Furthermore, we may assume the metric takes the form
\begin{equation}
\begin{split}
g = &-\Omega^2(u,v)dudv+r^2(u,v)\left(A(u,v)dx+B(u,v)dy \right)^2+(B(u,v)dx+C(u,v)dy)^2,
\end{split}
\end{equation}
where $AC-B^2=1$, and $A+C>0$.\\ We shall initially consider the case $A=C=1$ and $B=0$, to retain similarity to the spherical problem as in \cite{holzegel_einstein-klein-gordon-ads_2013}. This gives the torus the properties of being square and flat.  
\begin{defn} \label{C3SFTS}
	We say a spacetime has a square flat toroidal symmetry if it has topology $\mf=\mc{Q}^+\times \mathfrak{T}^2$, and can be equipped with a metric of the form 
	\begin{equation}\label{met}
	g = -\Omega^2(u,v)dudv+r^2(u,v)\left(dx^2+dy^2\right). 
	\end{equation}
\end{defn}
\begin{lem}
	For a metric of the form \eqref{met}, The system \eqref{EKG} reduces to
	\begin{equation}\label{EKG1}
	\pa_u\left(\frac{r_u}{\Omega^2} \right) = -4\pi r\frac{\left(\pa_u\psi\right)^2 }{\Omega^2},
	\end{equation}
	\begin{equation}\label{EKG2}
	\pa_v\left(\frac{r_v}{\Omega^2} \right) = -4\pi r\frac{\left(\pa_v\psi\right)^2}{\Omega^2},
	\end{equation}
	\begin{equation}\label{EKG3}
	r_{uv} = - \frac{r_ur_v}{r}+ \frac{2\pi a r}{l^2}\Omega^2\psi^2 - \frac{3}{4}\frac{r}{l^2}\Omega^2,
	\end{equation}
	\begin{equation}\label{EKG4}
	\left(\log\Omega\right)_{uv} = - 4\pi \pa_u\psi\pa_v\psi+\frac{r_ur_v}{r^2},
	\end{equation}
	\begin{equation}\label{EKG5}
	\pa_u\pa_v\psi =-\frac{r_u}{r}\psi_v-\frac{r_v}{r}\psi_u-\frac{\Omega^2a}{2l^2}\psi.
	\end{equation}
\end{lem}
\begin{proof}
	To see this we need to compute the Ricci curvature and evaluate both sides of \eqref{EKG}. \eqref{EKG1} and \eqref{EKG2} are the $uu$ and $vv$ components respectively. \eqref{EKG3} comes from the $uv$ component and \eqref{EKG4}, \eqref{EKG5} follow from the other components. Conversely if $r, \Omega, \psi$ solve \eqref{EKG1}-\eqref{EKG5}, the metric \eqref{met} with Klein--Gordon field $\psi$ solves \eqref{EKG}.\\
\end{proof}
\subsection{Renormalised Hawking mass}~\\ 
Making a suitable modification to \cite{holzegel_einstein-klein-gordon-ads_2013}, \cite{holzegel_self-gravitating_2012}, and \cite{holzegel_stability_2013}, we define the first renormalised Hawking mass as 
\begin{equation}\label{DHM}
\varpi_1 = \f{2r_ur_vr}{\Omega^2} + \f{r^3}{2l^2}.
\end{equation}
Through equations \eqref{EKG1}-\eqref{EKG5} we see that, this quantity satisfies the transport equations:
\begin{equation}\label{HM1}
\pa_u\varpi_1 = -8\pi r^2\f{r_v}{\Omega^2}(\pa_u\psi)^2 + \f{4\pi r^2a}{l^2}r_u\psi^2,
\end{equation}
\begin{equation}\label{HM2}
\pa_v\varpi_1 = -8\pi r^2\f{r_u}{\Omega^2}(\pa_v\psi)^2 + \f{4\pi r^2a}{l^2}r_v\psi^2.
\end{equation}
We may replace some of the original system of equations with \eqref{HM1} and \eqref{HM2}. This follows from the following lemma (where we assume derivatives to be taken in a weak sense).
\begin{lem}
	Suppose \eqref{EKG3}, \eqref{EKG5}, \eqref{HM1} and \eqref{HM2}, hold (where $\Omega$ is defined through \eqref{DHM}). Then we have that \eqref{EKG1} and \eqref{EKG2} hold. Furthermore if \eqref{EKG3} can be differentiated in $u$, then \eqref{EKG4} holds.  	
\end{lem} 
We may also express the wave equation for $r$ in terms of $\varpi_1$ as
\begin{equation}
r_{uv} = -\f{\Omega^2}{2}\left(\f{\varpi_1}{r^2}+\f{r}{l^2} \right) + \f{2\pi a}{l^2} \Omega^2\psi^2. 
\end{equation}
Similar to the situation in \cite{holzegel_stability_2013}, \eqref{HM1}, and \eqref{HM2} imply that we can think of the Hawking mass as the potential for a weighted $H^1$ energy. \\
\subsection{Toroidal AdS Schwarzschild solution}~\\ 
In standard Schwarzschildean coordinates, the toroidal AdS Schwarzschild spacetime is the manifold $\mf=\mathbb{R}\times (r_{+},\infty) \times \mathfrak{T}^2$, with Lorentzian metric
\begin{equation}\label{SM}
g = -\left(\f{-2M}{r}+\f{r^2}{l^2} \right)dt^2 +\left(\f{-2M}{r}+\f{r^2}{l^2} \right)^{-1}dr^2+r^2\left(dx^2+dy^2\right), 
\end{equation}
where $M,l>0$, and $r_+=\left(2Ml^2 \right)^{\f{1}{3}}$. This spacetime is asymptotically AdS, and as such a timelike boundary formally given as $\mi=\{r=\infty\}$ can be attached to the manifold.
As is typical in these coordinates the metric \eqref{SM} becomes singular on the hypersurface $\mh = \{r=r_{+}\}$. As such these coordinates are not ideally suited for analysis at the horizon. To circumvent this one defines a tortoise coordinate 
\begin{equation}\label{CH1KM}
r_*(r) := \int_{2r_{+}}^r\left(\f{-2M}{X}+\f{X^2}{l^2}\right)^{-1} dX,
\end{equation}
and can find a maximal analytic extension of the manifold by following a Kruskal-style argument. The full details of this can be found in \cite{holzegel_stability_2013} with minor modifications to the toroidal setting. \\ \\
In this paper we will be working with a double null coordinate system, the standard choice for \eqref{SM} is the Eddington Finkelstein chart, defined through the transformation
\begin{equation}
u = t-r_*,\quad v= t+r_*.
\end{equation}
The metric is then
\begin{equation}\label{EFM}
g = -\left(\f{-2M}{r}+\f{r^2}{l^2} \right)dudv+r^2\left(dx^2+dy^2\right),
\end{equation}
where $r$ is now a function of $u,v$ satisfying the following differential relations 
\begin{equation}
-2r_u=2r_v=\left(\f{-2M}{r}+\f{r^2}{l^2} \right).
\end{equation}
However these coordinates also degenerate at the event horizon. To extend through this degeneration we can fix a hypersurface $\{v=v_0>0\}$ and denote its past intersection with $\mi$ by $(u_0,v_0$). We now make the $u$ coordinate transformation along this surface given by the solution to
\begin{equation}
\f{d\hat{u}}{du} = \frac{\left(\f{-2M}{r}+\f{r^2}{l^2} \right)}{\frac{r^2}{l^2}},
\end{equation}
with $\hat{u}(u_0)=u_0$.  In the coordinate $\hat{u}$, the metric is now regular at $r=r_{+}$, and we may extend the spacetime to cover the region shown in the Penrose diagram \\
\begin{figure}[h!]
	\center\includegraphics[scale=0.7]{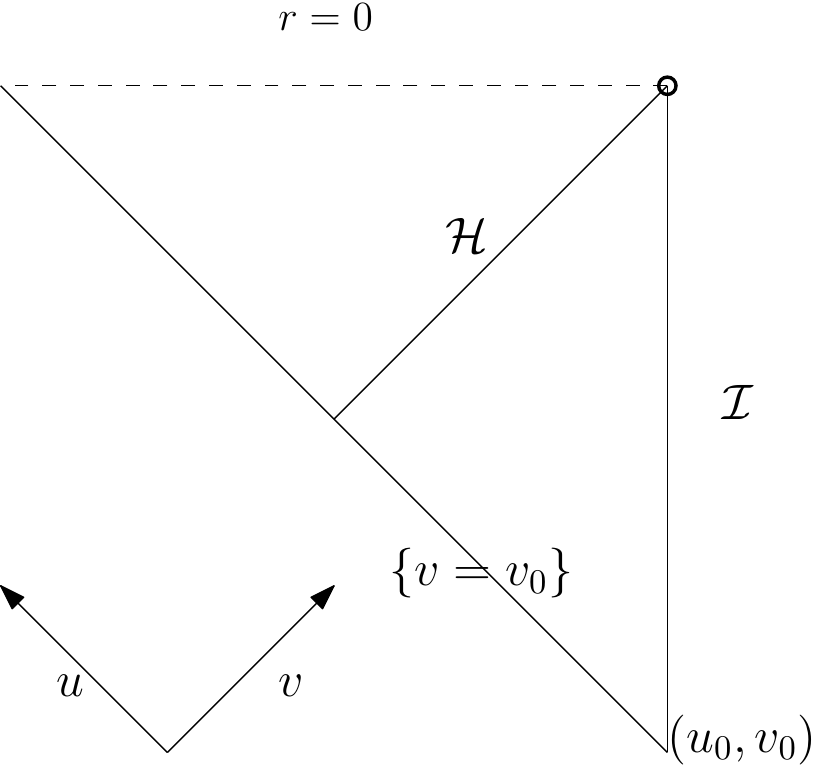}
	\caption{Penrose diagram for regularised Eddington Finkelstein chart.}
\end{figure}\\
Along this surface, $r$ satisfies:
\begin{equation}\label{GC}
-2r_{\hat{u}}=\frac{r^2}{l^2}.%\quad 2r_v=\left(\f{-2M}{r}+\f{r^2}{l^2} \right).
\end{equation} 
Fixing this relation is equivalent to fixing the $u$ coordinate along the ray $\{v=v_0\}$.
\subsection{The Klein-Gordon mass}
Throughout this paper three quantities related to the Klein-Gordon mass are used fairly interchangeably. This is largely to clean up the algebra. In order improve clarity we collect them here
\begin{itemize}
	\item $a$, denotes the Klein-Gordon mass,
	\item $g$, denotes the radial decay of the field, $g = -\f{3}{2}+\sqrt{\f{9}{4}+2a}$,
	\item $\kappa$, denotes the radical part of $g$, $\kappa =\sqrt{\f{9}{4}+2a}.$ 
\end{itemize}
We now collect some of the key values of the quantities and how they relate in a table below
\begin{center} 
	$\begin{array}{|c|c|c|c|}
	\hline
	&\text{BF Lower Bound} &  \text{ Conformal}& \text{BF Upper Bound}  \\
	\hline
	a & -\f{9}{8}  & -1 &-\f{5}{8}\\
	\hline
	\kappa & 0 &\hf &1\\
	\hline
	g & -\f{3}{2}&-1&-\hf\\
	\hline
	\end{array}$
\end{center}
\section{Wellposedness of the initial-boundary-value problem}
%\begin{rem}
In this section we briefly discuss the well-posedness of reduced Einstein--Klein-Gordon system introduced above. This follows closely the discussion in \cite{holzegel_einstein-klein-gordon-ads_2013}. The variables in \eqref{EKG1}-\eqref{EKG5}, are inconvenient to analyse the system. We anticipate the behaviour $\Omega^2\sim r^2,\psi \sim r^{-\f{3}{2}+\kappa}$ at the conformal boundary (where we expect $r\to \infty$). Furthermore when we introduce Neumann boundary conditions to the problem, the variable $\varpi_1$ will no longer form a potential for a finite $H^1$ energy. The quantity will diverge as $r\to \infty$. To rectify these issues we proceed by solving an equivalent system that has undergone a renormalisation scheme. 
%\end{rem}~\\
\subsection{Renormalised system}~\\ 
Motivated by \cite{holzegel_einstein-klein-gordon-ads_2013}, and the linear theory of \cite{dunn_kleingordon_2016}, we introduce the twisted derivative
\begin{equation}
\grt_\mu \psi = f\gr_\mu\left(\f{\psi}{f}\right), 
\end{equation} 
for a $C^1$ function $f$. We make the natural choice of twisting function:
\begin{equation}
r^g, \text{ where } g= -\f{3}{2} + \kappa.
\end{equation}
We define the second renormalised Hawking mass to be
\begin{equation}
\varpi_2 = \varpi_1 - 2\pi g\f{r^3}{l^2}\psi^2.
\end{equation}
The latter term has been introduced to cancel a divergence in $\varpi_1$ that appears as the boundary is approached.
\begin{lem}\label{RenLem}
	Define the variables 
	\begin{equation}
	r=\f{1}{\rt}, \quad \varpi_1 = \varpi_2 + 2\pi g \f{r^3}{l^2}\psi^2, 
	\end{equation} 
	then \eqref{EKG} is equivalent to
	\begin{equation}\label{REKG1}
	\pa_u\varpi_2=  -8\pi r^2\frac{r_v}{\Omega^2}(\grt_u\psi)^2 - 8\pi g\left(\varpi_2+2\pi g\f{r^3}{l^2}\psi^2 \right) \psi\grt_u\psi-4\pi\psi^2r_ug^2\f{\left( \varpi_2+2\pi g\f{r^3}{l^2}\psi^2\right) }{r},
	\end{equation}
	\begin{equation}\label{REKG2}
	\pa_v\varpi_2=-8\pi r^2\frac{r_u}{\Omega^2}(\grt_v\psi)^2 - 8\pi g\left(\varpi_2+2\pi g\f{r^3}{l^2}\psi^2 \right) \psi\grt_v\psi-4\pi\psi^2r_vg^2\f{\left( \varpi_2+2\pi g\f{r^3}{l^2}\psi^2\right) }{r},
	\end{equation}
	\begin{equation}\label{REKG3}
	\tilde{r}_{uv} = \Omega^2\rt^2\left(\f{3}{2}\varpi_2\rt^2-\f{2\pi g^2}{l^2\rt}\psi^2 \right),
	\end{equation}
	\begin{equation}\label{REKG4}
	\begin{split}
	\pa_v\left(r\grt_u\psi \right)&= \left( \kappa-\hf\right) r_u\grt_v\psi -\f{\Omega^2}{4}rV\psi,
	\end{split}
	\end{equation}
	with auxiliary variables
	\begin{equation}\label{auxvar}
	\Omega^2=-\f{4r^4\rt_u\rt_v}{\mu_1}, \quad \mu_1= -\f{2\varpi_1}{r}+\f{r^2}{l^2}, \quad 	V = \f{2g^2}{r^3}\varpi_1 +\f{8\pi ag}{l^2}\psi^2.
	\end{equation}
	we remark that equation \eqref{REKG4} can also be expressed as
	\begin{equation}\label{REKG5}
	\pa_u\left(r\grt_v\psi \right)= \left( \kappa-\hf\right) r_v\grt_u\psi -\f{\Omega^2}{4}rV\psi.
	\end{equation}
\end{lem} 
\begin{proof}
	This is a straightforward, if slightly lengthy calculation. We study the derivatives of the variables $\rt$, $\varpi_2$, and use the equations $\eqref{EKG1}$ - $\eqref{EKG5}$ to simplify. For the Klein-Gordon equation we substitute the definition of the twisted derivative and simplify. 
\end{proof}
\textbf{The domain} \newline
We seek to construct a solution of the equations \eqref{EKG1}-\eqref{EKG5} in a small triangular domain of the form
\begin{equation}
\Delta_{\delta,u_0} := \{(u,v)\in \mathbb{R}^2: u_0\le v\le u_0+\delta, v<u\le u_0+\delta \},
\end{equation}
We shall impose initial conditions on the ray $\{v=u_0\}$ and boundary conditions on
\begin{equation}
\mi = \overline{\Delta}_{\delta,u_0} \backslash \Delta_{\delta,u_0} = \{(u,v)\in\overline{\Delta}_{\delta,u_0}:u=v\}.
\end{equation}\\
\begin{figure}[h!]
	\begin{center}	
		\includegraphics[scale=0.7]{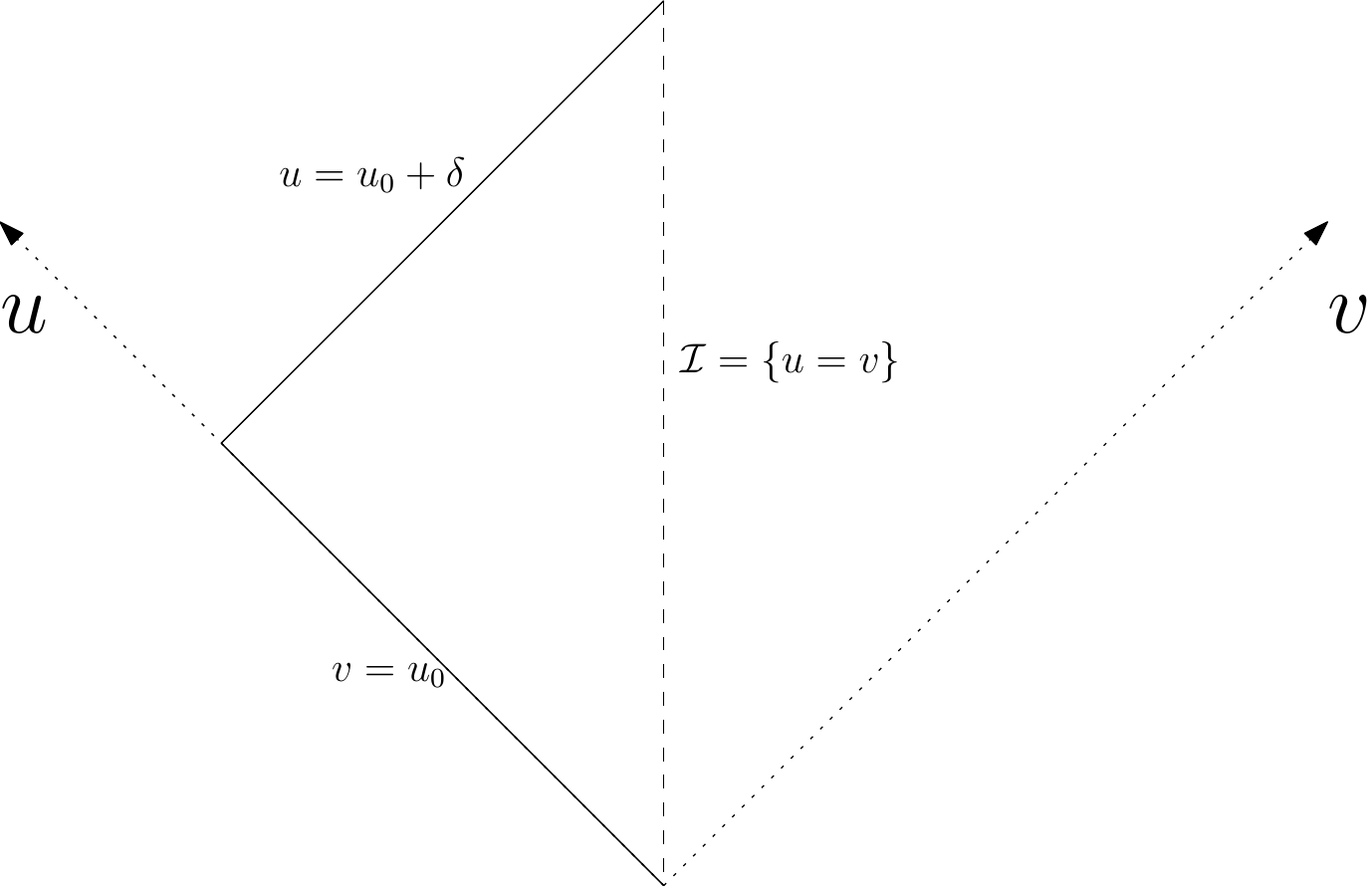}
		\caption{Diagram of $ \Delta_{\delta,u_0}$}
	\end{center}
\end{figure}
\\
It is sometimes useful to refer to the coordinates
\begin{equation}
t = \f{u+v}{2}, \qquad \rho = \frac{u-v}{2}
\end{equation}
Note that $\mi = \{\rho=0\}$ is parameterised by $t$ and further that the vector field
\begin{equation}
T= \pa_u+\pa_v,
\end{equation}
is tangent to $\mi$.
\subsection{Initial data and boundary conditions}\label{IaBd}
\subsubsection{Initial data}
\begin{defn}\label{CH3IDS}
	Let $\mc{N}=(u_0,u_1]$ be a real interval. Then a \textit{free data} set is a pair of functions $(\overline{\rt},\overline{\psi})\in C^2(\mc{N})\times C^1(\mc{N})$ such that:
	\begin{itemize}
		\item $\overline{\rt}>0$ and $\overline{\rt}_u>0$ in $\mc{N}$, as well as $\lim_{u\to u_0}\overline{\rt}_u = \f{1}{2l^2}$ and $\lim_{u\to u_0} \overline{\rt}_{uu}=0$.
		\item We have the bounds\\
		\begin{equation}\label{WPID1}
		\int_{u_0}^{u_1}\left[\left( \overline{\grt}_u\overline{\psi}\right)^2+\overline{\psi}^2  \right](u-u_0)^{-2}du < \infty, 
		\end{equation}
  \begin{equation}\label{WPID2}
		\sup_{\mc{N}}\abs{\overline{\psi}\cdot\overline{\rt}^{-\f{3}{2}+\kappa}}+ \sup_{\mc{N}}\abs{r^\hf\overline{\grt}_u\overline{\psi}}<\infty.
		\end{equation}
	\end{itemize}
	Here $\overline{\grt}_u$ is a twisted derivative with twisting function: $\overline{f} = \left( \hf(u-u_0) \right) ^{\f{3}{2}-\kappa}$.
\end{defn} 
With a free data set we are now able to construct a \textit{complete initial data} set $(\overline{\rt},\overline{\psi},\overline{\varpi_2},\overline{\rt_v})$.\\
Let $M_N>0,$ we define $\overline{\varpi_2}$ as the unique $C^1(\mc{N})$ solution to:
\begin{equation}
\begin{split}
\pa_u\overline{\varpi_2}=  &2\pi \f{\overline{r}^2}{\overline{r}_u}\left(-\f{2\overline{\varpi_2}}{\overline{r}}+\f{\overline{r}^2}{l^2} - 4\pi g\f{\overline{r}^2}{l^2}\overline{\psi}^2 \right) (\overline{\grt}_u\overline{\psi})^2 - 8\pi g\left(\overline{\varpi_2}+2\pi g\f{\overline{r}^3}{l^2}\overline{\psi}^2 \right) \overline{\psi}\overline{\grt}_u\overline{\psi}\\&-4\pi\overline{\psi}^2\overline{r}_ug^2\f{\left( \overline{\varpi_2}+2\pi g\f{\overline{r}^3}{l^2}\overline{\psi}^2\right) }{\overline{r}},
\end{split}
\end{equation}
with boundary condition
\begin{equation}
\lim_{u\to u_0} \overline{\varpi_2}= M_N.
\end{equation}
We define $\overline{\rt_v}$ is a similar way, as the unique $C^1(\mc{N})$ solution of the ODE
\begin{equation}
\pa_u\overline{\tilde{r}_{v}} = \f{\overline{\rt}^2\overline{\rt}_u\overline{\rt_v}}{-\f{2\overline{\varpi_2}}{\overline{r}}+\f{\overline{r}^2}{l^2}-4\pi g\f{\overline{r}^2}{l^2}\overline{\psi}^2 }\left(\f{3}{2}\overline{\varpi_2}\overline{\rt}^2-\f{2\pi g^2}{l^2\overline{\rt}}\overline{\psi}^2 \right),
\end{equation}
with boundary condition
\begin{equation}
\lim_{u\to u_0}\overline{\rt_v}= -\f{1}{2l^2}.
\end{equation}
\begin{rem}
	The choice of $\overline{\rt}$ is equivalent to choosing the scale of the $u$ coordinate along $\mc{N}$. It represents the gauge freedom of this problem.\\
	The choice of $\overline{\psi}$ is free, providing the conditions \eqref{WPID1}, and \eqref{WPID2} hold. The value $M_N$ is free provided it is strictly positive.\\
	The choice of boundary condition for $\overline{\rt_v}$ is to ensure that initially $\overline{\rt}_u+\overline{\rt_v} =0$, the (later chosen) boundary conditions ensure this propagates along $\mi$.
\end{rem}~\\
\subsubsection{Boundary conditions}~\\ \\
%\textbf{Notation}\\
%We now define the function
%\begin{equation}
%\rho = \f{u-v}{2},
%\end{equation}
%and denote twisting with the function $\rho$ by
%\begin{equation}
%\grh_\mu\psi := \rho^{\f{3}{2}-\kappa}\gr_\mu(\psi\rho^{-\f{3}{2}+\kappa}).
%\end{equation}
%The $\Hu^1$ norm over a set $\mc{U}\subset\Delta_{\delta,u_0}$ is given by
%\begin{equation}
%\begin{split}
%\norm{\psi}{\Hu^1(\mc{U})}^2 := \int_\mc{U}\left(\rho^{-2}\left(\abs{\grh\psi}^2 \right) + \rho^{-2}\psi^2  \right)dudv.
%\end{split}
%\end{equation}
%The space $\Hu_0^1(\Delta_{\delta,u_0})$ is given by the completion of $C_c^\infty(\Delta_{\delta,u_0})$ in the $\Hu^1$ norm.\\ \\
%If a function $\phi \in \Hu^1\left(\Delta_{\delta,u_0} \right)$ has the property that $\phi|_{\{v=u_0\}}=\phi|_{\{u=u_0+\delta\}}=0$ in a trace sense, then we say it is a test function for the Neumann problem. We denote the set of these functions by $\mc{T}_\mc{N}(\Delta_{\delta,u_0})$.\\
\textbf{Boundary conditions for $\rt$}\\
In order that we produce a spacetime that is asymptotically AdS we will insist 
\begin{equation}\label{metricbc}
\rt|_\mi = 0.
\end{equation}
As a consequence of $T := \pa_u+\pa_v$ being tangent to $\mi$, we see that
\begin{equation}
T(\rt) = 0,
\end{equation}   
along $\mi$.\\ \\
\textbf{Boundary conditions for $\psi$}\\
The boundary conditions we seek to impose on the field $\psi$ are either the Neumann conditions:
\begin{equation}\label{BC1}
\rho^{-\hf-\kappa}\left(\grt_\rho\psi \right)=0,
\end{equation}
or else the Dirichlet conditions:
\begin{equation}\label{BC2}
\rho^{-\f{3}{2}+\kappa}\psi = 0.
\end{equation}
Since we seek solutions at the $C^0\Hu^1$-level of regularity, we should understand these boundary conditions to hold in a weak sense, see \cite{Dunn_Thesis_2018}.

\subsection{Wellposedness}\label{WPS}~
\subsubsection{Regularity}~ \\
In this section we discuss the regularity we require for a weak solution to exist.
We define the $C^0\Hu^1$ norm by
\begin{equation}
\begin{split}
\norm{\psi}{C^0\Hu^1(\Delta)}^2 := &\sup_{(u,v)\in\Delta} \int_v^u\left(\rho^{-2}\left(\grh_u\psi \right)^2 + \rho^{-2}\psi^2  \right)du'\\ + &\sup_{(u,v)\in\Delta} \int_{v_0}^v\left(\rho^{-2}\left(\grh_v\psi \right)^2 + \rho^{-2}\psi^2  \right)dv',
\end{split}
\end{equation}
where
\begin{equation}
\grh_\mu\psi := \rho^{\f{3}{2}-\kappa}\gr_\mu(\psi\rho^{-\f{3}{2}+\kappa}).
\end{equation}
\begin{defn}
	A weak solution to the renormalised Einstein--Klein-Gordon system is an element of the function space:
	\begin{equation}
	\mathfrak{W} = \{(\rt,\varpi_2,\psi):\rt\in C^1_{loc.},\psi\in C^0\Hu^1,\varpi_2\in W^{1,1}_{loc},\quad \rt_{uv},\rt_{uu},\psi_u, (\varpi_2)_u\in C^0_{loc.}\}
	\end{equation}	
	that satisfies \eqref{REKG1}-\eqref{REKG4} in a weak sense. That is equations \eqref{REKG1} and \eqref{REKG3} hold classically, \eqref{REKG2} holds almost everywhere and \eqref{REKG4} holds weakly.
\end{defn}

\begin{rem}
	We note at this point that if we have a weak solution to the renormalised Einstein--Klein-Gordon system system we necessarily have that $\rt_{uuv}, \Omega,\Omega_u \in C^0_{loc.}$. 
\end{rem}

\begin{lem}
	If we have a weak solution to the renormalised Einstein--Klein-Gordon system  then the equations \eqref{EKG1}-\eqref{EKG5} hold weakly, hence we can say the metric \eqref{met} solves \eqref{EKG} weakly and we have a $C^0$ metric. 
\end{lem}

\begin{thm}\label{WP}
	Fix $0<\kappa<\f{2}{3}$, let $(\overline{\rt},\overline{\psi})$ be a free data set on $\mc{N}=(u_0,u_1]$, and fix Neumann or Dirichlet boundary conditions. Then there exists a $\delta>0$ such that the following holds. There exists a weak solution $(\rt,\varpi_2,\psi)\in \mathfrak{W}$ of the renormalised Einstein Klein-Gordon equations in the triangle $\Delta_{\delta,u_0}$, such that
	\begin{itemize}
		\item $\rt \to 0 \text{ as $\mi$ is approached},$
		\item $\psi$ satisfies the boundary condition weakly,
		\item The functions $\psi$ and $\rt$ agree with $\overline{\psi}$ and $\overline{\rt}$ respectively when restricted to $v=u_0$.
	\end{itemize}
\end{thm}
\begin{proof}
	The proof of this theorem follows from \cite{holzegel_einstein-klein-gordon-ads_2013}. There is a minor difference in that there is a slight change in subleading terms compared to the spherical case, but these terms can be treated as small error terms. The proof follows a Banach fixed point theorem argument, constructing a map whose fixed point is a solution to \eqref{REKG1}-\eqref{REKG4} and then establishing that it is a contraction map when acting on a sufficiently small ball in the space $\mathfrak{W}$.
\end{proof}
\begin{rem}
	If more regularity is assumed on the initial data then just as in \cite{holzegel_einstein-klein-gordon-ads_2013} it may be shown that we have a classical solution. The boundary conditions hold classically, and $\mc{T}\psi:=  -\f{r_u}{\Omega^2}\pa_v\psi+ \f{r_v}{\Omega^2}\pa_u\psi$ decays like $\rho^{\f{3}{2}-\kappa}$, as the boundary is approached. 
\end{rem}
\begin{rem}
	As the field $\psi\in C^0\Hu^1$, it obeys the usual energy estimates. These may be shown by working at a higher level of regularity, and recovered by a density argument (similar to proposition 8.1 in \cite{holzegel_einstein-klein-gordon-ads_2013}). In later sections when deriving these energy estimates, we will see boundary terms that won't make sense at the current level of regularity. We can however see they vanish at a higher level of regularity, and thus may be dropped from the estimates.
\end{rem}

\subsection{Geometric uniqueness}\label{GU}~\\ 
When solving system \eqref{REKG1} - \eqref{REKG5} it is important to note that we have made a choice  of gauge with respect to which the boundary conditions have been stated. Thus \textit{a priori} we might expect our solution to be dependant on this gauge. In this section we shall restate the boundary conditions in a geometric fashion, and establish the local uniqueness, up to diffeomorphism, of toroidal solutions to the Einstein--Klein-Gordon system for these boundary conditions. This is the problem of geometric uniqueness as discussed in \cite{friedrich_initial_2009}.

In order to address this issue we note that the relevant boundary conditions are invariant under a change of null coordinates.
\begin{defn}
	For $\mc{U}\subset \mathbb{R}^{1+1}$ we define the norm
	\begin{equation}
	\begin{split}
	\norm{\psi}{\Hu_g^1(\mc{U})}^2 := \int_\mc{U}\left(\rt^{-2}\left(\abs{\grt\psi}^2 \right) + \rt^{-2}\psi^2  \right)dudv.
	\end{split}
	\end{equation}
	We can define $\Hu^1_{0,g}(\mc{U})$ as the completion of $C_c^\infty\left(\mc{U}\right)$ with the $\Hu^1_{g}$-norm.  
\end{defn}
\begin{figure}
	\begin{center}	
		\includegraphics[scale=1]{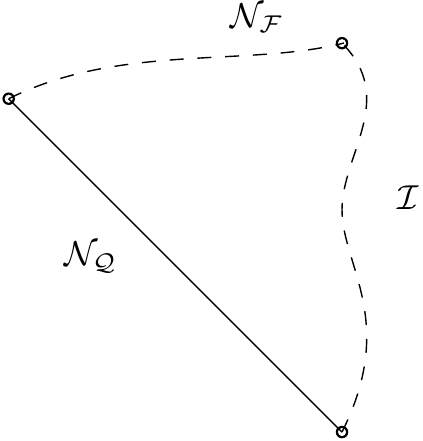}
		\caption{Diagram of a development}
	\end{center}
\end{figure}
\begin{defn}\label{C3GUME}
	Let $\mc{N}$ be an interval of the form $\mc{N} = (u_0,u_1]$. Given an initial data set $(\overline{\rt},\overline{\psi})$ satisfying the initial conditions in section \ref{IaBd} we say a development $\mathfrak{D}$ is a triple $(\mc{M},g,\psi)$ such that $(\mc{M},g)$ is a smooth manifold with $C^0$ Lorentzian metric, $\psi$ is $C^0\Hu^1$ function on $\mc{M}$, and the following hold:
	\begin{itemize}
		\item $(\mc{M},g,\psi)$ is a square flat toroidally symmetric weak solution to the EKG system, with area radius $r$ being a $C^1$ function with $r>0$. 
		\item The quotient manifold $\mc{Q}= \mf / \mathfrak{T}^2$ with its induced Lorentzian metric is a manifold with boundary $\mc{N}_\mc{Q}$ which is a null ray, diffeomorphic to a subset $\mc{N}$ of the form $(u_0,u_0+\epsilon)$, for some $\epsilon>0$. If $\varphi$ is such a diffeomorphism: $\varphi:\mc{N}_\mc{Q}\to(u_0,u_0+\epsilon)$, then\footnote{We abuse notation by using the same letter for objects on $\mathfrak{D}$ and $\mc{Q}$} $\psi\circ\varphi= \overline{\psi}|_{(u_0,u_0+\epsilon)}$ and $\rt\circ\varphi = \overline{\rt}|_{(u_0+u_0+\epsilon)}$.
		\item $\mc{Q}$ admits a a system of global bounded null coordinates, and may be embedded conformally into a subset of $\mathbb{R}^{1+1}$. The boundary of $\mc{Q}$ with respect to the topology of $\mathbb{R}^{1+1}$ is composed of a future boundary $\mc{N}_{\mc{F}}$, a past boundary which coincides with $\mc{N}_\mc{Q}$ and a $C^1$ time-like boundary $\mc{I}$ given by the level set $\mc{I} = \{(u,v)\in\mathbb{R}^{1+1}: \rt(u,v)=0\}$. The metric $g$ is asymptotically AdS, with conformal infinity $\mi$.
		\item `Global hyperbolicity' holds in the sense that all past directed inextendible causal curves in $\mc{Q}$ either intersect $\mc{N}_\mc{Q}$, or have a limit point on $\mi$.
		\item The field $\psi$ satisfies the following integrability conditions:\\
		For each constant $v$ ray, $R_v\subset\mc{Q}$ we have:  
		\begin{equation}
		\int_{R_v} \f{r^4}{\abs{r_u}}(\grt_u\psi)^2	+\f{\abs{r_u}}{r}\psi^2d\bu < \infty,
		\end{equation}
		and each constant $u$ ray, $R_u\subset\mc{Q}$ we have:  
		\begin{equation}
		\int_{R_u} \abs{\f{ r^2 r_u}{\Omega^2}}(\grt_v\psi)^2+ \f{ \abs{r_v} }{r} \psi^2 d\bv < \infty.
		\end{equation}
	\end{itemize}  
\end{defn}

Note that as a consequence of the definition, $\psi$ satisfies the Klein-Gordon equation and boundary conditions in the following weak sense:
				\begin{itemize}
			\item \emph{Dirichlet:} we have
			\begin{equation}
			\int_{\mc{Q}}\left(  \grt^\mu\phi\grt_\mu\psi -V\psi\phi\right) r^2 \Omega^2 du dv= 0,
			\end{equation}
			 for all $\phi \in \Hu^1_{0,g}\left(\mc{Q} \right) $, and moreover $\rt^{-\f{3}{2}+\kappa}\psi=0$ on $\mi$ in a trace sense.
			\item \emph{Neumann:} we have 
			\begin{equation}
			\int_{\mc{Q}}\left(  \grt^\mu\phi\grt_\mu\psi -V\psi\phi\right) r^2 \Omega^2 du dv = 0,
			\end{equation}
			for all $\phi \in \Hu^1_{g}\left(\mc{Q} \right)$  with $\phi|_{\mc{N}_\mc{Q}}=\phi|_{\mc{N}_\mc{F}}=0$ in a trace sense.
		\end{itemize}
		Here $V$ is as defined in Lemma \ref{RenLem}. Stated in this form, the weak formulations (and hence the boundary conditions for $\psi$) are invariant under a change of null coordinates. As a consequence of the fact that the boundary conditions are invariant under a change of $u,v$ coordinates, we have:

\begin{lem}\label{MDl}
	Let $(\overline{\rt},\overline{\psi})$ be a data set satisfying the conditions of section \ref{IaBd}. Let $(\mc{M}_i,g_i,\psi_i)$ be two developments of $(\overline{\rt},\overline{\psi})$ with the same (Dirichlet or Neumann) boundary conditions. Then both $(\mc{M}_i,g_i,\psi_i)$ are extensions of a common development. 
\end{lem}
\begin{proof}
By a change of $u, v$ coordinates, near $\mi$ we can reduce to the local problem studied in the well posedness section above, for which we have already established a uniqueness result. 
\end{proof}

In the usual way, we can extend this local result to a global one by constructing a partial ordering on the set of developments.
\begin{thm}
	A data set $(\rt,\psi)$ satisfying the initial and boundary conditions of section \ref{IaBd} admits a maximal development. This development is unique up to isometry.
\end{thm}
\begin{proof}
	We can follow the proof of \S8.1 of \cite{holzegel_self-gravitating_2012}. There are minor modifications required for the $\Hu^1$ setting, but no significant changes.
\end{proof}

\subsection{The final renormalised Hawking mass}

In the discussion above, it suffices to consider the renormalised Hawking mass $\varpi_2$, which has been modified to be finite at $\mi$. In order to discuss the stability of the toroidal AdS Schwarzschild spacetime we will require a further modification which we will later see has some very useful monotonicity properties. We shall introduce this here in order to state the following extension principles, but it will mainly be of use in later sections.

\begin{defn}
	The final renormalised Hawking mass is
	\begin{equation}
	\varpi = \f{2r_ur_vr}{\Omega ^2}e^{4\pi g \psi^2}+\f{r^3}{2l^2},
	\end{equation}
\end{defn}
Note that $\varpi$ is clearly invariant under a change of $u, v$-coordinate.
\begin{lem}\label{HME}
	$\varpi$ satisfies the following differential equations
	\begin{equation}
	\begin{split}
	&\pa_u \varpi =  -\f{8\pi r^2 r_v}{\Omega^2}(\grt_u\psi)^2e^{4\pi g\psi^2} + \f{4\pi g^2 r_u }{r}\varpi \psi^2  +\f{r_ur^2}{l^2}f(\psi^2),\\
	&\pa_v \varpi =  -\f{8\pi r^2 r_u}{\Omega^2}(\grt_v\psi)^2e^{4\pi g\psi^2} + \f{4\pi g^2 r_v }{r}\varpi \psi^2  +\f{r_vr^2}{l^2}f(\psi^2),
	\end{split}
	\end{equation}
	where
	\begin{equation}
	f(\psi^2) = e^{4\pi g \psi^2}\left(4\pi a\psi^2-\f{3}{2}  \right)-2\pi g^2\psi^2 + \f{3}{2}.
	\end{equation}
\end{lem}
\begin{proof}
	We begin by directly studying the $u$ derivative of $\varpi$
	\begin{equation}
	\pa_u \varpi = \pa_u\left(\f{2r_ur_v r}{\Omega^2}\right)e^{4\pi g\psi^2} + \f{16g\pi r_ur_vr}{\Omega^2}\psi\psi_ue^{4\pi g\psi^2} +\f{3}{2l^2}r^2r_u,
	\end{equation}
	equation \eqref{EKG1} and \eqref{EKG3} imply that
	\begin{equation}
	\pa_u\left(\f{2r_ur_v r}{\Omega^2}\right) = -8\pi\f{r^2r_v}{\Omega^2}\psi_u^2 + \f{4\pi r^2a}{l^2}r_u\psi^2 - \f{3}{2l^2}r^2r_u,
	\end{equation}
	so
	\begin{equation}
	\pa_u \varpi = -8\pi\f{r^2r_v}{\Omega^2}\psi_u^2e^{4\pi g\psi^2} + \f{4\pi r^2a}{l^2}r_u\psi^2e^{4\pi g\psi^2} - \f{3}{2l^2}r^2r_ue^{4\pi g\psi^2} +\f{16g\pi r_ur_vr}{\Omega^2}\psi\psi_ue^{4\pi g\psi^2} +\f{3}{2l^2}r^2r_u.
	\end{equation}
	Converting to twisted derivatives we see that
	\begin{equation}
	\pa_u \varpi =  -\f{8\pi r^2 r_v}{\Omega^2}(\grt_u\psi)^2e^{4\pi g\psi^2} + \f{8\pi g^2r_u^2 r_v}{\Omega^2}\psi^2e^{4\pi g\psi^2}+\f{4\pi r^2 a}{l^2}r_u\psi^2e^{4\pi g\psi^2} + \f{3}{2l^2}r^2r_u\left( 1- e^{4\pi g\psi^2}\right). 
	\end{equation}  
	We turn to studying
	\begin{equation}
	F := \f{8\pi g^2r_u^2 r_v}{\Omega^2}\psi^2e^{4\pi g\psi^2}+\f{4\pi r^2 a}{l^2}r_u\psi^2e^{4\pi g\psi^2} + \f{3}{2l^2}r^2r_u\left( 1- e^{4\pi g\psi^2}\right), 
	\end{equation} 
	which we write in terms of $\varpi$ as
	\begin{equation}
	F = \f{4\pi g^2 r_u }{r}\varpi \psi^2 - 2\pi \f{r_ur^2g^2}{l^2}\psi^2 +\f{4\pi r^2 a}{l^2}r_u\psi^2e^{4\pi g\psi^2} + \f{3}{2l^2}r^2r_u\left( 1- e^{4\pi g\psi^2}\right),
	\end{equation}
We now expand this expression and factorise. This allows us to see how the divergent terms cancel.
	\begin{equation}
	\begin{split}
	F =& \f{4\pi g^2 r_u }{r}\varpi \psi^2 + \f{2\pi}{l^2}r_ur^2\left( -g^2+2a-3g\right)\psi^2 + \f{4\pi r_ur^2a}{l^2}\underbrace{\left(e^{4\pi g\psi^2}-1 \right)}_{\sim \psi^2}\psi^2\\&+\f{3}{2l^2}r^2r_u\underbrace{\left(1+4\pi g\psi^2 - e^{4\pi g\psi^2} \right)}_{\sim \psi^4},  
	\end{split}
	\end{equation}
	recalling the relation
	\begin{equation}
	-g^2+2a-3g = 0,
	\end{equation}
	we see the divergent terms are no longer present and
	\begin{equation}
	F = \f{4\pi g^2 r_u }{r}\varpi \psi^2  + \f{4\pi r_ur^2a}{l^2}\left(e^{4\pi g\psi^2}-1 \right)\psi^2+\f{3}{2l^2}r^2r_u\left(1+4\pi g\psi^2 - e^{4\pi g\psi^2} \right),  
	\end{equation}
	which factorises to
	\begin{equation}
	F = \f{4\pi g^2 r_u }{r}\varpi \psi^2  +\f{r_ur^2}{l^2}\left(e^{4\pi g \psi^2}\left(4\pi a\psi^2-\f{3}{2}  \right)-2\pi g^2\psi^2 + \f{3}{2}\right),
	\end{equation}
	thus proving the result for the $\pa_u\varpi$ equation.
	By the symmetry of the equations the $\pa_v\varpi$ result is analogous.
\end{proof}
\begin{lem} \label{C3HMCL}
	$\varpi$ is constant along $\mi$. 
\end{lem}
\begin{proof}
Assuming sufficient regularity, this follows by computing $T\varpi = (\pa_u + \pa_v)\varpi$ using the equations above and taking the limit $u \to v$. At the $\Hu^1$-level, the result can be established by making use of energy estimates for $\psi$, see \cite{Dunn_Thesis_2018}.
\end{proof}
\section{Extension principles}\label{EP}
In this section we now wish to control aspects of the maximal development's geometry, in particular how singularities may form. For this we need two extension principles for the spacetime. 
\subsection{Interior extension principle}
\begin{prop}
	Let $(\mc{Q}^+ \times \mathfrak{T}^2, g, \psi)$ denote the maximal $\mathfrak{W}$ extension of an asymptotically AdS initial data set for the system \eqref{EKG1}-\eqref{EKG5}. Suppose $p=(U,V)\in \overline{\mc{Q}^+}$. Suppose that
	\begin{itemize}
		\item the set:
		\begin{equation}
		\mc{C}= [U',U] \times [V',V] \backslash \{p\}\subset{Q^+}
		\end{equation}
		has finite spacetime volume for some $(U', V')\neq (U, V)$, and
		\item there exist constants 
		\begin{equation}		
		0<r_0\le r\le R < \infty, \quad \text{for all $(u,v)\in \mc{C}$,}
		\end{equation} 
	\end{itemize}
	then $p\in \mc{Q}^+.$
\end{prop}
\begin{proof}
	The proof of this is similar to \cite{kommemi_global_2013} and \cite{holzegel_self-gravitating_2012}. The key difference is that we are working with a slightly lower level of regularity. In particular the function $\psi_v$ belongs to $L^2$, but is not necessarily continuous. This means that standard contraction map argument cannot be performed in $C^k$ spaces. The extension principle can nevertheless be established in a similar manner but exploiting the absolute continuity of $\psi$. The proofs for this have been omitted but can be found in the appendix of \cite{Dunn_Thesis_2018}.
	\begin{coro}\label{iep}
		For a free data set that contains a (marginally) trapped surface (that is a point on $\mc{N}$ such that $\overline{r_v}\le 0$), the quotient of the maximal development of the initial data set contains a subset as shown in the Penrose diagram:\\
		\begin{center}
			\includegraphics[scale=0.7]{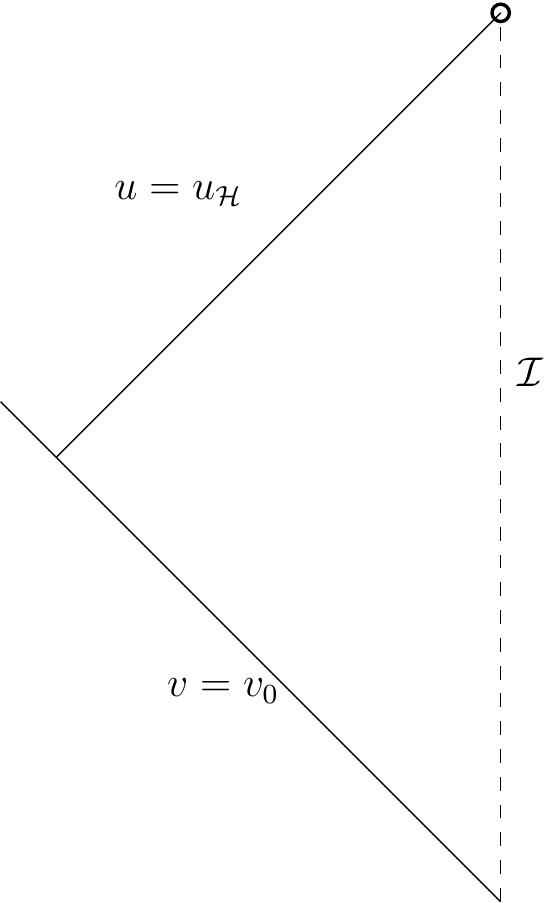}
		\end{center}
		Here $u_\mh$ is the boundary of $u$-constant null rays on which $r\to \infty$ on. Furthermore this set belongs to $\mc{Q}^+$. 
	\end{coro}
\end{proof}
\begin{proof}
	The local wellposedness gives us a solution in a small triangle where $r\to \infty$ along a constant $u$ ray. Now the initial data contains points where $r_v\le 0$. From the Raychaudhuri equation \eqref{EKG2} we see that this inequality propagates in $v$, hence there are points in the spacetime that cannot reach $\mc{I}$. As we have a solution in a small triangle in which $r\to \infty$ along any $u=const$ ray, there exists some $u_\mh$ such that for $u<u_\mh$, $r\not\to \infty$ along these rays. Finally we can see that that the ray $u=u_\mh$ is regular, as $r$ is monotonic, the extension principle (for finite $v$) forbids singularities along it.
\end{proof}
\subsection{Extension principle near $\mi$}
\begin{thm}\label{ENPI}
	For a $\mathfrak{W}$ solution to \eqref{EKG}-\eqref{EKG5} in a triangular region $\Delta_{d,u_0}$, assume that we have:
	\begin{itemize}
		\item The corner condition:
		\begin{equation}
		\lim_{v\to u_0+d}\rt(u_0+d,v)=0.
		\end{equation}
		\item
		For any constant $v$-ray $\mc{N}(v)$ contained in $\Delta_{d,u_0}$ and intersecting $\mi$, there exists a constant $K>0$
		\begin{equation}\label{cond1}
		\int_{u_\mi}^{u} \f{r^4}{-r_u}(\grt_u\psi)^2(\bar{u},v)	-r_u\psi^2(\bar{u},v)d\bu  +\sup_{\mc{N}(v)}\abs{{\f{r^\f{5}{2}}{-r_u}\grt_u\psi}}+\sup_{\mc{N}(v)}\abs{r^{\f{3}{2}-\kappa}\psi}+\sup_{\mc{N}(v)}\abs{\varpi-M}^{\hf} <K.
		\end{equation} 
		(These are all geometric quantities which will be used to form an initial data set on the ray  $\mc{N}(v)$.)
		\item The following bound:
		\begin{equation}\label{cond2}
		\min\left( \inf_{\Delta_{d,u_0}}\abs{\f{1}{l^2}-\f{2\varpi}{r^3}},\inf_{\Delta_{d,u_0}}\abs{\f{1}{l^2}-\f{2\varpi_1}{r^3}}\right)  > c.
		\end{equation}
		(This condition is to prevent potential degenerations away from null infinity.)
		\item There exists a constant $r_{min} >0$ such that in $\Delta_{d,u_0}$
		\begin{equation}
		r\ge r_{min}.
		\end{equation}
		(This will be used in the estimation of $\rt$.)
	\end{itemize}
	Then there exists a $\delta^*>0$ such that the solution can be extended to the set $\Delta_{d+\delta^*, u_0}$.
\end{thm} 
	\begin{lem}\label{iepl1}
	In $\Delta_{d,u_0}$ we have the following estimates
	\begin{equation}
	\rt(u,v)>0\hspace{5pt}, \rt_u(u,v)>0, \hspace{5pt} \lim_{u\to v}\abs{\rt_u-\f{1}{2l^2}} = 0, \text{ and } 	\lim_{u\to v} \rt(u,v)=0.
	\end{equation}
\end{lem}
\begin{proof}
	We switch to a double null coordinate system where $\rt_u= - \rt_v = \frac{1}{2l^2}$ and integrate \eqref{REKG3}.
\end{proof}
	\begin{lem}\label{iepl2}
	In $\Delta_{d,u_0}$ we have the following estimates
	\begin{equation}
	\lim_{u\to v}\rt_{uu}(u,v)=0.
	\end{equation}
\end{lem}
	\begin{proof}
	The idea is to integrate the quantity $\rt_{uuv}$ in $v$, this however requires control over $\rt_v, \Omega^2, \varpi_2, \pa_u\varpi_2,\rt_{uv}$ and $\varpi_1$ in the region $\Delta_{d,u_0}$. These are established through integrating, and differentiating the system. For details, see \cite{Dunn_Thesis_2018}.
\end{proof}
	\begin{lem}\label{twsiteq}
	There exists a constant $C_{g,M,l}>0$ such that 
	\begin{equation}
	\f{1}{C_{g,M,l}}\left( \left(\hat{\gr}_u\psi \right)^2+\psi^2 \right)\rho^{-2}\le \left(\left(\grt_u\psi \right)^2+\psi^2  \right)r^2\le C_{g,M,l}\left( \left(\hat{\gr}_u\psi \right)^2+\psi^2 \right)\rho^{-2}.
	\end{equation}
	That is twisting with $r$ and $\rho$ are equivalent in $H^1$ type norms.
\end{lem}
\begin{proof}
	This follows similarly to Lemma 5.2. in \cite{holzegel_einstein-klein-gordon-ads_2013}.
\end{proof}
\begin{proof}[Proof of Theorem \ref{ENPI}]~\\
From the interior wellposedness results we can extend to the set $\Delta_{d+\tilde{\delta},u_0}\cap\{v\le u_0+d+\tilde{\delta}-\epsilon\}$, for some $\tilde{\delta}>0$, which depends on $\epsilon$ from the continuity. We now extend to a triangle $\Delta_{d+\delta^*,u_0}$. We note that from the previous lemmas we have on each $v=const$ ray in $\Delta_{d,u_0}$ that the function $\rt$ restricted to this space is admissible as part of an initial data set. Lemma \ref{iepl1}, and corollary \ref{twsiteq} show us that $\psi$ restricted to the ray is also admissible as part of an initial data set. Now let $\delta$ be the time of existence of a solution using this data set, but with $K$ replaced by $2K$ and $c$ by $\f{c}{2}$. Now by choosing the ray $v_c = u_0+d-\f{\delta}{2}$. By the above argument (and using continuity), we can extend our solution to the ray $(u_0+d-\f{\delta}{2}+\delta^*]\times \{u_0+d-\f{\delta}{2}\}$ for some $\delta^*<\f{\delta}{2}$ such that the conditions \eqref{cond1}, and \eqref{cond2} hold on 
$(u_0+d-\f{\delta}{2}+\delta^*]\times \{u_0+d-\f{\delta}{2}\}$ with $K$ replaced by $2K$ and $c$ by $\f{c}{2}$. We then apply the local existence result to extend the solution to $\Delta_{d+\delta^*,u_0}.$
\end{proof}
\section{Perturbed toroidal AdS Schwarzschild data and maximal development}\label{IDS}
In view of (a generalisation of) Birkhoff's theorem we know that if we choose $\psi=0$ then our solution will be an isometric subset of the toroidal Schwarzschild-AdS solution. We thus choose initial data for $\psi$ that is quantifiably close to $0$, i.e.\ represents a small perturbation of the black hole. Under this smallness assumption we will then prove various estimates about the derived quantities on the initial data ray that we will need in the evolution.
\subsection{Initial data}
\subsubsection{The free data}
Let $b>0$, we will later take this to be a sufficiently small quantity.
\begin{defn}
	Let $\mc{N} = (u_0,u_1]\times\{v_0\}$.
	We define our initial radial function
	\begin{equation}
	\overline{\rt}(u) = \f{u-u_0}{2l^2},
	\end{equation} 
	consistent with the gauge choice of \eqref{GC}.
	The free data consists of a $C^1\left( \mc{N}\right) $ function $\psi$ such that
	\begin{equation}
	\left( \int_{\mc{N}}\left( \left(\overline{\grt}_u\overline{\psi}\right)^2+\overline{\psi}^2 \right)\overline{r}^{2}du\right)^\hf+ 			\sup_{\mc{N}}\abs{\overline{\psi}\overline{r}^{\f{3}{2}-\kappa}}+\sup_{\mc{N}}\abs{{r^{\hf+\f{s}{2}}}\left( \overline{\grt}_u\overline{\psi}\right) }=:  b^2,
	\end{equation}
	where $0<s<1$. (The choice of $s$ is technical and we only expect to see $r^{-\hf}$ decay of the $u$ derivative propagating in the system, we do however need this initial smallness in the problem in order to prove various results about the spacetime). 
\end{defn}    
\subsubsection{Deduced quantities}
From $\overline{\psi}$ and $\overline{\rt}$, we can define the following derived quantities for our system:\\ \newline
\textbf{Initial Hawking mass}:\\ \newline
We define $\overline{\varpi}$ as the unique $C^1({\mc{N})}$ solution to
\begin{equation}\label{varpiID}
\begin{split}
\pa_u \overline{\varpi} = &\f{2\pi\overline{r}^2}{\overline{r}_u}\left(\f{\overline{r}^2}{l^2} - \f{2\overline{\varpi}}{\overline{r}} \right) (\overline{\grt}_u\overline{\psi})^2 + \f{4\pi g^2 \overline{r}_u }{\overline{r}}\overline{\varpi} \overline{\psi}^2  \\&+\f{\overline{r}_u\overline{r}^2}{l^2}\left( e^{4\pi g \overline{\psi}^2}\left(4\pi a\overline{\psi}^2-\f{3}{2}  \right)+2\pi(3g-2a)\overline{\psi}^2 + \f{3}{2}\right), 
\end{split}
\end{equation}
with boundary condition 
\begin{equation}
\lim_{u\to u_0} \overline{\varpi} = M.
\end{equation}
From this we then define
\begin{equation}
\overline{\varpi_2} = \overline{\varpi} e^{-4\pi g\overline{\psi}^2}-\f{\overline{r}^3}{2l^2}\left( e^{-4\pi g\overline{\psi}^2}+4\pi g\overline{\psi}^2 -1\right).
\end{equation}
(Recall that the local well-posedness is stated in terms of $\overline{\varpi_2}$, but in practice we shall prefer to use $\overline{\varpi}$.)
\\ \newline
\textbf{The quantity $\overline{r_v}$}:\\ \newline
Recall equation \eqref{EKG1} holds classically. Defining the variable 
\begin{equation}
\chi:= -\f{\Omega^2}{4r_u},
\end{equation}
we can rewrite \eqref{EKG1} as
\begin{equation}
\pa_u\log\chi = \f{4\pi r}{r_u}\left( \gr_u\psi\right)^2.
\end{equation}
Solving this ODE and using the definition of the Hawking mass, one gets the following expression for $r_v$
\begin{equation}
r_v=\chi|_\mi \cdot\left(\f{-2\varpi}{r}+\f{r^2}{l^2} \right)e^{4\pi g\psi^2}\exp\left(\int_{u_0}^{u}\f{4\pi r}{r_u}(\gr_u\psi)^2du\right),  
\end{equation}
we will later on make a gauge choice where $\chi|_\mi = \hf$. We choose
\begin{equation}
\overline{r_v}=\hf\left(\f{-2\overline{\varpi}}{\overline{r}}+\f{\overline{r}^2}{l^2} \right)e^{4\pi g\overline{\psi}^2}\exp\left(\int_{u_0}^{u}\f{4\pi \overline{r}}{\overline{r}_u}(\gr_u\overline{\psi})^2du\right),  
\end{equation}
we remark that while we have not used a twisted derivative in the definition of $\overline{r_v}$, the initial data choices allow us to see that is indeed an integrable quantity. It is also easy to see that $\overline{r_v}$ is independent of choice of $u$-coordinate on the data. \\ \newline
\textbf{The quantity $\Omega^2$}: \\ \newline
We finally define the $C^1(\mc{N})$ quantity 
\begin{equation}
\overline{\Omega}^2=- \f{4\overline{r}_u\overline{r_v}}{\left( \f{-2\overline{\varpi}}{\overline{r}}+\f{\overline{r}^2}{l^2}\right) }e^{4\pi g\overline{\psi}^2}.
\end{equation}
\subsection{Consequences of the smallness}~\\ \\
Define the regular and marginally trapped region of the initial data ray to be 
\begin{equation}
\sr_{v_0}\cup\mathcal{A}_{v_0} = \mc{N}\cap\{u\in \mc{N}: \overline{r_v}(u)>0\}\cup\{u\in \mc{N}: \overline{r_v}(u)=0\}.
\end{equation}
\begin{lem}\label{COS}
	We have that for $b>0$ sufficiently small, on $\sr_{v_0}\cup\mathcal{A}_{v_0} $, 
	\begin{equation}\label{HMIV}
	\sup_{u\in\sr_{v_0}\cup\mathcal{A}_{v_0} }\abs{\overline{\varpi} - M} \le C_{l,g}b^2,
	\end{equation}
	note that in particular if $b$ is small enough $\varpi>0$. Further, defining the Schwarzschild value 
	\begin{equation}
	\overline{r_v}^s := \hf\left(-\f{2M}{\overline{r}}+\f{\overline{r}^2}{l^2} \right),
	\end{equation}
	we have
	\begin{equation}\label{rvdataest}
	\sup_{u\in\sr_{v_0}\cup\mathcal{A}_{v_0} }\abs{\overline{r_v}-\overline{r_v}^s} \le C_{l,a,M}b^2.
	\end{equation}
	For small enough initial data there exist points on $\mc{N}$ such that
	\begin{equation}\label{C3MTSE1}
	\overline{r_v}\le 0.
	\end{equation}
	Furthermore there is a unique $u^*\in \mc{N}$ such that 
	\begin{equation}\label{C3MTSE2}
	\overline{r_v}(u^*)=0.
	\end{equation}  
	Defining $r_{min} := \overline{r}(u^*)$
	we have
	\begin{equation} \label{rmine}
	\abs{r_{min}-r_+ } \le  C(b),
	\end{equation}
	where $C(b) \to 0$ as $b\to 0$, and we recall $r_+ := (2Ml^2)^{\f{1}{3}}$.
\end{lem} 
\begin{proof}
	Define a bootstrap region to be
	\begin{equation}
	\mc{B}_{v_0}:=\sr_{v_0}\cup\mathcal{A}_{v_0} \cap \left\{r\ge \f{r_+}{2}\right\},
	\end{equation}
	Clearly this set is closed, non empty, and connected. We need to show it is open to complete the bootstrap argument.\\
	We define
	\begin{equation}
	f\left( \overline{\psi}^2\right)  = \left( e^{4\pi g \overline{\psi}^2}\left(4\pi a\overline{\psi}^2-\f{3}{2}  \right)+2\pi(3g-2a)\overline{\psi}^2 + \f{3}{2}\right).
	\end{equation}
	The equation for $\overline{\varpi}$ \eqref{varpiID} is thus
	\begin{equation}
	\begin{split}
	\pa_u \overline{\varpi} &= \f{2\pi\overline{r}^2}{\overline{r}_u}\left(\f{\overline{r}^2}{l^2} - \f{2\overline{\varpi}}{\overline{r}} \right) (\overline{\grt}_u\overline{\psi})^2 + \f{4\pi g^2 \overline{r}_u }{\overline{r}}\overline{\varpi} \overline{\psi}^2  +\f{\overline{r}_u\overline{r}^2}{l^2}f(\overline{\psi}^2)\\
	& = \varpi\underbrace{\left(-\f{4\pi \overline{r}}{\overline{r}_u}\left(\overline{\grt}_u\overline{\psi} \right)^2 + \f{4\pi g \overline{r}_u}{\overline{r}}\overline{\psi}^2  \right)}_{=:h} +\f{2\pi \overline{r}^4}{l^2\overline{r}_u}\left(\overline{\grt}_u\overline{\psi} \right)^2 +\f{\overline{r}_u\overline{r}^2}{l^2}f(\overline{\psi}^2)\\
	&= \overline{\varpi} h -4\pi \overline{r}\left(\overline{r}\overline{\grt}_u\overline{\psi} \right)^2 +\f{\overline{r}_u\overline{r}^2}{l^2}f(\overline{\psi}^2).
	\end{split}
	\end{equation}
	Solving this equation
	%	\begin{equation}
	%	\pa_u\left( \overline{\varpi} \exp\left(-\int_{u_0}^u h du \right)  \right) = -4\pi \overline{r}\left(\overline{r}^\hf\overline{\grt}_u\overline{\psi} \right)^2\exp\left(-\int_{u_0}^u h du \right) +  \f{\overline{r}_u\overline{r}^2}{l^2}f(\overline{\psi}^2)\exp\left(-\int_{u_0}^u h du \right),
	%	\end{equation}
	%	and thus
	\begin{equation}
	\begin{split}
	&\overline{\varpi}   - M\exp\left(\int_{u_0}^u h du\right) \\= &\exp\left(\int_{u_0}^u h du\right) \int_{u_0}^u -4\pi \overline{r}\left(\overline{r}^\hf\overline{\grt}_u\overline{\psi} \right)^2\exp\left(-\int_{u_0}^u h du \right)  +  \f{\overline{r}_u\overline{r}^2}{l^2}f(\overline{\psi}^2)\exp\left(-\int_{u_0}^u h du \right)du. 
	\end{split}
	\end{equation}
	We estimate $h$ by
	\begin{equation}
	\begin{split}
	\abs{h}	&\le C_{l,g}b^2\left(\overline{r}^{-{4+s}}+\overline{r}^{-4+2\kappa} \right)(-\overline{r}_u) .
	\end{split}
	\end{equation}
	We see that in $\mc{B}_{v_0}$
	\begin{equation}
	\int_{u_0}^u \abs{h}du \le C_{l,g}\f{b^2}{r_{+}}.
	\end{equation}
	\begin{equation}
	\begin{split}
	\overline{\varpi} - M 
	\ge \exp\left(\int_{u_0}^u h du\right) \int_{u_0}^u -4\pi \overline{r}^{1-s}b^2\exp\left(C_{l,g}\f{b^2}{r_{+}}\right)  +  \f{\overline{r}_u\overline{r}^2}{l^2}f(\overline{\psi}^2)\exp\left(-\int_{u_0}^u h du \right) du.
	\end{split}
	\end{equation}
	(Note $\overline{r}^{1-s}du \sim\overline{r}^{-1-s}d\bar{r}$  hence the need for our initial data to have the additional $s$ smallness). Continuing the estimation 
	\begin{equation}
	\begin{split}
	\overline{\varpi} - M 	\ge& -C_{l,g}\exp\left( C_{l,g}\f{b^2}{r_{+}}\right) \left( \f{r_+}{2}\right) ^{-s}b^2+ \exp\left(\int_{u_0}^u h du\right) \int_{u_0}^u\f{\overline{r}_u\overline{r}^2}{l^2}f(\overline{\psi}^2)\exp\left(-\int_{u_0}^u h du \right) du.
	\end{split}
	\end{equation}
	For $b^2$ small
	\begin{equation}
	f(\overline{\psi}^2) \le 8\pi a g \overline{\psi}^4,
	\end{equation}
	so 
	\begin{equation}
	\begin{split}
	\exp\left(\int_{u_0}^u h du\right) \int_{u_0}^u\f{\overline{r}_u\overline{r}^2}{l^2}f(\overline{\psi}^2)\exp\left(-\int_{u_0}^u h du \right) du &\ge -C_{g,l}\exp\left( C_{l,g}\f{b^2}{r_{+}}\right)b^4\int_{\infty}^{\overline{r}}r^{-4+4\kappa} d\overline{r}\\
	&\ge  -C_{g,l}\exp\left( C_{l,g}\f{b^2}{r_{+}}\right)b^4r^{-3+4\kappa}_{+}.
	\end{split}
	\end{equation}
	And thus
	\begin{equation}
	\overline{\varpi} - M \ge -C_{l,g}\exp\left( C_{l,g}\f{b^2}{r_{+}}\right) r_{+}^{-s}b^2 -C_{g,l}\exp\left( C_{l,g}\f{b^2}{r_{+}}\right)b^4r^{-3+4\kappa}_{+} \ge -b^2f(b),
	\end{equation}
	where $f(b)$ goes to a positive constant as $b\to 0$. \\
	Estimating the other direction 
	\begin{equation}
	\pa_u \overline{\varpi} \le h\overline{\varpi},
	\end{equation}
	we quickly see
	\begin{equation}
	\begin{split}
	\overline{\varpi} - M 	&\le  \f{C_{M,l,g}}{r_{+}}b^2.
	\end{split}
	\end{equation}
	We conclude that
	\begin{equation}
	\abs{\overline{\varpi}-M} \le f(b) b^2,
	\end{equation}
	where $f(b) \to C>0$ as $b\to 0$.\\
	We now show the second statement \eqref{rvdataest}.
	Recall that 
	\begin{equation}\label{rvid}
	\overline{r_v}=\hf\left(\f{-2\overline{\varpi}}{r}+\f{\overline{r}^2}{l^2} \right)\exp\left(4\pi g\overline{\psi}^2+\int_{u_0}^{u}\f{-8\pi l^2}{\overline{r}}(\gr_u\overline{\psi})^2du\right),  
	\end{equation}
	so
	\begin{equation}
	\begin{split}	
	\abs{\overline{r_v}-\overline{r_v}^s} 
	&\le \f{1}{\overline{r}}\abs{\overline{\varpi}-M}+ \abs{\f{-2\overline{\varpi}}{\overline{r}}+\f{\overline{r}^2}{l^2}}\abs{\left(\exp\left(4\pi g\overline{\psi}^2+\int_{u_0}^{u}\f{-8\pi l^2}{\overline{r}}(\gr_u\overline{\psi})^2du\right)-1 \right)}\\
	&\le \f{C_{l,M,a}}{r}b^2 + \abs{\f{-2\overline{\varpi}}{r}+\f{\overline{r}^2}{l^2}}\abs{\left(\exp\left(4\pi g\overline{\psi}^2+\int_{u_0}^{u}\f{-8\pi l^2}{\overline{r}}(\gr_u\overline{\psi})^2du\right)-1 \right)}.
	\end{split}		
	\end{equation}
	Noting the estimate
	\begin{equation}
	\abs{\gr_u\overline{\psi}}^2 \le C_{l,g}b^2\overline{r}^{-3+2\kappa},
	\end{equation}
	we quickly can see that for $b<1$	
	\begin{equation}
	\abs{\left(\exp\left(4\pi g\overline{\psi}^2+\int_{u_0}^{u}\f{-8\pi l^2}{\overline{r}}(\gr_u\psi)^2du\right)-1 \right)} \le C_{g,l}b^2\overline{r}^{-3+2\kappa}.		
	\end{equation}
	From here it follows that
	\begin{equation}\label{rvidest}
	\abs{\overline{r_v}-\overline{r_v}^s} \le  b^2\left( \f{C_{l,M,a}}{r} + C_{g,l}\overline{r}^{-1+2\kappa}\right).
	\end{equation}
	Restricting to $\kappa \in (0,\hf]$ we can find a constant $C_{M,l,a}>0$, such that
	\begin{equation}
	\abs{\overline{r_v}-\overline{r_v}^s} \le  C_{M,l,g}b^2,
	\end{equation}
	proving \eqref{rvdataest}. \\
	%	Finally if we evaluate estimate \eqref{rvidest} at the value $u^+$ (defined through $\overline{r}(u^+)=r_+$).
	%	\begin{equation}
	%	\abs{r_v(u(r_+))} \le  b^2\left( \f{C_{l,M,a}}{r_+} + C_{g,l}r_+^{-1+2\kappa}\right)
	%	\end{equation}
	For statements \eqref{C3MTSE1} and \eqref{C3MTSE2}, define $u^+$ by $\overline{r}(u^+)=r_+$.
	Now in our coordinate system we have 
	\begin{equation}
	\left( \overline{r_v}^s\right) _u = -\f{M}{l^2}-\f{\overline{r}^3}{l^4} <0,
	\end{equation}
	showing that $\overline{r_v}^s$  is monotone on $\mc{N}$. For $\epsilon>0$ consider $\tilde{u} = u^+ +\epsilon$ so
	\begin{equation}
	\overline{r_v}^s({\tilde{u}}) <0,
	\end{equation}
	coupled with the estimate
	\begin{equation}
	\overline{r_v}^s(\tilde{u}) - Cb^2 \le \overline{r_v}(\tilde{u}) \le Cb^2 +\overline{r_v}^s(\tilde{u}), 
	\end{equation}
	implies for small initial data there is a point on $\mc{N}$ where $\overline{r_v}\le0$. 
	Clearly we can repeat this argument and find a $b^2$ and $u$ value say $u_m$ where $\overline{r_v}(u_m)\ge0$.\\
	From continuity we know that there exists at least one value of  $u\in[u_m,\tilde{u}]$ such that $\overline{r_v}(u)=0$. Viewing the radial equation in terms of $\varpi$
	\begin{equation}
	(\overline{r_v})_u = -\f{\overline{\Omega^2}}{2}\left(\f{\overline{r}}{2l^2}\left(3-e^{4\pi g\overline{\psi}^2} \right) + \f{\overline{\varpi}}{\overline{r}^2}e^{4\pi g\overline{\psi^2}}  \right) + \f{2\pi a}{l^2}\overline{\Omega^2}\overline{\psi}^2 < 0,
	\end{equation}
	for small $b$. We have $\overline{r_v}$ is monotonic, and this zero is unique. Let $u^*$ denote this value, and denote it's $\overline{r}$ value by $r_{min}$.
	Recalling the definition of $\overline{r_v}$ we see that for this value of $u$ the relationship 
	\begin{equation}
	r_{min}^3 = 2l^2\overline{\varpi}({u^*}).
	\end{equation}
	So
	\begin{equation}
	\abs{r_{min}^3- r_+^3} =2l^2\abs{\overline{\varpi}({u^*})-M} \le Cb^2.
	\end{equation}
	This implies the inequality
	\begin{equation}
	r_{min} \ge \left(r_+^3-Cb^2 \right)^{\f{1}{3}}, 
	\end{equation}
	and we deduce that
	\begin{equation}
	r_{min} \ge r_+ - C(b),
	\end{equation}
	where $C(b) \to 0$ as $b\to 0$.
	From here we see that for $b$ chosen small enough, the inequality in $\mathcal{B}_{v_0}$
	\begin{equation}
	r\ge r_{min}\ge r_+ - C(b)  > \f{r_+}{2},
	\end{equation}
	holds. Hence $\mathcal{B}_{v_0}$ is open and  
	\begin{equation}
	\mathscr{R}_{v_0}\cup\mathcal{A}_{v_0} \cap \left\{r\ge \f{r_+}{2}\right\}=\mathscr{R}_{v_0}\cup\mathcal{A}_{v_0} .
	\end{equation}
\end{proof}	
\subsubsection{Maximal development and set up}~\\ \\
We let $\mc{Q}\subset \mathbb{R}^2$ denote the quotient by $\mathfrak{T}^2$ of the maximal development from perturbed toroidal Schwarzschild-AdS data. From the geometric uniqueness statement we know that this is unique up to diffeomorphism. Denote the regular region $\sr =\{(u,v)\in\mc{Q} : r_v(u,v) > 0 \}, $ and let $N(u)\subset\mc{Q}$ denote the outgoing characteristic null-line $u=const$ emanating from the initial data.
\begin{lem}
	In the maximal development we have the following properties
	\begin{itemize}
		\item The set of outgoing null rays along which $r \to \infty$ is non-empty:
		\begin{equation}
		\{u>u_0:N(u)\subset \sr \text{ and } \rt\to 0 \text{ along } N(u)\} \ne \emptyset.
		\end{equation}
		\item For 
		\begin{equation}
		u_\mh := \sup_{u>u_0}\{u:N(u)\subset \sr \text{ and } \rt\to 0 \text{ along } N(u)\},
		\end{equation}
		and 
		\begin{equation}
		\sr_\mh := \sr\cap\{u_0<u<u_\mh\}, \quad \text{ and } \quad \overline{\sr_\mh} := \sr_\mh\cup N(u_\mh), 
		\end{equation}
		we have that
		\begin{equation}\label{C3SF}
		\sup_{\sr_\mh}v = \sup_{N(u_\mh)}v.
		\end{equation}
		\item Define `null infinity' $\mi =\{(u,v_{\infty}(u))|u_0<u<u_\mh\}$, where $v_{\infty}(u)$ is the value of $v$ such that:
		$\lim_{v\to v_{\infty(u)}}\rt(u,v)=0$. We can reparametrise $\mi$ by $\{u_\mi(v),v|v\in \mc{Q}\}$, where $u_\mi(v)$ is the $u$ coordinate of the past limit point at which the $v=const$ ray intersects $\mi$. \\
		Then there exists a double null system $(u,v)$ covering $\sr_\mh$, such that
		\begin{equation}
		\chi|_\mi = \hf, \quad \overline{\rt}_u = \f{1}{2l^2}.
		\end{equation}  
	\end{itemize}
\end{lem}
\begin{proof}
	From continuity the data set contains a point where $r_v<0$, then we simply apply corollary \ref{iep}. As $N(u_\mh)$ is regular \eqref{C3SF} follows from the fact that a first singularity cannot form along it.
	Letting
	\begin{equation}
	\hat{u} = h(u),\quad \hat{v} = g(v),
	\end{equation}
	we see that under these transformations we have that
	\begin{equation}
	\hat{\chi} := \f{\hat{\Omega}^2}{-4r_{\hat{u}}} = \f{\chi}{g'},
	\end{equation}
	and 
	\begin{equation}
	-r_{\hat{u}}=\f{-r_u}{h'}.
	\end{equation}
	So choosing
	\begin{equation}
	g'(v) = \f{\Omega^2}{-2r_u}(v,v) = 2\chi(v,v),
	\end{equation}
	and 
	\begin{equation}
	f'(u) = \f{-2l^2r_u}{r^2}(u,v_0),
	\end{equation}
	we then have that at $\mi$ 
	\begin{equation}
	\hat{\chi} = \hf,
	\end{equation}
	and on the initial data ray $v=v_0$
	\begin{equation}
	-r_{\hat{u}} = \f{r^2}{2l^2}.
	\end{equation}
	The latter choice being consistent with \eqref{GC}.\\
	We now switch to these coordinates and drop the hats.\\
	\begin{figure}[h!]
		\begin{center}
			\includegraphics[scale=0.7]{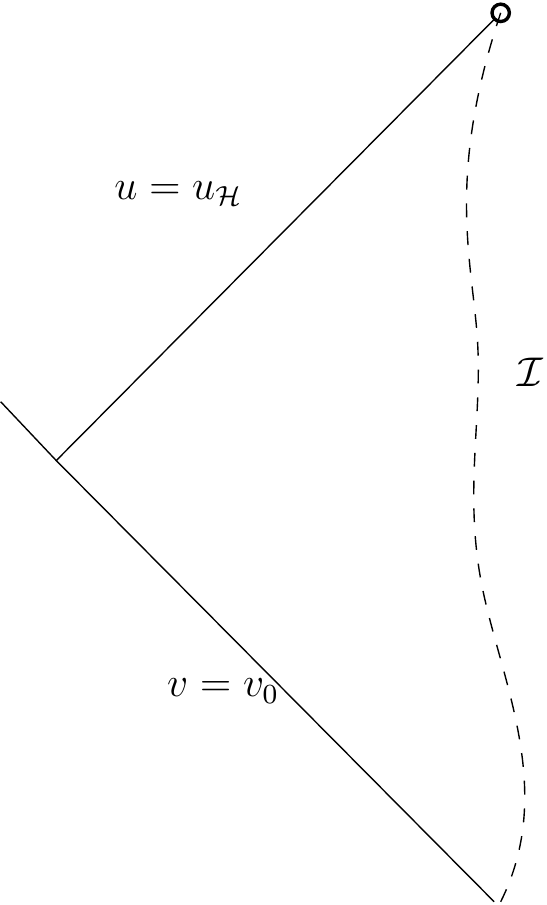}	
		\end{center}
		\caption{Depiction of (a subset of) the Penrose diagram}
	\end{figure}

	Note that in these coordinates null infinity is no longer a straight line. We do not include the future limit point of the ray $u=u_\mh$ in null infinity, as it is a priori possible for $\rt \to 0$ along this ray. We will, as part of the proof of orbital stability, show this not to be the case.   
\end{proof}
\begin{lem}
	In ${\sr}_\mh$ we have that 
	\begin{equation}
	r\ge r_{min},
	\end{equation}
\end{lem}
\begin{proof}
	We can write
	\begin{equation}
	r(u,v) = \overline{r}(u) + \int_{v_0}^v r_v(u,\hat{v})d\hat{v},
	\end{equation}
	as we are in the regular region we know the integral is positive ($r_v>0$), and lower bounds on the initial data ($\overline{r}\ge r_{min}>0$) prove the result.
\end{proof}
\begin{lem}
	We have in $\mc{Q}^+$ that
	\begin{equation}
	r_u <0.
	\end{equation}
\end{lem}
\begin{proof}
	Integrating equation \eqref{EKG1} from $\mi$ yields the inequality
	\begin{equation}
	\f{r_u}{\Omega^2} \le -\hf,
	\end{equation}
	the result follows.
\end{proof}
\textbf{Curves of constant $r_X$ and $r_Y$}\\ \newline
We define $r_X$ to be the solution to
\begin{equation}
-\f{2M}{r_X} + \f{r_X^2}{l^2} = d,
\end{equation}
where $d>0$ has been chosen small enough so that
\begin{equation}\label{C3RX}
\log\f{r_X}{r_{min}} < \f{1}{2\abs{a}}.
\end{equation}
This may be chosen from the continuity of $r$.\\ Now we define $r_Y$ to be the solution to 
\begin{equation}
-\f{2M}{r_Y} + \f{r_Y^2}{l^2} = c,
\end{equation}
with $c<d$. So $r_Y < r_X$.
Our Penrose diagram looks like: \\
\begin{figure}[h!]
	\centering\includegraphics[scale=0.8]{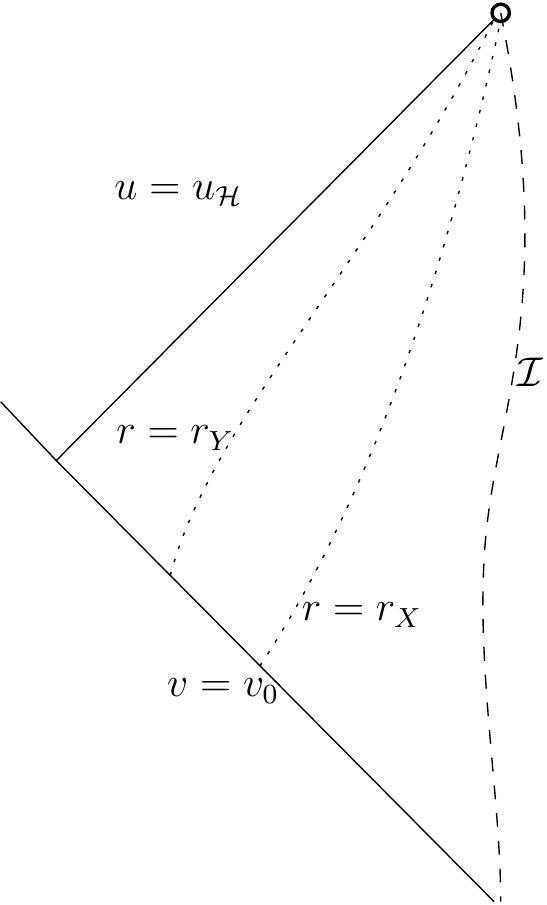}
	\caption{Penrose diagram of spacetime depicting $r=const$ curves.}
\end{figure}
\newpage
\textbf{Geometric norms}\\
\begin{itemize}
	\item The $H^1_{d}$ norm
	\begin{equation}
	\begin{split}
	\norm{\psi}{H^1_{d}}^2(u,v) &=\int_{u_\mi}^{u} \f{r^2 r_v}{\Omega^2}(\gr_u\psi)^2 +\f{(-r_u) }{r} \psi^2d\bu + \int_{v_0}^{v} -\f{ r^2 r_u}{\Omega^2}(\gr_v\psi)^2+ \f{ r_v }{r} \psi^2 d\bv.
	\end{split}
	\end{equation}
	This is the standard degenerate energy norm one expects to see from exploiting the Kodama vector field $\mathcal{T}$ of the system (in context of the energy momentum tensor of the field, see \cite{holzegel_stability_2013} $\S5.2$). It is however not finite for our boundary conditions but will be useful when considering regions of bounded $r$.
	\item The $\Hu^1_{d}$ norm
	\begin{equation}
	\begin{split}
	\norm{\psi}{\Hu^1_{d}}^2(u,v) &=\int_{u_\mi}^{u} \f{r^2 r_v}{\Omega^2}(\grt_u\psi)^2 +\f{(-r_u) }{r} \psi^2d\bu + \int_{v_0}^{v} -\f{ r^2 r_u}{\Omega^2}(\grt_v\psi)^2+ \f{ r_v }{r} \psi^2 d\bv.
	\end{split}
	\end{equation}
	This norm naturally arises from considering the renormalised Hawking mass as an energy potential. It suffers from degeneration on the first order terms at the apparent horizon (where $r_v=0$), and a sub optimal weight on the zeroth order terms. This norm will be the basis of our estimates. 
	\item The $\Hu^1_{1}$ norm   
	\begin{equation}
	\norm{\psi}{\Hu^1_{1}}^2(u,v) =\int_{u_\mi}^{u} \f{r^4}{-r_u}(\grt_u\psi)^2	+\f{(-r_u) }{r} \psi^2d\bu + \int_{v_0}^{v} -\f{ r^2 r_u}{\Omega^2}(\grt_v\psi)^2+ \f{ r_v }{r} \psi^2 d\bv.
	\end{equation}
	After the redshift estimates we will be able to show control of this norm on the spacetime. It does not suffer degeneration at the apparent horizon. In fact it actually controls a zeroth order term with a better weight but this is only clear after exploiting a Hardy type inequality.
	\item The $\Hu^1$ norm
	\begin{equation}
	\norm{\psi}{\Hu^1}^2(u,v) =\int_{u_\mi}^{u} \f{r^4}{-r_u}(\grt_u\psi)^2	-r_u\psi^2d\bu + \int_{v_0}^{v} -\f{ r^2 r_u}{\Omega^2}(\grt_v\psi)^2+ \f{ r_v }{r} \psi^2 d\bv.
	\end{equation}
	After a Hardy inequality we will be able to see this is equivalent to the  $\Hu^1_{1}$ norm. We will however find this form more useful.
	\item  The $\Lu^2$ norm
	\begin{equation}
	\norm{\psi}{\Lu^2}^2(u,v) =\int_{u_\mi}^{u}	-r_u\psi^2d\bu + \int_{v_0}^{v} \f{ r_v }{r} \psi^2 d\bv.
	\end{equation}
	\item For convenience we also define the flux quantity
	\begin{equation}
	\fl(u,v) = \norm{\psi}{\Hu^1}^2(u,v)+\norm{\psi}{\Hu^1}^2(u,v_0).
	\end{equation}
\end{itemize}
It is worth noting that these norms are all invariant under a change of double null coordinates and are thus geometric in their nature.\\ \\
\textbf{Geometric derivatives}~\\ \\
Recall that we defined the Kodama vector field
\begin{equation}
\mc{T} = -\f{r_u}{\Omega^2}\pa_v + \f{r_v}{\Omega^2}\pa_u, 
\end{equation}
we can define its orthogonal compliment 
\begin{equation}
\mc{R} = -\f{r_u}{\Omega^2}\pa_v - \f{r_v}{\Omega^2}\pa_u = (dr)^\sharp. 
\end{equation}
We can also define twisted versions of these operators by
\begin{equation}
\tilde{\mc{T}}\psi = -\f{r_u}{\Omega^2}\grt_v\psi + \f{r_v}{\Omega^2}\grt_u\psi = \mc{T}\psi,
\end{equation}
and
\begin{equation}
\tilde{\mc{R}}\psi = -\f{r_u}{\Omega^2}\grt_v\psi - \f{r_v}{\Omega^2}\grt_u\psi.
\end{equation} 
In terms of this operator, the Neumann boundary condition becomes:
\begin{equation}
\rt^{-\hf-\kappa}\tilde{\mc{R}}\psi =0,  \quad \text{on $\mi$}.
\end{equation}

\section{Orbital stability and completeness of null infinity}\label{OSS}
In this section we prove two key theorems:
\begin{thm}\label{basicests}(Orbital Stability: Basic Estimates)\\
	In  $\mathscr{R}_\mh$ for $\kappa \le \hf$, there exists a constant $b>0$, such that we have the following estimates 
	\begin{equation}\label{OSE}
	\begin{split}
	\abs{2\chi(u,v)-1}^\hf&+\abs{\varpi(u,v)-M}^{\hf}+\abs{\f{r^\f{5}{2}}{-r_u}\grt_u\psi(u,v)}+\abs{\f{r^{\f{5}{2}-\kappa}}{-r_u}\psi_u(u,v)} + \abs{r^{\f{3}{2}-\kappa}\psi(u,v)}\\ &+  \norm{\psi}{\Hu^1}(u,v)\le C_{l,M,\kappa}\left( \norm{\psi}{\Hu^1}(u_{\mh},v_0)+\sup_{I(v_0)}\left(\abs{\f{r^\f{5}{2}}{-r_u}\grt_u\psi} +\abs{r^{\f{3}{2}-\kappa}\psi} \right) \right). 
	\end{split}
	\end{equation}
\end{thm}
\begin{equation}
\end{equation}
\begin{thm}(Completeness of Null Infinity) \\
	Let 
	\begin{equation}
	v_m = \sup_{v\ge v_0}\{v|(u_\mh,v)\in\mc{Q}\}.
	\end{equation}
	Then it is the case that $v_m=\infty$.
\end{thm}
The idea will be to make use of a bootstrap argument, where we bootstrap on the size of the perturbation.

\subsection{Basic estimates}~
\subsubsection{The bootstrap}~\\ 
Let $\tilde{u}\in[u_0,u_\mh]$ and define the region
\begin{equation}
\hat{\mc{B}}(\tilde{u}) = \mc{R}_\mh\cap\{u_0< u<\tilde{u}\}.
\end{equation}
We will bootstrap on the condition
\begin{equation}\label{BS}
\abs{r^{\f{3}{2}-\kappa}\psi}<b.
\end{equation}
We let 
\begin{equation}
u_{max} := \sup_u \left\{u\in \mathscr{R}:\abs{r^{\f{3}{2}-\kappa}\psi}<b \right\},
\end{equation}
so the bootstrap region is then defined as
\begin{equation}
\mc{B} = \hat{\mc{B}}(u_{max}) \subset \sr_\mh. 
\end{equation} 
We aim to prove that $\mc{B} = \sr_\mh = \hat{\mc{B}}(u_{\mh})$. 
It is clear that $\mc{B}$ is open, connected, and non empty. Hence we aim to show that $\mc{B}$ is a closed subset of $\sr_\mh$. For this we assume that $u_{max}<u_\mh$ is fixed, and that in $\overline{\mc{B}}$ we can improve the bound \eqref{BS} (it trivially holds at $u=u_0$ as this is an initial data point).
%\begin{rem}
%	Globally we will require that $0<b<1$, is chosen small enough that the following bounds
%	\begin{equation}
%	\begin{split}
%	&1\ge e^{4\pi g\psi^2} \ge e^{4\pi gb^2r^{-3+2\kappa}} \ge e^{-1},\\
%	&1 \le e^{-4\pi g\psi^2} \le e^{-4\pi gb^2r^{-3+2\kappa}} \le e,
%	\end{split}
%	\end{equation}
%	hold.
%\end{rem}
\subsection{The renormalised Hawking mass}
%\begin{defn}
	Recall that we defined the final renormalised Hawking mass as
	\begin{equation}
	\varpi = \f{2r_ur_vr}{\Omega ^2}e^{4\pi g \psi^2}+\f{r^3}{2l^2}.
	\end{equation}
	The final renormalised Hawking mass provides a potential for the $\Hu^1$ geometric energy, and for small enough $\psi$ it satisfies monotonicity properties. Coupled with a redshift estimate this leads to an energy estimate for $\psi$. From here Sobolev embeddings can be used to recover the bootstrap assumption.  \\ \\
%\end{defn}
 In order to see that $\varpi$ is montonotic we need to get a sign for it first. To see this  we need to study the region where function $f(\psi^2)$ is positive. Plotting the curve shows the following global behaviour.
 \begin{figure}[!h]
 	\centering\includegraphics[scale=0.8]{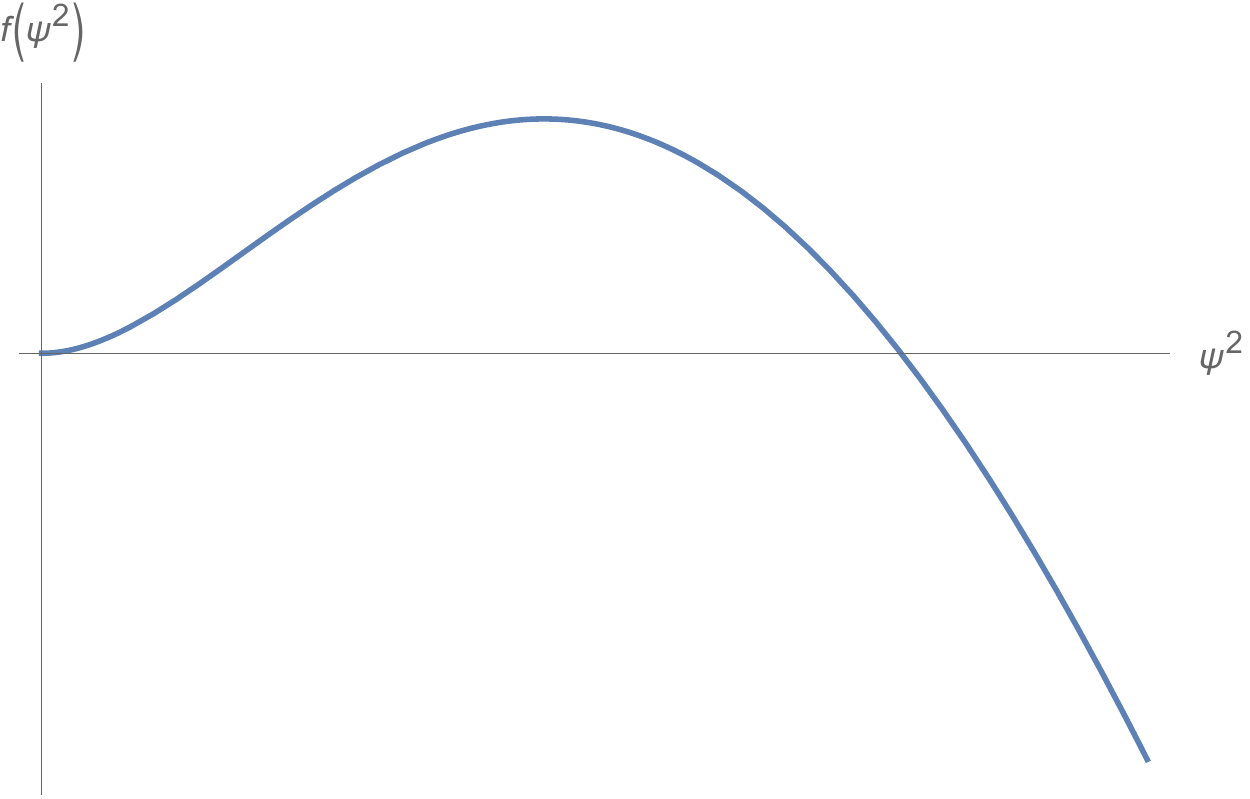}
 	\caption{Plot of $f(\psi^2)$.}
 \end{figure}\\
We thus need to choose $b$ small enough so that $\mc{B}$ is contained in the region where $f$ is positive. The following lemma quantifies the local behaviour of $f$. 
\begin{lem}[Bounds for $f(\psi^2)$] \label{HMTL1}~\\
	For 
	\begin{equation}
	f(\psi^2) = e^{4\pi g\psi^2}4\pi a \psi^2 - \f{3}{2}e^{4\pi g \psi^2}-2\pi g^2\psi^2+\f{3}{2},
	\end{equation}
	we have that for $b$ small enough that in $\mc{B}$ 
	\begin{equation}
	f(\psi^2) \ge 0,
	\end{equation}
	and
	\begin{equation}
	f(\psi^2) \le  8\pi^2ag\psi^4.
	\end{equation}
\end{lem}
\begin{proof}
	Define
	\begin{equation}
	f(x) = e^{4\pi gx}4\pi a x - \f{3}{2}e^{4\pi g x}-2\pi g^2x+\f{3}{2},
	\end{equation}
	and consider its values for $x \ge 0$. As $f(0)=0$, we proceed to study the functions behaviour near this point.\\
	We compute
	\begin{equation}
	f'(x)  = 16\pi^2gaxe^{4\pi gx}+2\pi g^2 e^{4\pi g x}-2\pi g^2,
	\end{equation}
	so $f'(0)=0$. Further computation shows
	\begin{equation}
	f''(x) =  16\pi^2g^2\left(4\pi ax+\kappa \right)e^{4\pi g x} ,
	\end{equation} 
	which remains positive for $x \le \f{1}{2\pi}\f{\kappa}{\f{9}{4}-\kappa^2}$.\\ We thus have the following differential inequality that for $x \in \left[ 0,\f{1}{2\pi}\f{\kappa}{\f{9}{4}-\kappa^2}\right] $
	\begin{equation}
	f''(x)\ge 0, \quad f(0)=f'(0)=0.
	\end{equation}
	Solving this differential inequality yields the first result.\\
	For the second result consider
	\begin{equation}
	g(x) := f(x) - 8\pi^2agx^2,
	\end{equation}
	again we see that $g(0)=0$. Computing
	\begin{equation}
	\begin{split}
	g'(x)&= 16\pi^2 gax\left( e^{4\pi gx} -1\right)  + 2\pi g^2 \left( e^{4\pi g x}-1\right) \le 0,
	\end{split}
	\end{equation}
	and solving this differential inequality yields that
	\begin{equation}
	f(x) \le  8\pi^2agx^2,
	\end{equation}
	for non-negative $x$.
\end{proof}
\begin{coro}\label{HMC1}
	In $\mc{B}$ we have
	\begin{equation}
	\varpi(u,v)\le \varpi(u_\mi,v) = M.
	\end{equation}
\end{coro}
\begin{proof}
	Estimating the $\pa_u\varpi$ from \eqref{HMTL1} we have
	\begin{equation}
	\pa_u \varpi \le \f{4\pi g^2 r_u }{r}\varpi \psi^2.
	\end{equation}
	A Gronwall estimate then implies
	\begin{equation}
	\varpi(u,v) \le  \varpi(u_\mi,v)\exp\left(\int_{u_\mi}^{u}4\pi g^2 \f{r_u}{r}\psi^2du \right)\le \varpi(u_\mi,v) = M. 
	\end{equation}
\end{proof}
\begin{rem}
	We have in $\mc{B}$
	\begin{equation}
	\varpi(u,v_0):=\varpi_0 =\overline{\varpi}\ge M-C_{l,g}b^2\ge \f{M}{2} >0.
	\end{equation}
	for $b^2$  small enough. 
\end{rem}
\begin{proof}
	We use the smallness of $b$, and our initial data estimate \eqref{HMIV}. 
\end{proof}
\begin{coro}\label{HMBB}
	In $\mc{B}$ we have 
	\begin{equation}
	0<\f{M}{2}\le \varpi.
	\end{equation}
\end{coro}
\begin{proof}
	Identical to the proof of corollary \ref{HMC1}, we estimate $\pa_v\varpi$ and use the lower bound for $\varpi_0$.
\end{proof}
\begin{coro}
	In $\mc{B}$ we have
	\begin{equation} \label{HMC3}
	\pa_u \varpi \ge 8\pi^2g\f{ar_ur^2}{l^2}\psi^4-\f{8\pi r^2 r_v}{\Omega^2}(\grt_u\psi)^2 +\f{4\pi g^2 r_u }{r}M \psi^2,
	\end{equation}
	\begin{equation}
	\pa_v \varpi \le 8\pi^2g\f{ar_vr^2}{l^2}\psi^4-\f{8\pi r^2 r_u}{\Omega^2}(\grt_v\psi)^2 +\f{4\pi g^2 r_v }{r}M \psi^2.
	\end{equation}
\end{coro}
\begin{proof}
	This follows immediately from the lemmas \ref{HME}, \ref{HMTL1} and corollary \ref{HMC1}.
\end{proof}

\begin{coro} \label{HMC2}
	In $\mc{B}$ we have that
	\begin{equation}\label{HMEE}
	\pa_u \varpi \le -\f{8\pi r^2 r_v}{\Omega^2}(\grt_u\psi)^2e +\f{4\pi g^2 r_u }{r}M \psi^2 \le 0,
	\end{equation}
	\begin{equation}
	\pa_v \varpi \ge  -\f{8\pi r^2 r_u}{\Omega^2}(\grt_v\psi)^2+ \f{4\pi g^2 r_v }{r}M \psi^2 \ge 0.
	\end{equation}
\end{coro}
\begin{proof}
	This follows immediately from the lemmas \ref{HME}, \ref{HMTL1} and corollary \ref{HMBB}.
\end{proof}

\begin{coro}
	We have in $\mc{B}$
	\begin{equation}
	\abs{\varpi-M} \le C_{a,M,l}\left( \norm{\psi}{\underline{H}^1_{d}}^2(u,v)+b^4\right).
	\end{equation}	
\end{coro}
\begin{proof}
	The bound from above follows trivially from corollary \ref{HMC2}, from below we integrate $\pa_u\varpi$ and use estimate \eqref{HMC3}
	\begin{equation}
	\begin{split}
	\varpi-M= \varpi - \varpi(u_\mi,v)  &= \int_{u_\mi}^u-\f{8\pi r^2 r_v}{\Omega^2}(\grt_u\psi)^2e^{4\pi g\psi^2} + \f{4\pi g^2 r_u }{r}\varpi \psi^2  +\f{r_ur^2}{l^2}f(\psi^2)du\\&\ge
	\int_{u_\mi}^u-\f{8\pi r^2 r_v}{\Omega^2}(\grt_u\psi)^2 + \f{2\pi g^2 r_u }{r}M \psi^2du  +\int_{u_\mi}^u\f{8\pi^2ag}{l^2}\psi^4r^2r_udu
	\\&-C_{a,M}\norm{\psi}{\underline{H}^1_{d}}^2(u,v) -C_{a,l}b^4\int_{u_\mi}^ur_ur^{-4+2\kappa}du \\&\ge -C_{a,M,l}\left( \norm{\psi}{\underline{H}^1_{d}}^2(u,v)+b^4\right). 
	\end{split}
	\end{equation}
\end{proof}
\begin{lem}\label{OHME}
	We have in the regular region
	\begin{equation}
	\abs{\f{\varpi_1-M}{r^{2\kappa}}}\le C_{M,g,l}b^2.
	\end{equation}
\end{lem}
\begin{proof}
	Expanding
	\begin{equation}
	\f{\varpi_1-M}{r^{2\kappa}} = \f{1}{r^{2\kappa}}\left(\varpi e^{-4\pi g\psi^2}-M \right) + \left(1-e^{-4\pi g\psi^2} \right)\f{r^{3-2\kappa}}{2l^2} , 
	\end{equation}
	and applying the estimates
	\begin{equation}
	M- C_{M,g,l}b^2 \le \varpi \le M,
	\end{equation}
	and (for small $b$)
	\begin{equation}
	8\pi g\psi^2 \le 1- e^{-4\pi g\psi^2}\le 4\pi g\psi^2 <0.
	\end{equation}
	We get
	\begin{equation}
	\begin{split}
	\f{\varpi_1-M}{r^{2\kappa}} &\le \f{M}{r^{2\kappa}}\left( e^{-4\pi g\psi^2}-1 \right)\\
	&\le \f{M}{r^{2\kappa}}\cdot 8\pi(-g)r^{-3+2\kappa}b^2\\
	&\le 8\pi M(-g)r_{min}^{-3}b^2 \le C_{M,g}b^2.
	\end{split}
	\end{equation}
	In the other direction
	\begin{equation}
	\begin{split}
	\f{\varpi_1-M}{r^{2\kappa}} &\ge \f{1}{r^{2\kappa}}\left((M- C_{M,g,l}b^2)e^{-4\pi g\psi^2}-M\right) + \f{8\pi gr^{3-2\kappa}\psi^2}{2l^2}\\
	&\ge \f{M}{r^{2\kappa}}\left(e^{-4\pi g\psi^2}-1 \right) -\f{ C_{M,g,l}b^2}{r^{2\kappa}}e^{-4\pi g\psi^2} +\f{8\pi gr^{3-2\kappa}\psi^2}{2l^2}\\
	&\ge \f{M}{r^{2\kappa}}\cdot(-4\pi g\psi^2)-\f{ C_{M,g,l}b^2}{r^{2\kappa}}e +\f{4\pi g}{l^2}b^2\\
	&\ge -M4\pi gb^2r^{-3}-\f{ C_{M,g,l}b^2}{r^{2\kappa}}e +\f{4\pi g}{l^2}b^2\\
	&\ge -  C_{M,g,l}b^2,
	\end{split}
	\end{equation}
	thus concluding the proof.
\end{proof}

\begin{lem}\label{techv}
	Let $\alpha+\beta >2\kappa$ be non negative numbers, and $D>0$, a positive constant. The following estimate holds in $\mathcal{B}$ (for $b$ small enough)
	\begin{equation}
	\f{\varpi_1}{r^\alpha}+Dr^\beta\ge D(1+\epsilon)r^\beta\ge Dr^{\beta}  >0,
	\end{equation}
	where $0<\epsilon\le\f{M}{D}r_{min}^{-(\beta+\alpha)}-c_{M,g,l}b^2$.
	\begin{proof}
		\begin{equation}
		\begin{split}
		\f{\varpi_1}{r^\alpha}+Dr^\beta &\ge Dr^\beta + \f{M}{r^\alpha}- C_{M,g,l}b^2r^{2\kappa-\alpha} \\
		&\ge Dr^{\beta}\left(1+ \f{M}{D}r_{min}^{-(\beta+\alpha)}- \f{C_{M,g,l}b^2}{D}r^{2\kappa-(\alpha+\beta)} \right) \\
		&\ge Dr^{\beta}\left(1+\f{M}{D}r_{min}^{-(\beta+\alpha)}- \f{C_{M,g,l}b^2}{D}r_{min}^{2\kappa-(\alpha+\beta)} \right)\\
		&\ge D(1+\epsilon)r^\beta\ge Dr^{\beta}>0.
		\end{split}
		\end{equation}
	\end{proof}
\end{lem}
\begin{rem}
	This result is purely technical, while we know that $\varpi>0$ we don't necessarily have positivity for $\varpi_1$ in $\sr_\mh$. Estimates for $\varpi_1$ are useful due to how the quantity algebraically interacts with the system. It is often coupled with a term like $Dr^\beta$ which can be used to absorb the negativity.
\end{rem}~\\
\subsection{Estimates for $r_u$}
\begin{thm}\label{thmru}
	In the bootstrap region $\mc{B}$, we have the existence of a uniform constant $C=C(l,g)$ such that
	\begin{equation}
	\f{1}{C}r^2 \le -r_u.
	\end{equation}
\end{thm} 
We split the proof into three lemmas.
\begin{lem}
	In the region $\mc{B}\cap\{r_{min}\le r\le r_X\}$,  we have that 
	\begin{equation}
	r_X^{-2} Cr^2 \le -r_u. 
	\end{equation}
\end{lem}
\begin{proof}
	We restrict ourselves to the region $r_{min}\le r\le r_X$, and consider \eqref{EKG3}, 
	\begin{equation}
	r_{uv} = - \frac{r_ur_v}{r}+ \frac{2\pi a r}{l^2}\Omega^2\psi^2 - \frac{3}{4}\frac{r}{l^2}\Omega^2.
	\end{equation}
	It follows that
	\begin{equation}
	r_{uv} \le - \frac{r_ur_v}{r},
	\end{equation} 
	which we integrate and estimate
	\begin{equation}
	-r_u \ge -\overline{r_u}\cdot\overline{r}\cdot \f{1}{r} \ge \f{r_{min}^3}{2l^2r_X} = : C >0.
	\end{equation}
	We extend to
	\begin{equation}
	r_X^{-2} Cr^2 \le -r_u. 
	\end{equation}	
\end{proof}
We now study the set where $r$ is unbounded.
\begin{lem}[Region Splitting]
	We have on the region $\mc{B}\cap\{r\ge r_Y\}$, the following inequality
	\begin{equation}
	\f{r^2}{l^2}-\f{2M}{r} \ge C_{Y,l,M}r^2.
	\end{equation}
\end{lem}
\begin{proof}
	Let $C_{Y,l}=\min(\f{1}{2l^2},\f{d}{r_Y^2})$ and define
	\begin{equation}
	f(r) = \f{r^2}{l^2}-\f{2M}{r} - C_{Y,l,M}r^2.
	\end{equation}
	Then 
	\begin{equation}
	f(r_Y) = d - C_{Y,l,M}r_Y^2 \ge 0,
	\end{equation}
	and
	\begin{equation}
	f'(r) = \f{1}{r}\left(\f{r^2}{l^2}+\f{2M}{r} \right) + 2r\left(\f{1}{2l^2}-C_{Y,l,M} \right) \ge 0, 
	\end{equation}
	whence the result follows.
\end{proof}
\begin{lem}\label{struest}
	We have in the region $\mc{B}\cap\{r\ge r_Y\}$ that
	\begin{equation}
	\f{1}{C_{Y,M}}r^2 \le -r_u \le C_{Y,M} r^2.
	\end{equation}
\end{lem}
\begin{proof}
	Using the wave equation for the radial function
	\begin{equation}
	\tilde{r}_{uv} = \f{-3\tilde{r}_ur_v}{r}+\f{3}{4l^2}\f{\Omega^2}{r} -\f{2\pi a}{l^2}\f{\psi^2\Omega^2}{r},
	\end{equation}
	we may write this as 
	\begin{equation}
	\tilde{r}_{uv} = \rt_u\cdot f(u,v)r_v,
	\end{equation}
	where
	\begin{equation}
	f(u,v) = \left(\f{6\varpi}{r^2}+\f{3r}{l^2}\left(1-e^{-4\pi g \psi^2}\right) +\f{6\varpi}{r^2}\left(e^{-4\pi g\psi^2}-1 \right) -\f{8\pi a}{l^2}r\psi^2 \right)\f{e^{4\pi g \psi^2}}{\f{r^2}{l^2}-\f{2\varpi}{r}}. 
	\end{equation}
	We have that
	\begin{equation}
	\rt_u =\f{1}{2l^2}\exp\left(\int_{r(v_0)}^{r(v)}f(u,v)dr\right).
	\end{equation}
	We now wish to show that $f$ is integrable. By an elementary estimation, and noting that
	\begin{equation}
	\abs{1-e^{-4\pi g \psi^2}} \le -8\pi g \psi^2,
	\end{equation} 
	for $b^2$ small, we see
	\begin{equation}
	\abs{f} \le \f{1}{\abs{\f{r^2}{l^2}-\f{2M}{r}}}\left(\abs{6Mr^{-2}}+\abs{\f{24\pi(-g)b^2}{l^2}r^{-2+2\kappa}}+\abs{48\pi(-g)b^2r^{-3+2\kappa}}+\abs{\f{2\pi(-a)}{l^2}b^2 r^{-2+2\kappa}}\right).
	\end{equation}
	We now restrict to $r\ge r_Y$, here we have
	\begin{equation}
	\abs{f} \le C_{Y,M}r^{-4+2\kappa},
	\end{equation}
	so there exists a constant $C_{Y,M}$ such that
	\begin{equation}
	\f{1}{2l^2}\exp\left(-C_{Y,M}\int_{r(v_0)}^{r(v)}r^{-4+2\kappa}dr \right)\le \rt_u \le \f{1}{2l^2}\exp\left(C_{Y,M}\int_{r(v_0)}^{r(v)}r^{-4+2\kappa}dr \right).
	\end{equation}
	Integrating gives
	\begin{equation}
	\f{1}{2l^2}\exp\left(C_{Y,M}\left(r^{-3+2\kappa} - \overline{r}^{-3+2\kappa} \right)  \right)\le \rt_u \le \f{1}{2l^2}\exp\left(C_{Y,M}\left(\overline{r}^{-3+2\kappa} - r^{-3+2\kappa} \right) \right),
	\end{equation}
	and we deduce the bound
	\begin{equation}
	\f{1}{2l^2}\exp\left(-C_{Y,M}r_{min}^{-3+2\kappa} \right)\le \rt_u \le \f{1}{2l^2}\exp\left(C_{Y,M}r_{min}^{-3+2\kappa} \right).
	\end{equation}
	From here the result follows.
\end{proof}
From these three lemmas the proof of theorem \ref{thmru} follows.
\subsection{Energy estimates}
\subsubsection{Degenerate energy estimates from the Hawking mass}
\begin{lem}\label{EE1}
	In $\mc{B}$ we have
	\begin{equation}
	\norm{\psi}{\underline{H}^1_{d}}(u,v)\le C_{g,M,\kappa, Y}\left( \norm{\psi}{\underline{H}^1_{d}}(u_{\mh},v_0)+b^4\right).
	\end{equation}
\end{lem}
\begin{proof}
	We start by noting
	\begin{equation}
	\begin{split}
	\int_{u_\mi}^{u_\mh}-\pa_{\bu}\varpi(\bu,v_0)d\bu  &\ge \int_{u_\mi}^u -\pa_{\bu}\varpi(u,v_0)d\bu\\
	=&\varpi|_\mi-\varpi(u,v_0) + \varpi(u,v) - \varpi(u,v)\\ =& \int_{u_\mi}^{u}-\pa_{\bu}\varpi(\bu,v)d\bu + \int_{v_0}^{v} \pa_{\bv}\varpi(u,\bv)d\bv\\
	\ge&\int_{u_\mi}^{u} \f{8\pi r^2 r_v}{\Omega^2}(\grt_u\psi)^2e^{-1} +\f{2\pi g^2 (-r_u) }{r}M \psi^2d\bu \\&+ \int_{v_0}^{v} -\f{8\pi r^2 r_u}{\Omega^2}(\grt_v\psi)^2+ \f{2\pi g^2 r_v }{r}M \psi^2 d\bv\\
	\ge&C_{g,M}\left( \int_{u_\mi}^{u} \f{r^2 r_v}{\Omega^2}(\grt_u\psi)^2 +\f{(-r_u) }{r} \psi^2d\bu + \int_{v_0}^{v} -\f{r^2 r_u}{\Omega^2}(\grt_v\psi)^2+ \f{r_v }{r} \psi^2 d\bv\right). \\
	\end{split}
	\end{equation}
	This gives us control of the $\Hu^1_d$ norm from initial data, providing we can prove the LHS integral is controlled by the $\Hu^1_d$ norm. Recall 
	\begin{equation}
	-\pa_u \varpi \le  \f{8\pi r^2 r_v}{\Omega^2}(\grt_u\psi)^2e^{4\pi g\psi^2} - \f{4\pi g^2 r_u }{r}\varpi \psi^2  -\f{r_ur^2}{l^2}8\pi^2ag\psi^4, 
	\end{equation}
	we see the first two terms form the $\Hu_d^1$ norm. We need to control the final term.\\
	From the bootstrap assumption we form the bound
	\begin{equation}
	\psi^4r^2 \le b^4 r^{-4+4\kappa},
	\end{equation}
	using this in the integral
	\begin{equation}
	\begin{split}
	\int_{u_\mi}^{u}  -8\pi^2g\f{ar_ur^2}{l^2}\psi^4du &\le  b^4 C_{g}\int_{u_\mi}^u\f{r_u}{3-4\kappa}\pa_u\left(r^{-3+4\kappa} \right)du\\
	&= \f{b^4C_a}{3-4\kappa} \left[r^{-3+4\kappa} \right]^u_{u_\mi} = \f{b^4C_g}{3-4\kappa}r^{-3+4\kappa}\\ &\le C_{g,r_{min}}\f{1}{3-4\kappa}b^4\le C_{g,r_{min}} b^4.  
	\end{split}
	\end{equation}
	We conclude
	\begin{equation}
	\int_{u_\mi}^{u_\mh}-\pa_{\bu}\varpi(\bu,v_0)d\bu  
	\le C_{M,g}\norm{\psi}{\underline{H}^1_{d}}^2(u_{\mh},v_0)+C_{g}b^4. 
	\end{equation}
	The result then follows.
\end{proof}~\\
\subsubsection{Energy estimates in $\{r\ge r_Y\}$}~\\ \\
We now wish to improve our estimates to the $\Hu^1$ norm and recover pointwise estimates. We begin in the region where $r_v$ is bounded away from $0$. We will see that the standard theory of Hardy and Sobolev estimates can be recovered with twisted derivatives. 
\begin{lem}[Norm equivalence away from degeneration] \label{NERY} ~\\
	For $\mc{B}\cap\{r\ge r_Y\}$ we have that
	\begin{equation}
	C_{l} \norm{\psi}{\underline{H}_d^1}^2(u,v) \le \norm{\psi}{\underline{H}^1_1}^2(u,v)\le C_{l,Y} \norm{\psi}{\underline{H}_d^1}^2(u,v).
	\end{equation}
\end{lem}
\begin{proof}
	Note that 	on the domain $\{r\ge r_Y\}\cap\mc{B}$
	\begin{equation}
	\begin{split}
	\f{r^2r_v}{\Omega^2} &= \left(\f{r^2}{l^2}-\f{2\varpi}{r} \right) \f{e^{-4\pi g\psi^2}}{-4r_u}r^2,
	\end{split}
	\end{equation}
	estimating from either side
	\begin{equation}
	\begin{split}
	C_{l}\f{r^4}{-r_u} \ge \left(\f{r^2}{l^2}-\f{2\varpi}{r} \right) \f{e^{-4\pi g\psi^2}}{-4r_u}r^2
	\ge\left(\f{r^2}{l^2}-\f{2M}{r} \right) \f{r^2}{-4r_u}\ge C_{l,Y}\f{r^4}{-r_u}.
	\end{split}
	\end{equation}	
\end{proof}
\begin{lem}[Hardy Inequality]\label{HI}
	We have in $\mc{B}$
	\begin{equation}
	\f{1}{C_g}\norm{\psi}{\Hu_1^1}^2(u,v)\le\norm{\psi}{\Hu^1}^2(u,v) \le C_{g}\norm{\psi}{\Hu_1^1}^2(u,v).
	\end{equation}
\end{lem}
\begin{proof}
	Fix $u_1<u_2<u_\mi \in \mc{B}$, and let $\chi$ be a bump function with the following properties 
	\begin{equation}
	\chi(r(u))=
	\begin{cases}
	0  &u\le u_1,\\
	1  &u\ge u_2,\\
	\text{Smooth and bounded by $1$} &\text{ otherwise.}\\
	\end{cases}
	\end{equation}
	Then
	\begin{equation}
	\norm{(1-\chi)\psi}{\Lu^2}^2 = \int_{u_2}^{u} -r_u (1-\chi)^2\psi^2 du \le \sup_{(u_2,u)}r\cdot \int_{u_2}^{u} \f{-r_u}{r} \psi^2du \le C_{r(u_2)} \int_{u_\mi}^{u} \f{-r_u}{r}\psi^2du.
	\end{equation}
	Looking at
	\begin{equation}
	\begin{split}
	\norm{\chi\psi}{\Lu^2}^2&=\int_{u_\mi}^{u_1}\left( \chi\psi r^{\f{3}{2}-\kappa}\right)^2\pa_u\left(\f{r^{-2+2\kappa}}{2-2\kappa} \right)du\\ &= \left[\f{r(\chi\psi)^2}{2-2\kappa} \right]^{u_1}_{u_\mi}+ \int_{u_\mi}^{u_1} \f{1}{1-\kappa}\chi r\psi\grt_u\left( \chi\psi\right) rdu\\
	&\le\f{1}{1-\kappa} \norm{\chi\psi}{\Lu^2} \cdot \left( \int_{u_\mi}^{u_1} \left( \grt_u\left( \chi\psi\right)\right) ^2 \f{r^2}{-r_u}du\right)^\hf,
	\end{split}
	\end{equation}
	we have
	\begin{equation}
	\begin{split}
	\norm{\chi\psi}{\Lu^2}^2 &\le \f{1}{(1-\kappa)^2}\int_{u_\mi}^{u_1} \left( \grt_u\left( \chi\psi\right)\right) ^2 \f{r^2}{-r_u}du\\
	&\le C_{\kappa} \int_{u_\mi}^{u_1}\chi^2(\grt_u\psi)^2\f{r^2}{-r_u} + \f{r^2}{-r_u}\left(\psi\pa_u\chi \right)^2 du.
	\end{split}
	\end{equation}
	Studying the latter term
	\begin{equation}
	\f{r^2}{-r_u}\left(\psi\pa_u\chi \right)^2 = \f{r^2}{-r_u}\left(r_u\psi\pa_r\chi \right)^2 = -r_u\psi^2r^2(\pa_r\chi)^2,
	\end{equation}
	note that due to regularity of $\chi$, and the boundedness of the region
	\begin{equation}
	\int_{u_\mi}^{u_1} -r_u\psi^2r^2(\pa_r\chi)^2 du = \int_{u_2}^{u_1} -r_u\psi^2r^2(\pa_r\chi)^2 du \le C \int_{u_\mi}^{u} \f{-r_u}{r}\psi^2du.
	\end{equation} 	
	This then implies
	\begin{equation}
	\f{1}{C_{\kappa}}\norm{\chi\psi}{\Lu^2}^2 \le \int_{u_\mi}^{u_1}(\grt_u\psi)^2\f{r^2}{-r_u}du + \int_{u_\mi}^{u} \f{-r_u}{r}\psi^2du,
	\end{equation}
	combining all the results
	\begin{equation}
	\norm{\psi}{\Lu^2}^2 = \int_{u_\mi}^{u} -r_u \psi^2 du \le C_\kappa\left( \int_{u_\mi}^{u_1}(\grt_u\psi)^2\f{r^2}{-r_u}du + \int_{u_\mi}^{u} \f{-r_u}{r}\psi^2du\right). 
	\end{equation}
	The result follows.	
\end{proof}
\begin{lem}[Sobolev Inequality]\label{SI}
	For $\psi \in \Hu^1(\mc{B})$ there  exists $C_{g,M,l}>0$ such that
	\begin{equation}
	\abs{r^{-g}\psi}(u,v) \le C_{g,M,l}\left( \norm{\psi}{\Hu^1}(u,v) +\norm{\psi}{\Lu^2}(u_\mh,v)\right)  .
	\end{equation}
\end{lem}
\begin{proof}
	We begin by writing
	\begin{equation}
	\begin{split}
	r^{-g}\psi(u,v)-r^{-g}\psi(z,v)&= \int_z^u\pa_{u'}(r^{-g}\psi)(u',v)du' =\int_z^u(\grt_{u'}\psi)r^{-g}du'\\
	&=\int_z^u\left( \f{\left( \grt_{u'}\psi\right)  r^2}{\sqrt{-r_{u'}}}\cdot r^{-2-g}\sqrt{-r_{u'}}\right) (u',v)du'\\
	&\le \left(\int_z^u\left( \f{(\grt_{u'}\psi)^2r^4}{-r_{u'}}\right)(u',v) du'\right)^\hf \left(\int_z^u\left( -r^{-4-2g}r_{u'}\right)(u',v)du' \right)^\hf \\
	&\le \left(\int_{u_\mi}^u\left( \f{(\grt_{u'}\psi)^2r^4}{-r_{u'}}\right)(u',v) du'\right)^\hf \left(\int_{u_\mi}^u\left( -r^{-4-2g}r_{u'}\right)(u',v)du' \right)^\hf 
	\\
	&\le \norm{\psi}{\Hu^1}(u,v) \left(\int_{u_\mi}^{u}\left( -r^{-4-2g}r_{u'}\right)(u',v) du' \right)^\hf.
	\end{split}
	\end{equation}
	Now as
	\begin{equation}
	\int_{u_\mi}^{u}\left( -r^{-4-2g}r_{u'}\right) (u',v)du' = \left[\f{r^{-3-2g}}{-3-2g} \right]^{r(u,v)}_{\infty} = \f{1}{2\kappa} r^{-2\kappa}(u,v),
	\end{equation}
	we have
	\begin{equation}
	\abs{r^{-g}\psi(u,v)} \le \abs{r^{-g}\psi(z,v)} + \f{1}{2\kappa} r^{-\kappa}(u,v)\norm{\psi}{\Hu^1}(u,v).
	\end{equation}
	Integrating $z$ over the whole $u$ ray
	\begin{equation}
	r^{-g}\psi(u,v)\int_{u_\mi}^{u_\mh}dz \le \int_{u_\mi}^{u_\mh} r^{-g}\abs{\psi(z,v)}dz + Cr^{-\kappa}(u,v)\norm{\psi}{\Hu^1}(u,v)\int_{u_\mi}^{u_\mh}dz.
	\end{equation}
	As $\int_{u_\mi}^{u_\mh}dz=C_{dom.}$ a domain dependant constant. We need only worry about
	\begin{equation}
	\begin{split}
	\int_{u_\mi}^{u_\mh} r^{-g}\abs{\psi(z,v)}dz  &= \int_{u_\mi}^{u_\mh} \sqrt{-r_z}\abs{\psi(z,v)}\cdot \f{r^{-g}}{\sqrt{-r_z}}(z,v)dz \\
	& \le \left(\int_{u_\mi}^{u_\mh}-r_z\psi^2(z,v)dz \right)^\hf \left(\int_{u_\mi}^{u_\mh}\f{r^{-2g}}{-r_z}(z,v)dz \right)^\hf \\
	&\le \norm{\psi}{\Lu^2}(u_\mh,v)\left(\int_{u_\mi}^{u_\mh}\f{r^{-2g}}{-r_z}(z,v)dz \right)^\hf.
	\end{split}
	\end{equation}
	Estimating the latter term
	\begin{equation}
	\int_{u_\mi}^{u_\mh}\f{r^{-2g}}{-r_z}(z,v)dz = \int_{u_\mi}^{u_\mh}\f{r^{-2g}}{r_z^2}\cdot-r_z(z,v)dz = \int_{r_\mh}^{\infty}\f{r^{-2g}}{r_z^2}dr \le C \left[\f{ r^{-2g-3}}{-2g-3} \right]_{r_{\mh}}^{\infty} <C_g r_{min}^{-2\kappa}.
	\end{equation}
	(recalling the results of theorem \ref{thmru}). We thus conclude
	\begin{equation}
	\abs{r^{-g}\psi(u,v)} \le C_{g,M,l}\left( \norm{\psi}{\Hu^1}(u,v) + \norm{\psi}{\Lu^2}(u_\mh,v)\right) . 
	\end{equation}
	where $C_{g,M,l}$ is a positive constant depending on the domain. The result then follows.
\end{proof}
\begin{coro}\label{RSO}
	For $\psi \in \Hu_d^1\left( \{r\ge r_Y\}\cap\mc{B}\right) $  we have that
	\begin{equation}
	\abs{r^{\f{3}{2}-\kappa}\psi}(u,v) \le C_{g,M,l}\left( \norm{\psi}{\underline{H}_d^1}(u,v)+\norm{\psi}{\Lu^2}(u_\mh,v)\right) .
	\end{equation}
\end{coro}
\begin{proof}
	This is just an application of lemma \ref{NERY}, \ref{HI} and \ref{SI}.
\end{proof}~\newline
\subsection{Red shift estimates in  $\{r\le r_X\}$}~ \\ \\
In this region we are bounded away from $\mi$. As such we are not worried about diverging fluxes there, and we do not need to work within the twisted framework. Primarily we are concerned with the degeneration of $r_v$. To combat this we use a redshift argument from \cite{holzegel_stability_2013} adapted to this setting.
\begin{lem}\label{RSNE}
	In the region $\{r\le r_X\}\cap\mc{B}$ we have a constant $C_{X,g}>0$, such that
	\begin{equation}
	\f{1}{C_{X,g}}\norm{\psi}{\Hu_d^1}(u,v)\le \norm{\psi}{H_d^1}(u,v)\le C_{X,g}\norm{\psi}{\Hu_d^1}(u,v).
	\end{equation}
\end{lem}

\begin{lem}
	We have in the region $\{r\le r_X\}\cap\mc{B}$
	\begin{equation}
	\abs{r^{\f{3}{2}}\f{r\psi_u}{r_u}} \le  C_{X,g}\left(\abs{\f{r^\f{5}{2}}{-r_u}\grt_u\psi} +\abs{r^{\f{3}{2}-\kappa}\psi} \right). 
	\end{equation}
\end{lem}
\begin{lem}[Basic Red Shift Estimate] \label{rs1}
	We have that in the region $\{r\le r_X\}\cap \mc{B}$
	\begin{equation}
	\abs{\f{r\psi_u}{r_u}(u,v)} \le C_{a,l,X}\left[\sup_{D(u,v)}\norm{\psi}{\Hu^1_{d}} +\sup_{I(v_0)}\abs{r^{\f{3}{2}}\f{r\psi_u}{r_u}} + \sup_{I(v_0)}\abs{r^{\f{3}{2}-\kappa}\psi} \right] + \abs{a}\abs{\psi}(u,v).
	\end{equation}
\end{lem}
\begin{proof}
	The proof of Lemma 4.9 of \cite{holzegel_stability_2013} can be modified easily to the setting and refer the reader there for the details. The core idea is to write the Klein-Gordon equation as
	\begin{equation}\label{RS1}
	\pa_v\left(\f{r\psi_u}{r_u}(u,v)\right) = -\psi_v + \f{2ar\chi\psi}{l^2}-\f{r\psi_u}{r_u}\rho,
	\end{equation}
	where
	\begin{equation}
	\rho := 2\chi\left[ \f{\varpi_1}{r^2}+\f{r}{l^2}-8\pi r\f{a}{l^2}\psi\right]. 
	\end{equation}
	Now using lemma \ref{techv} we see that 
	\begin{equation}
	\f{\varpi_1}{r^2}+\f{r}{l^2} \ge \f{3}{4l^2}r,
	\end{equation}
	we see that for $b$ sufficiently small enough that red shift factor
	\begin{equation}
	\f{\rho}{\chi} > \f{3r_{\min}}{2l^2}>0.
	\end{equation}
	Integrating \eqref{RS1} by using the Duhamel formula gives
	\begin{equation}
	\begin{split}
	\f{r\psi_u}{r_u}(u,v) &= \left(\f{r\psi_u}{r_u}(u,v_0)\right)\cdot\exp\left(\int_{v_0}^v-\rho(u,\bv)d\bv \right)\\& + \int_{v_0}^v\left[\exp\left(-\int_{\bv}^v\rho(u,\hat{v})d\hat{v} \right)\left(-\psi_v+\f{2r\chi a\psi}{l^2}(u,\bv) \right)   \right]   d\bv. 
	\end{split}
	\end{equation}
	The result then follows by estimating this quantity.
\end{proof}
\begin{thm}[Red Shift Estimate]\label{C3RSE}
	In the region $\{r\le r_X\}\cap\mc{B}$ we have a constant $C_{g,l,X}>0$ such that
	\begin{equation}\label{C2RSSE6}
	\abs{r^{\f{3}{2}}\f{r\psi_u}{r_u}(u,v)} \le C_{g,l,X}\left[\sup_{D(u,v)}\norm{\psi}{\Hu^1_{d}} +\sup_{I(v_0)}\abs{r^{\f{3}{2}}\f{r\psi_u}{r_u}} + \sup_{I(v_0)}\abs{r^{\f{3}{2}-\kappa}\psi} \right],  
	\end{equation}
	furthermore we also have
	\begin{equation}\label{C2RSSE5}
	\abs{r^{\f{3}{2}-\kappa}\psi}(u,v) \le   C_{g,l,X}\left[\sup_{D(u,v)}\norm{\psi}{\Hu^1_{d}} +\sup_{I(v_0)}\abs{r^{\f{3}{2}}\f{r\psi_u}{r_u}} + \sup_{I(v_0)}\abs{r^{\f{3}{2}-\kappa}\psi} \right]. 
	\end{equation}
\end{thm}
\begin{proof}
	We want to drop the $\psi$ on the right hand side of lemma \ref{rs1}, this is done by integrating from the $r_X$ curve toward the horizon in $u$.
	\begin{equation}
	\abs{\psi}(u,v) \le \abs{\psi(u_{r_X},v)} + \int_{u_{r_X}}^u \abs{\f{r\psi_u}{r_u}}\f{-r_u}{r}du.
	\end{equation}
	To make the next estimate clearer we quickly remark that
	\begin{equation}\label{C2RSSE4}
	\int_{u_{r_X}}^u\f{-r_u}{r}du = \ln\left(\f{r_X}{r(u,v)} \right) \le \ln\left(\f{r_X}{r_{min}} \right) < \f{1}{2\abs{a}}.
	\end{equation}
	Inserting \eqref{C2RSSE4} into lemma \ref{rs1}, we see
	\begin{equation}
	\begin{split}
	\abs{\psi}(u,v) \le \abs{\psi(u_{r_X},v)} + C_{a,l,X}\left[\sup_{D(u,v)}\norm{\psi}{\Hu^1_{d}} +\sup_{I(v_0)}\abs{r^{\f{3}{2}}\f{r\psi_u}{r_u}} + \sup_{I(v_0)}\abs{r^{\f{3}{2}-\kappa}\psi} \right]  + \hf\abs{\psi}(u,v).
	\end{split}
	\end{equation}
	Absorbing the $\psi$ term on the LHS gives
	\begin{equation}
	\abs{\psi}(u,v) \le \abs{\psi(u_{r_X},v)} + C_{a,l,X}\left[\sup_{D(u,v)}\norm{\psi}{\Hu^1_{d}} +\sup_{I(v_0)}\abs{r^{\f{3}{2}}\f{r\psi_u}{r_u}} + \sup_{I(v_0)}\abs{r^{\f{3}{2}-\kappa}\psi} \right]. 
	\end{equation}
	Which in this region implies
	\begin{equation}
	\abs{r^{\f{3}{2}-\kappa}\psi}(u,v) \le \abs{r^{\f{3}{2}-\kappa}\psi(u_{r_X},v)} + C_{a,l,X}\left[\sup_{D(u,v)}\norm{\psi}{\Hu^1_{d}} +\sup_{I(v_0)}\abs{r^{\f{3}{2}}\f{r\psi_u}{r_u}} + \sup_{I(v_0)}\abs{r^{\f{3}{2}-\kappa}\psi} \right]. 
	\end{equation}
	Recalling that $\{r\le r_X\}\subset \{r\le r_Y\}$, and invoking corollary \ref{RSO}  
	\begin{equation}
	\abs{r^{\f{3}{2}-\kappa}\psi}(u,v) \le  C_{a,l,X}\left[\sup_{D(u,v)}\norm{\psi}{\Hu^1_{d}} +\sup_{I(v_0)}\abs{r^{\f{3}{2}}\f{r\psi_u}{r_u}} + \sup_{I(v_0)}\abs{r^{\f{3}{2}-\kappa}\psi} \right], 
	\end{equation}
	showing \eqref{C2RSSE5}. \eqref{C2RSSE6} then follows. 
\end{proof}
\begin{coro}\label{rsc}
	In the region $\{r\le r_X\}\cap\mc{B}$ we have a constant $C_{a,l,X}>0,$ such that
	\begin{equation}
	\abs{\f{r^\f{5}{2}}{r_u}\grt_u\psi} \le C_{a,l,X}\left[\sup_{D(u,v)}\norm{\psi}{\Hu^1_{d}} +\sup_{I(v_0)}\abs{r^{\f{3}{2}}\f{r\psi_u}{r_u}} + \sup_{I(v_0)}\abs{r^{\f{3}{2}-\kappa}\psi} \right]  .
	\end{equation}
\end{coro}
\begin{proof}
	We trivially estimate by
	\begin{equation}
	\begin{split}
	\f{r^\f{3}{2}}{-r_u}\grt_u\psi &= \f{r^\f{3}{2}}{-r_u}\gr_u\psi+gr^{\hf}\psi
	\\
	\abs{\f{r^\f{3}{2}}{-r_u}\grt_u\psi} &\le C_X\abs{r^\f{3}{2}\f{r\psi_u}{r_u}}+\abs{gr^\f{3}{2}\psi}.
	\end{split}
	\end{equation}
	Whence the result follows.
\end{proof}
\subsection{Estimates in the whole bootstrap region}
\begin{lem}\label{RSE}
	We have in $\mc{B}$
	\begin{equation}
	\norm{\psi}{\Hu^1_1}(u,v) \le  C\left[\sup_{D(u,v)}\norm{\psi}{\Hu^1_{d}} +\sup_{I(v_0)}\left( \abs{\f{r^\f{5}{2}}{-r_u}\grt_u\psi} +\abs{r^{\f{3}{2}-\kappa}\psi}  \right) \right]. 
	\end{equation}
\end{lem} 
\begin{proof}
	Firstly let us note
	\begin{equation}
	\begin{split}
	\int_{u_\mi}^u \f{r^4}{-r_u}(\grt_u\psi)^2+\f{(-r_u) }{r} \psi^2du &= \int_{u_\mi}^{u}\mathbbm{1}_{\{r\ge r_X\}}\left(  \f{r^4}{-r_u}(\grt_u\psi)^2+\f{(-r_u) }{r} \psi^2\right) du\\& + \int_{u_\mi}^u\mathbbm{1}_{\{r\le r_X\}}\left(  \f{r^4}{-r_u}(\grt_u\psi)^2+\f{(-r_u) }{r} \psi^2\right) du.
	\end{split}
	\end{equation}
	Now from lemma \ref{NERY} we have
	\begin{equation}
	\int_{u_\mi}^{u} \mathbbm{1}_{\{r\ge r_X\}}\left( \f{r^4}{-r_u}(\grt_u\psi)^2+\f{(-r_u) }{r} \psi^2\right) du \le C_{l,Y}\norm{\psi}{\Hu_d^1}(u,v),
	\end{equation}
	and from lemma \ref{RSNE}
	\begin{equation}
	\int_{u_\mi}^u \mathbbm{1}_{\{r\le r_X\}}\left( \f{r^4}{-r_u}(\grt_u\psi)^2+\f{(-r_u) }{r} \psi^2\right) du \le C_{X,g} \int_{u_\mi}^u\mathbbm{1}_{\{r\le r_X\}}\left(  \f{r^4}{-r_u}(\gr_u\psi)^2+\f{(-r_u) }{r} \psi^2\right) du.
	\end{equation}
	Using the results of theorem \ref{C3RSE} we see
	\begin{equation}
	\begin{split}
	\int_{u_\mi}^u \mathbbm{1}_{\{r\le r_X\}}\f{r^4}{-r_u}(\gr_u\psi)^2du &= \int_{u_{\mi}}^u\mathbbm{1}_{\{r\le r_X\}} \left( r^{\f{3}{2}}\f{r\psi_u}{r_u}\right)^2 \f{-r_u}{r}du\\
	&\le \sup_{r\le r_X} \abs{r^{\f{3}{2}}\f{r\psi_u}{r_u}}^2\ln\left( \f{r_X}{r_{min}}\right) \\
	&\le C_{g,l,X}\left[\sup_{D(u,v)}\norm{\psi}{\Hu^1_{d}} +\sup_{I(v_0)}\abs{r^{\f{3}{2}}\f{r\psi_u}{r_u}} + \sup_{I(v_0)}\abs{r^{\f{3}{2}-\kappa}\psi} \right].  
	\end{split}
	\end{equation}
	An application of corollary \ref{rsc} gives
	\begin{equation}
	\begin{split}
	\abs{r^{\f{3}{2}}\f{r\psi_u}{r_u}} = \abs{\f{r^{\f{5}{2}}}{r_u}\left(\grt_u\psi + \f{gr_u}{r}\psi \right)}\le C_{g}\left(\abs{\f{r^\f{5}{2}}{-r_u}\grt_u\psi} +\abs{r^{\f{3}{2}-\kappa}\psi} \right) ,
	\end{split}
	\end{equation}
	and the result follows.
\end{proof}
\begin{coro}
	We have in the bootstrap region 
	\begin{equation}
	\norm{\psi}{\Hu^1}(u,v) \le C_{g,M,\kappa}\left( \norm{\psi}{\Hu^1}(u_{\mh},v_0)+\sup_{I(v_0)}\left(\abs{r^\hf\grt_u\psi} +\abs{r^{\f{3}{2}-\kappa}\psi} \right)  +b^4\right). 
	\end{equation}
\end{coro}
\begin{coro}
	In the region $\mc{B}$ we have that for $b$ sufficiently small
	\begin{equation}
	\abs{r^{\f{3}{2}-\kappa}\psi} < b^{\f{3}{2}}.
	\end{equation}
\end{coro}
\begin{proof}
	From lemma \ref{SI} we have that
	\begin{equation}
	\begin{split}
	\abs{r^{\f{3}{2}-\kappa}\psi}(u,v) &\le  C_{g,M,\kappa}\left( \norm{\psi}{\Hu^1}(u,v)+\norm{\psi}{\Lu^2}(u_\mh,v)\right)\\  &\le C_{g,M,\kappa}\left( \norm{\psi}{\Hu^1}(u_{\mh},v_0)+\sup_{I(v_0)}\left(\abs{\f{r^\f{5}{2}}{-r_u}\grt_u\psi} +\abs{r^{\f{3}{2}-\kappa}\psi} \right)  +b^4\right),
	\end{split}
	\end{equation}
	so from the smallness of the initial data we conclude that
	\begin{equation}
	\abs{r^{\f{3}{2}-\kappa}\psi} < C_{g,M,\kappa}b^2 < b^{\f{3}{2}}.
	\end{equation}
\end{proof}
\begin{thm}\label{C2BSE}
	We have that $\mc{B} = \mathscr{R}_\mh$.
\end{thm}
\begin{proof}
	We know that $\mc{B}$ is an open non empty subset of $\mathscr{R}_\mh$. Now fix a point $(u^*,v^*) \in \mc{B}$, and take a sequence $(u_n,v_n)\to (u^*,v^*)$, as $n \to \infty$, from the continuity of $r$ and $\psi$ we must have that
	\begin{equation}
	\abs{r^{\f{3}{2}-\kappa}\psi}(u^*,v^*) \le  b^{\f{3}{2}}<b.
	\end{equation}    
	So we conclude that $(u^*,v^*)\in \mc{B}$. $\mc{B}$ is then closed, and hence $\mc{B} = \mathscr{R}_\mh$.
\end{proof}
\begin{rem}
	It follows from theorem \ref{C2BSE} that
	\begin{equation}\label{C3EE}
	\begin{split}
	\abs{\varpi(u,v)-M}^{\hf} +& \abs{r^{\f{3}{2}-\kappa}\psi(u,v)} +  \norm{\psi}{\Hu^1}(u,v)\\ &\le C_{l,M,g}\left( \norm{\psi}{\Hu^1}(u_{\mh},v_0)+\sup_{I(v_0)}\left(\abs{\f{r^\f{5}{2}}{-r_u}\grt_u\psi} +\abs{r^{\f{3}{2}-\kappa}\psi} \right) \right). 
	\end{split}
	\end{equation}
	holds in $\mathscr{R}_\mh$.
\end{rem}
\subsection{Consequence of the bootstrap estimates}
\subsubsection{Metric function estimates}
\begin{lem}
	We have in the regular region the estimate
	\begin{equation}
	\abs{2\chi-1}^\hf \le C_{g,M,\kappa}\left( \norm{\psi}{\Hu^1}(u_{\mh},v_0)+\sup_{I(v_0)}\left(\abs{\f{r^{\f{5}{2}}}{-r_u}\grt_u\psi} +\abs{r^{\f{3}{2}-\kappa}\psi} \right) \right). 
	\end{equation}
\end{lem}
\begin{proof}
	Recall that $\chi = \f{\Omega^2}{-4r_u}$ satisfies the equation 
	\begin{equation}
	\pa_u\left( \ln\chi\right) = \f{4\pi r}{r_u}\left(\pa_u\psi \right)^2. 
	\end{equation}
	Integrating this equation gives
	\begin{equation}
	\chi = \chi|_{\mi}\exp\left(\int_{u_\mi}^u-\f{4\pi r}{-r_u}\left(\pa_u\psi \right)^2du \right). 
	\end{equation}
	Using a Young estimate in one direction, and the negativity of the integrand in the other, we have
	\begin{equation}
	\hf\exp\left(-C_{g}\norm{\psi}{\Hu^1}^2(u,v))\right)  \le \chi \le \hf.
	\end{equation}
	The result now follows  from the energy estimate \eqref{C3EE}.
\end{proof}
\begin{coro}
	There exists a constant $C_{g}>0$ such that in the regular region
	\begin{equation}
	-2\left(1+C_gb^2 \right) r_u \le \Omega^2 \le -2 r_u.
	\end{equation}
\end{coro}
\begin{coro}
	In the set $\{r\ge r_Y\}\cap \mathscr{R}_\mh$ we have from lemma \ref{struest} that
	\begin{equation}
	\Omega^2\le C_Y r^2.
	\end{equation}
\end{coro}
\begin{lem}\label{C3RVNI}
	We have that 
	\begin{equation}
	\rt_v|_\mi = -\f{1}{2l^2}.
	\end{equation}
\end{lem}
\begin{proof}
	Recall that from the definition of the Hawking mass
	\begin{equation}
	r_v =  \f{\Omega^2}{-4r_u}\left(\f{r^2}{l^2}-\f{2\varpi}{r} \right)e^{-4\pi g\psi^2}. 
	\end{equation}
	Which implies
	\begin{equation}
	\rt_v =  -\chi\left(\f{1}{l^2}-\f{2\varpi}{r^3} \right)e^{-4\pi g\psi^2}. 
	\end{equation}
	Taking the limit $r\to \infty$, we have that
	\begin{equation}
	\rt_v|_\mi = -\f{1}{2l^2}.
	\end{equation}
\end{proof}
\begin{lem}
	In $\mathscr{R}_\mh$
	\begin{equation}
	r_v \le \f{1}{2l^2}\left( 1+  C_{g,l}b^2\right) r^2.
	\end{equation}
\end{lem}
\begin{proof}
	Write the $\rt_{uv}$ equation \eqref{REKG3} as
	\begin{equation}\label{C3RERVE}
	\rt_{uv} = \varpi\Omega^2e^{-4\pi g\psi^2}\f{3}{4r^4} - \f{1}{rl^2}\Omega^2\left(\f{3}{4}e^{-4\pi g\psi^2}-\f{3}{4} + 2\pi a \psi^2 \right). 
	\end{equation}
	For $\psi^2$ small
	\begin{equation}
	0\le \f{3}{4}e^{-4\pi g\psi^2}-\f{3}{4} + 2\pi a \psi^2 \le \pi g^2\psi^2.
	\end{equation}
	So in \eqref{C3RERVE} we drop the positive terms, and bound by
	\begin{equation}
	\begin{split}
	\rt_{uv} &\ge  - \f{1}{rl^2}\Omega^2\left(\f{3}{4}e^{-4\pi g\psi^2}-\f{3}{4} + 2\pi a \psi^2 \right)\\
	&\ge -C_{g,l}b^2r^{-4+2\kappa}(-r_u).
	\end{split}
	\end{equation}
	We then integrate this inequality, and use lemma \ref{C3RVNI} to see
	\begin{equation}
	\rt_v \ge -\f{1}{2l^2}\left( 1+  C_{g,l}b^2\right). 
	\end{equation}
	The result then follows.
\end{proof}
\begin{lem}
	In $\mathscr{R}_\mh$
	\begin{equation}
	r_v \le \hf\left(1+Cb^2 \right)  \Omega^2.
	\end{equation}
\end{lem}
\begin{proof}
	Integrating \eqref{EKG2} shows
	\begin{equation}
	\begin{split}
	\f{r_v}{\Omega^2}(u,v) &\le \f{r_v}{\Omega^2}(u,v_0)= \f{1}{4}\left(\f{r^2}{l^2} -\f{2\varpi}{r} \right)\f{e^{-4\pi g\psi^2}}{-r_u}(u,v_0)\\&=\hf\left(1 - \f{l^2\varpi}{r^3} \right)e^{-4\pi g\psi^2}(u,v_0) \le \hf\left(1+Cb^2 \right).  
	\end{split}
	\end{equation}
	From here the result follows.
\end{proof}
We now collect these estimates in the convenient corollary: 
\begin{coro}[Global Estimates]\label{GloEst}
	There exists constants $C_i>0$ depending on $g,M,l$ such that in $\mathscr{R}_\mh$
	\begin{equation}
	r_v \le C_{1}r^2\le-C_{2}r_u\le C_3\Omega^2 \le -C_{4}r_u.
	\end{equation}	
	
\end{coro}
\begin{lem}
	In the region $\{r\ge r_Y\}\cap \mathscr{R}_\mh$ we have
	\begin{equation}
	\Omega^2 \le C_{Y,l}r_v.
	\end{equation} 
\end{lem}
\begin{proof}
	\begin{equation}
	\Omega^2 = \f{-4r_ur_v}{-\f{2\varpi}{r}+\f{r^2}{l^2}}e^{4\pi g\psi^2} \le C_{Y,l} r_v.
	\end{equation}
\end{proof}
\begin{coro}
	In the region $\{r\ge r_Y\}\cap \mathscr{R}_\mh$ we have
	\begin{equation}
	r^2 \le C_{Y,l}r_v.
	\end{equation} 
\end{coro}
\begin{coro}[Stronger estimates away from the degeneration]\label{SEst}
	In the region $\mathscr{R}_\mh\cap\{r\ge r_Y\}$, there exists constants $C_i>0$, depending on $Y,g,M,l$ such that
	\begin{equation}
	r^2 \le C_1r_v \le-C_{2}r_u\le C_3\Omega^2 \le C_4 r^2.
	\end{equation}	
\end{coro}
\begin{lem}
	Recalling $\mu_1 = -\frac{2\varpi_1}{r}+\frac{r^2}{l^2},$ from \eqref{auxvar} and defining $\mu =-\frac{2\varpi}{r}+\frac{r^2}{l^2}$. In $\mathscr{R}_\mh$ we have
	\begin{equation}
	\mu_1  = e^{-4\pi g\psi^2}\mu.
	\end{equation}
\end{lem}
\begin{proof}
	This follows from the observation
	\begin{equation}
	\varpi_1 = \varpi e^{-4\pi g\psi^2}-\f{r^3}{2l^2}\left(e^{-4\pi g\psi^2}-1\right).
	\end{equation}
\end{proof}
\begin{coro}
	In $\mathscr{R}_\mh$ we have
	\begin{equation}
	e \mu \ge \mu_1 \ge \mu.
	\end{equation}
\end{coro}
\subsubsection{Pointwise $u-$derivative decay}~\\
To complete our estimates we need to control $\grt_u\psi$  in the region intersecting $\mi$. We proceed by integrating the Klein-Gordon equation in this region.
\begin{lem} \label{direst}
	We have in $\{r\ge r_Y\}\cap \mathscr{R}$
	\begin{equation}
	\abs{r^\hf\grt_u\psi} \le C_{g,M,\kappa,Y}\left( \norm{\psi}{\Hu^1}(u,v)+\sup_{I(v_0)}\left(\abs{r^\hf\grt_u\psi}\right) \right). 
	\end{equation}
\end{lem}
\begin{proof}
	Recall equation \eqref{REKG4}
	\begin{equation}\label{twe}
	\pa_v\left( r\grt_u\psi\right) = -r_u\left(-\hf+\kappa \right) \grt_v\psi -\f{\Omega^2}{4}rV\psi,
	\end{equation}
	where 
	\begin{equation}
	V = \f{2g^2}{r^3}\varpi_1 +\f{8\pi ag}{l^2}\psi^2.
	\end{equation}
	Using the results of the bootstrap arguments, and the metric function estimates we can easily see that in $\{r\ge r_Y\}\cap \mathscr{R}$
		\begin{equation} \label{PotEst}
		\abs{V-\f{2g^2M}{r^3}}\le C_{M,g,l}b^2r^{-3+2\kappa},
		\end{equation}
		from which the following estimate follows
		\begin{equation}
		\abs{V}\le C_{M,l,g}r^{-3+2\kappa}.
		\end{equation}
	Integrate equation \eqref{REKG4} to get
	\begin{equation}
	\begin{split}
	\abs{r\grt_u\psi(u,v)}&\le\abs{r\grt_u\psi(u,v_0)} + \abs{\int_{v_0}^v-r_u\left(-\hf+\kappa \right) \grt_v\psi dv} + \abs{\int_{v_0}^v -\f{\Omega^2}{4}rV\psi dv}.
	\end{split}
	\end{equation}
	Estimating term by term, working from left to right
	\begin{equation}
	\abs{r\grt_u\psi(u,v_0)} = \abs{r^\hf \cdot r^\hf\grt_u\psi(u,v_0)} \le \sup_{I(v_0)}\abs{r^\hf\grt_u\psi(u,v)} \cdot C \cdot r^\hf(u,v_0).
	\end{equation}
	As we are in the regular region we have that $r(u,v)\ge r(u,v_0)$. Thus
	\begin{equation}
	\abs{r\grt_u\psi(u,v_0)} \le C\sup_{I(v_0)}\abs{r^\hf\grt_u\psi(u,v)} \cdot r^\hf(u,v).
	\end{equation}
	For the next term 
	\begin{equation}
	\begin{split}
	\int_{v_0}^v-r_u\left(-\hf+\kappa \right) \grt_v\psi &\le C_{\kappa}\left( \int_{v_0}^v\f{-r_ur^2}{\Omega^2}\left( \grt_v\psi\right)^2dv   \right)^\hf\left( \int_{v_0}^v -r_u\f{\Omega^2}{r^2}dv \right)^\hf \\
	&\le C_{\kappa}\norm{\psi}{\Hu_d^1}(u,v)\left( \int_{v_0}^v -r_u\f{\Omega^2}{r^2}dv \right)^\hf.
	\end{split}
	\end{equation}
	From theorem \ref{thmru}
	\begin{equation}
	\begin{split}
	\int_{v_0}^v -r_u\f{\Omega^2}{r^2}dv
	&\le  C_{M,l,g}\int_{v_0}^vr_vdv \\ &\le C_{M,l,g}r. 
	\end{split}
	\end{equation}
	So we have 
	\begin{equation}
	\abs{\int_{v_0}^v-r_u\left(-\hf+\kappa \right) \grt_v\psi dv} \le C_{M,l,g}r^\hf\norm{\psi}{\Hu^1_d}(u,v). 
	\end{equation}
	For the final term we compute
	\begin{equation}
	\begin{split}
	\int_{v_0}^v -\f{\Omega^2}{4}rV\psi dv &\le C_{M,l,g}\int_{v_0}^{v}r_v r^{-2+2\kappa}\psi dv\\
	&\le  C_{M,l,g}\left(\int_{v_0}^{v}\f{r_v}{r}\psi^2dv \right)^\hf\left(\int_{v_0}^{v} r_v r^{-3+4\kappa}dv \right)^\hf\\
	&\le C_{M,l,g}\norm{\psi}{\Hu_d^1}(u,v)\left(\int_{v_0}^{v} r^{-3+4\kappa}r_vdv \right)^\hf\\
	&\le  C_{M,l,g}r^{-1+2\kappa}\norm{\psi}{\Hu_d^1}(u,v)\\
	&\le C_{M,l,g}r^\hf \norm{\psi}{\Hu_d^1}(u,v).
	\end{split}
	\end{equation}
	Combining all these estimates we have
	\begin{equation}
	\abs{r\grt_u\psi(u,v)} \le C_{M,l,g}r^\hf\left( \norm{\psi}{\Hu_d^1}(u,v)+ \sup_{I(v_0)}\abs{\grt_u\psi(u,v)}   \right),
	\end{equation}
	from which we deduce the result.
\end{proof}
\begin{coro}
	Using theorem \ref{thmru} we see this estimate is equivalent to
	\begin{equation}
	\abs{\f{r^\f{5}{2}}{-r_u}\grt_u\psi} \le C_{g,M,l,Y}\left( \norm{\psi}{\Hu^1}(u,v)+\sup_{I(v_0)}\left(\abs{\f{r^\f{5}{2}}{-r_u}\grt_u\psi} +\abs{r^{\f{3}{2}-\kappa}\psi} \right) \right),
	\end{equation}
	or alternatively
	\begin{equation}
	\abs{\f{r^{\f{5}{2}-\kappa}}{-r_u}\psi_u}\le C_{g,M,l,Y}\left( \norm{\psi}{\Hu^1}(u,v)+\sup_{I(v_0)}\left(\abs{\f{r^\f{5}{2}}{-r_u}\grt_u\psi} +\abs{r^{\f{3}{2}-\kappa}\psi} \right) \right).
	\end{equation}
\end{coro}

\subsection{Completeness of null infinity}
\begin{prop}
	Let $v_m = \sup_{v\ge v_0}\{v|(u_\mh,v)\in\mc{Q}\}$, then $v_m=\infty$.
\end{prop}
\begin{proof}
	This result is an adaptation of the proof in \cite{holzegel_stability_2013} to this setting. Consider curves of constant $r$. In $\mathscr{R}_\mh$ we have that these are timelike and foliate $\mathscr{R}_\mh$. We now have two cases:
	\begin{itemize}
		\item None of the constant $r$ curves have a future limit $(u_\mh,v_m),$ (i.e. they all intersect the horizon).
		\item There is an $R$, such that $r=R$ has a future limit point $(u_\mh,v_m)$. (And hence also true for all $r=R'$ with $R'>R$).
	\end{itemize}
	We deal with the latter case first.\\
	Consider the infinite `zig-zag' curve as depicted below:
	\begin{center}
		\includegraphics[scale=1]{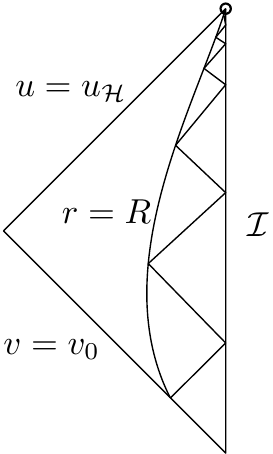}	
	\end{center}
	Now we see that the $v$-length of each constant $u$-piece $\mathcal{U}_i$ is uniformly bounded below. This is done by checking the bounds on $\chi$, and $\f{\mu_1}{r^2}\le \f{2}{l^2}$
	for large enough $R$. In this case we have
	\begin{equation}
	\f{l^2}{2}\chi\f{\mu_1}{r^2} \le \f{1}{2} < 1. 
	\end{equation}
	Recalling that $\chi\mu_1=r_v$, we derive
	\begin{equation}
	\int_{\mc{U}_i}dv \ge \f{l^2}{2}\int\chi\f{\mu_1}{r^2} dv= \f{l^2}{2}\int \f{r_v}{r^2}dv\ge \f{l^2}{2R}. 
	\end{equation}
	There are infinity many $\mc{U}_i$ in the zig-zag curve. (If there were a finite number then there must be some $N\in\mathbb{N}$ such that $\mc{U}_N$ is the ray $\gamma:(u_\mh,v), v\in (v_N,v_m)$. This ray is bounded to right of $r=R$, so we must have that $r=R$ has become null, a contradiction). It follows that $v_m=\infty$.\\ \\
	We now deal with the first case, here we must have that $\lim_{v\to v_m}r(u_\mh,v)=\infty$. We will assume that $v_m=V<\infty$, and contradict that $u=u_\mh$ is the last $u$-ray in which $r=\infty$ can be reached. First pick $r=R$ very large, in view of the bounds on $\varpi,\varpi_1,\chi$ we have that $\mu_1>c>0,$ and $\mu \ge e^{-1}c $ hold in $\overline{\mathscr{R}_\mh}$. This is trivial in $\overline{\mathscr{R}_\mh}\cap \{r\ge R\}$ by computation. For  $\overline{\mathscr{R}_\mh}\cap \{r\le R\}$ we have it is true by compactness, (since $r_v=0$ cannot hold, as this would contradict that $r\to\infty$ along any $u=const$ ray in $\overline{\mathscr{R}_\mh}$). Note that $\mu\ge e^{-1}\mu_1 = \f{r_v}{\eta}e^{-1}>ce^{-1}>0$, thus \eqref{cond2} holds. We now satisfy the conditions of the extension principle near infinity. We extend our spacetime to the depicted triangle
	\begin{center}
		\includegraphics[scale=0.7]{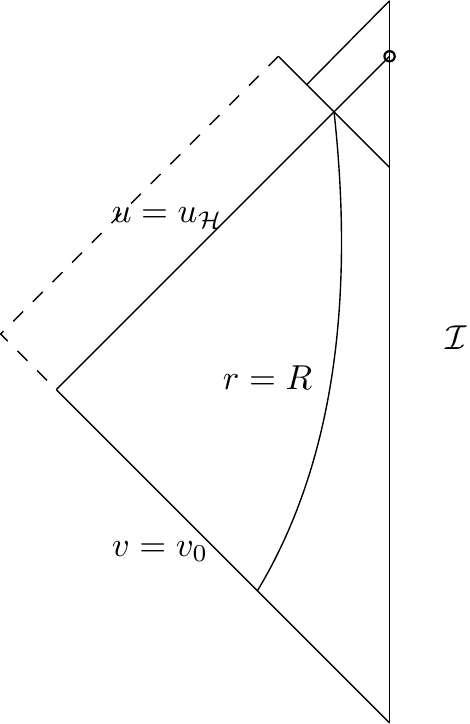}
	\end{center}
	This contradicts the assumption that $u_\mh$ is the last ray along which $\mi$ can be reached.	
\end{proof}
\section{Asymptotic stability}\label{AS}
We have now established that we have a complete black hole spacetime, which is asymptotically AdS. We seek now to prove that the $\Hu^1$ norm of the field $\psi$ is decaying exponentially in the $v$-coordinate to $0$. From here we can see that across any $u=const$ slice, that metric is decaying uniformly to a toroidal AdS Schwarzschild solution of mass $M$. This is in contrast to what was seen for  the linear problem in \cite{dunn_kleingordon_2016}. In that setting it was shown that polynomial decay of the field holds, but exponential decay does not.
The barrier to exponential decay was shown to be linked to null geodesics far away from the  horizon, possessing non-zero angular momentum, which take an arbitrary long time to fall into the black hole. For solutions to our toroidally symmetric problem we now show that this is no longer the case. The symmetry restrictions no longer allow us to construct null geodesics with this property, and the field decays exponentially.
We establish asymptotic stability through Morawetz estimates. The central result we aim to prove in this section is
\begin{equation}\label{asymstab}
\int_D \f{r^4}{-r_u}(\grt_u\psi)^2+(-r_u)\psi^2 +\f{-r_ur^2}{\Omega^2}(\grt_v\psi)^2 + \f{r_v}{r}\psi^2  dudv \le C_{l,g,M}\mathbb{F}(u,v).
\end{equation}
From this we can extract exponential decay of a key flux quantity, implying decay of the field. We follow the vector field method. However due to the complexities from the non-linearities and the ungeometrical nature of twisting, we prefer to work with the standard energy-momentum tensor rather than the twisted version of \cite{holzegel_boundedness_2014} . We use this to prove a global but low weighted integrated decay estimate. The low weights ensure that the technicalities of infinite fluxes are not present, and so twisting is not needed. We then optimise the weights by directly multiplying the Klein-Gordon equation in twisted form, as seen in the classical methods of Morawetz \cite{morawetz_decay_1961}. \\ 
We recall the definition of $\chi$ and introduce $\gamma$ as:
\begin{equation}
\chi := \f{\Omega^2}{-4r_u}, \quad \gamma := \f{\Omega^2}{4r_v}.
\end{equation}
In order to make the proof of this result more manageable we split it into three theorems
\begin{thm}[Low Weighted Degenerate Global Estimate]
	In $\mathscr{R}_\mh$, for $\kappa< \hf$, we have the following estimate 
	\begin{equation}
	\mathbb{I}_{deg}[\psi] := \int_D r^{-6}\left(\f{1}{\gamma^2}\psi_u^2+\f{1}{\chi^2}\psi_v^2\right) \f{\Omega^2r^2}{2}dudv+\int_D \left(\f{1-2 \kappa}{r}\right) \psi^2\f{\Omega^2}{2}r^2dudv \le  C_{l,g,M}\fl(u,v).
	\end{equation}
\end{thm}
\begin{rem}
	While this estimate is low weighted it has the advantage of holding globally on the spacetime. It is insufficient to prove exponential decay of $\psi$, due the degeneration at the boundary of the regular region appearing in the $\gamma$ factor, and that the powers of $r$ are too low to control an integrated $\Hu^1$ norm. However it allow us to localise estimates to either a region near $\mi$, or to a region $\{r\le r_X\}$, where more specialised vector fields can be used. The proof of this theorem is inspired from $\S 5.3$ of \cite{holzegel_stability_2013} but has been expanded into more detail (in particular towards the boundary terms and generalised to cover more choices of multiplier).   	
\end{rem}
\begin{thm}[Red Shift estimate]\label{REDSHIFTMOR}
	In $\mathscr{R}_\mh$ for $\kappa < \hf$, we have the following estimate
	\begin{equation}
	\begin{split}
	&\int_{D}  \f{1}{r^{7}}\left( \f{r^4}{-r_u}(\grt_u\psi)^2+\f{-r_ur^2}{\Omega^2}(\grt_v\psi)^2\right)  d\bar{u}d\bar{v}\\+&\int_D \left(\left( 1-2 \kappa \right) r\right)(-r_u) \psi^2+ \left(\left( 1-2\kappa\right) r\right)r_v \psi^2d\bar{u}d\bar{v}  \le C_{l,g,M}\fl(u,v).
	\end{split}
	\end{equation}
\end{thm}
\begin{rem}
	With the global estimate proven we use smoothed cut-off functions in order to remove the degeneration coming from the $\gamma$ term. This is done using a redshift vector field localised to the region $\{r\le r_X\}$. We may also convert back to using twisted derivatives. This estimate implies local energy decay, that is in any compact region we have the field is decaying exponentially but it is insufficient for a global decay statement.
\end{rem}
\begin{thm}[Morawetz Estimate]
	In $\mathscr{R}_\mh$, for $\kappa< \hf$, we have the following estimate
	\begin{equation}
	\int_D \f{r^4}{-r_u}(\grt_u\psi)^2+(-r_u)\psi^2 +\f{-r_ur^2}{\Omega^2}(\grt_v\psi)^2 + \f{r_v}{r}\psi^2  dudv \le C_{l,g,M}\mathbb{F}(u,v).
	\end{equation}
\end{thm}
\begin{rem}
	To show this, we localise a vector field to a region near $\mi$ where the spacetime is behaving like AdS. In this region the estimates of corollary \ref{SEst} hold. With these estimates we can sharpen the $r$ weights of theorem \ref{REDSHIFTMOR} to a Morawetz estimate. Exponential decay of the fields follows from this estimate. 
\end{rem}
\subsection{Useful estimates and identities}
\begin{lem}[Hardy estimate in $v$]\label{HestV}
	In $\{r\ge r_Y\}\cap \mathscr{R}_\mh$ the following estimate holds 
	\begin{equation}
	\int_{v_0}^v \psi^2 r_vdv \le C_{Y,l,g}\fl(u,v).
	\end{equation}	
\end{lem}
\begin{proof}
	\begin{equation}
	\begin{split}
	\int_{v_0}^v \psi^2 r_vdv & = \int_{v_0}^v\left(\psi r^{\f{3}{2}-\kappa} \right)^2\pa_v\left(\f{r^{-2+2\kappa}}{-2+2\kappa} \right)dv\\
	&= \left[ r\psi^2\f{1}{-2+2\kappa}\right]^v_{v_0} + \f{1}{2-2\kappa}\int_{v_0}^vr\psi\grt_v\psi dv \\
	&\le  C_g\norm{\psi}{\Hu^1}(u,v_0)+ \int_{v_0}^v 2 \left( \sqrt{r_v}\psi\right) \left( \f{r}{\sqrt{r_v}\left(2-2\kappa \right) }\grt_v\psi\right)  dv \\
	&\le   C_g\norm{\psi}{\Hu^1}(u,v_0) + \epsilon\int_{v_0}^v \psi^2 r_vdv +\f{1}{\epsilon\left(4-4\kappa \right)^2 } \int_{v_0}^v \f{r^2}{r_v}\left( \grt_v\psi\right) ^2dv,
	\end{split}
	\end{equation}	
	choosing $\epsilon<1$ and using corollary \ref{SEst} completes the proof.
\end{proof}~\\
\subsubsection{Vector field identities}~\\ \\
Let $X$ be a vector field of the following form: $X^u(u,v)\pa_u+X^v(u,v)\pa_v$.\\ We define the deformation tensor as
\begin{equation}
2{}^{X}\pi^{\alpha\beta}=\gr^\alpha X^\beta + \gr^\beta X^\alpha=g^{\alpha\gamma}\pa_\gamma X^\beta +g^{\beta\delta}\pa_\delta X^\alpha + g^{\alpha\gamma}g^{\beta\delta}g_{\gamma\delta,\mu}X^\mu.
\end{equation} 
We will usually suppress the $X$ in the notation for convenience. We compute the non zero components
\begin{equation}
\pi^{uu} = -\f{2}{\Omega^2}\pa_vX^u, \quad \pi^{vv} = -\f{2}{\Omega^2}\pa_uX^v,
\end{equation}
\begin{equation}
\pi^{uv}= -\f{1}{\Omega^2}\left(\pa_vX^v+\pa_uX^u \right) -\f{2}{\Omega^2}\left(\f{\Omega_u}{\Omega}X^u+\f{\Omega_v}{\Omega}X^v \right), 
\end{equation}
\begin{equation}
\pi^{xx} = \f{1}{r^3}\left(r_uX^u+r_vX^v \right),\quad \pi^{yy} = \f{1}{r^3}\left(r_uX^u+r_vX^v \right) .
\end{equation}
Defining the energy momentum tensor
\begin{equation}
\emt_{\mu\nu}[\psi]= \gr_u\psi\gr_\nu - \hf g_{\mu\nu}\gr_\sigma\psi\gr^\sigma\psi - \f{a}{l^2}\psi^2g_{\mu\nu},
\end{equation}
The non zero components are
\begin{equation}
\emt_{uu} = (\gr_u\psi)^2,\quad \emt_{vv} = (\gr_v\psi)^2,\quad\emt_{uv}=\f{a\Omega^2}{2l^2}\psi^2,
\end{equation}
\begin{equation}
\emt_{xx}=\emt_{yy} =  \f{2}{\Omega^2}r^2\gr_u\psi\gr_v\psi-\f{a}{l^2}r^2\psi^2.
\end{equation}
Computing the divergence of $\mathbb{T}$, we get the usual relation
\begin{equation}
\gr^\mu\emt_{\mu\nu}[\psi]= (\gr_\nu\psi)\left( \Box_g\psi -\f{2a}{l^2}\psi\right) =0.
\end{equation}
Now defining the energy current
\begin{equation}
J_\mu^{X}[\psi] = \emt_{\mu\nu}^XX^\nu,
\end{equation}
and the associated bulk term
\begin{equation}
\begin{split}
K^{X}[\psi]=&\gr^\mu J_\mu^{X}
= \emt_{\mu\nu}\pi^{\mu\nu}. 
\end{split}
\end{equation}
We can expand $K^X$ as
\begin{equation}
\begin{split}
K^{X}[\psi]=& -\f{2}{\Omega^2}(\pa_vX^u)(\pa_u\psi)^2 - \f{2}{\Omega^2}(\pa_uX^v)(\pa_v\psi)^2 \\& +(\pa_u\psi)(\pa_v\psi)\left(\f{4r_u}{\Omega^2 r}X^u+\f{4r_v}{\Omega^2r}X^v\right)\\& -\f{a}{l^2}\psi^2\left( \pa_uX^u+\left(2\f{r_u}{r}+2\f{\Omega_u}{\Omega} \right)X^u+\pa_vX^v+\left(2\f{r_v}{r}+2\f{\Omega_v}{\Omega} \right)X^v  \right). 
\end{split}
\end{equation}
Motivated heavily by the linear theory in \cite{dunn_kleingordon_2016} ($X = F(r)\mc{R}$), we consider a vector field of the form 
\begin{equation}
X = -\f{r_v}{\Omega^2}\cdot F(r)\pa_u +-\f{r_u}{\Omega^2}\cdot F(r)\pa_v. 
\end{equation}
Where $F$ is bounded and sufficiently smooth. We thus compute
\begin{equation}
\pa_vX^u = 4\pi r\f{(\pa_v\psi)^2}{\Omega^2}F-\f{r^2_v}{\Omega^2}F',
\end{equation}
\begin{equation}
\pa_uX^v = 4\pi r\f{(\pa_u\psi)^2}{\Omega^2}F-\f{r^2_u}{\Omega^2}F',
\end{equation}
\begin{equation}
\pa_uX^u +2\f{\Omega_u}{\Omega}X^u +\pa_vX^v +2\f{\Omega_v}{\Omega}X^v = -2\f{r_ur_v}{\Omega^2}F' -2\f{r_{uv}}{\Omega^2}F.
\end{equation}
Then we express $K^X$ as
\begin{equation}
K^{X} = K^{X}_{main}+K^{X}_{error},
\end{equation}
where the terms are defined by
\begin{equation}
\begin{split}
K^{X}_{main} &= 2F'\left( \f{r_v}{\Omega^2}\psi_u+\f{r_u}{\Omega^2}\psi_v\right)^2\\
&+\psi_u\psi_v\left(-\f{4r_ur_v}{\Omega^4}\left(F'+\f{2}{r}F \right) \right)\\
&-\f{a}{l^2}\psi^2\left(-2f-2\f{r_vr_u}{\Omega^2}\left(F'+\f{2}{r}F \right) -\f{2r_{uv}}{\Omega^2}F   \right), \\
K^{X}_{error}  &= \f{-16}{\Omega^4}\pi r \psi_u^2\psi_v^2F.
\end{split}
\end{equation}\\
\subsection{Low weighted global energy estimate}~\\ \\
We now prove the low weighted degenerate global Morawetz estimate. To do this we will study multipliers of the form 
\begin{equation}
F(r) = -r^{-n}.
\end{equation}
We need several lemmas to construct the estimate.
\begin{lem}
	In $\mathscr{R}_\mh$ for $\kappa <\f{1}{2}$ and $n\geq 2$
	\begin{equation}
	\int_{D(u,v)\times \mathbb{T}^2} \gr^\mu J_\mu^{X}  \le C_{l,g,M}\fl(u,v).
	\end{equation}
\end{lem}
\begin{proof}
	Studying the surface terms
	\begin{equation}
	\begin{split}
	\int_{D(u,v)\times \mathfrak{T}^2} \gr^\mu J_\mu^{X} &=  \int_{v_0}^v \left( \emt_{vv}V^v+\emt_{uv}X^u\right) r^2(u,\overline{v})d\overline{v} + \int_{u_\mi}^u\left(\emt_{uu}X^u+\emt_{uv}X^v \right)r^2(\overline{u},v)d\overline{u} 
	\\&-\int_{u_0}^u\left(\emt_{uu}X^u+\emt_{uv}X^v \right)r^2(\overline{u},v_0)d\overline{u}- \int_{\mi}\emt_{\mu\nu}X^{\mu}\hat{n}^\nu d\sigma_{\mi}.
	\end{split}
	\end{equation}~\\
	\textit{$\mi$ surface}~\\ \\
	The metric restricted to constant $\rt$ surfaces is given by
	\begin{equation}
	h=-\Omega^2\f{r_v}{-r_u}dv^2 + r^2\left(dx^2+dy^2 \right). 
	\end{equation}
	The surface form is given by
	\begin{equation}
	d\sigma_{\rt = const}=\sqrt{\abs{\Omega^2\f{r_v}{-r_u}r^4}}dvdxdy,
	\end{equation}
	and the unit normal 
	\begin{equation}
	\hat{n} = \sqrt{\f{-r_v}{\Omega^2r_u}}\pa_u -\sqrt{\f{-r_u}{\Omega^2r_v}}\pa_v.
	\end{equation}
	As this surface is timelike we consider the inward pointing unit vector. We compute
	\begin{equation}
	\hat{n}\sqrt{\abs{h}} = r^2\f{r_v}{-r_u}\pa_u-r^2\pa_v.
	\end{equation}
	Exploring the flux terms we see
	\begin{equation}
	\begin{split}
	\emt_{\mu\nu}X^{\mu}\hat{n}^\nu d\sigma_{\mi}=&\left( \emt_{uu}X^u+\emt_{uv}X^v\right)\f{r_v}{-r_u}r^2 -\left( \emt_{vv}X^v+\emt_{uv}X^u\right)r^2\\=& \left((\gr_u\psi)^2\f{r_v}{\Omega^2}r^{-n} +\f{a\Omega^2}{2l^2}\psi^2\f{r_u}{\Omega^2}r^{-n} \right)\f{r_v}{-r_u}r^2 - \left((\gr_v\psi)^2\f{r_u}{\Omega^2}r^{-n} +\f{a\Omega^2}{2l^2}\psi^2\f{r_v}{\Omega^2}r^{-n} \right)r^2\\
	=&r^{2-n}\f{\Omega^2}{-r_u}\left(\left(\mc{R}\psi \right)^2 - \left(\f{r_u+r_v}{\Omega^2}\right)\gr_u\psi\gr_v\psi+\f{a}{l^2}\f{r_ur_v}{\Omega^2}\psi^2   \right)\\
	=& r^{2-n}\f{\Omega^2}{-r_u}\left(\mc{R}\psi \right)^2  ,
	\end{split}
	\end{equation}
	Now using the relationship 
	\begin{equation}
	\mc{R}\psi = \tilde{\mc{R}}\psi -\f{2gr_ur_v}{r\Omega^2}\psi,
	\end{equation}	
	we see that all these terms vanish on the boundary.\\ \\
	\textit{Fixed  $u$, $v$ surfaces}~\\	
	We study the flux on the surface of fixed $u$. Splitting the region into a section where $r\le r_Y$, we see
	\begin{equation}
	\begin{split}
	\int_{v_0}^{v}\left(\emt_{vv} X^v + \emt_{uv}X^u\right)r^2(u,\bar{v}) d\bar{v} &= \int_{v_0}^{v}\f{-r_u}{\Omega^2r^{n-2}}\psi_v^2 -\f{a}{2l^2r^{n-2}}{r_v}\psi^2dv\\
	&\le \int_{v_0}^v \f{-2r_u}{\Omega^2r^{n-2}}(\grt_v\psi)^2 + \f{-2r_u}{\Omega^2}\cdot\f{g^2r_v^2}{r^n}\psi^2 -\f{a}{2l^2r^{n-2}}{r_v}\psi^2dv.
	\end{split}
	\end{equation}
	The middle term of the integrand can be estimated by 
	\begin{equation}
	\f{-2r_u}{\Omega^2}\cdot\f{g^2r_v^2}{r^n}\psi^2 \le C_{M,l,g}r^{2-n}r_v\psi^2 \le C_{Y,M,l,g}\f{r_v}{r}\psi^2. 
	\end{equation}
	The final term of the integrand, for $n\ge 3$  we can bound by 
	\begin{equation}
	-\f{a}{2l^2r^{n-2}}{r_v}\psi^2 \le -\f{a}{2l^2}\f{r_v}{r}\psi^2.
	\end{equation}
	For $n<3$ we can estimate by
	\begin{equation}
	-\f{a}{2l^2r^{n-2}}{r_v}\psi^2 \le -\f{ar_Y^{-n+3}}{2l^2}\f{r_v}{r}\psi^2.
	\end{equation}
	In the region $r\ge r_Y$, assuming $n \le 2$ 
	\begin{equation}
	\f{-2r_u}{\Omega^2}\cdot\f{g^2r_v^2}{r^n}\psi^2 \le C_{Y,l,g}r^{2-n}r_v\psi^2 \le C_{Y,l,g}r_v\psi^2.
	\end{equation} 
	We can control this by the initial energy using lemma \eqref{HestV}.
	We thus see that
	\begin{equation}
	\int_{v_0}^{v}\left(\emt_{vv} X^v + \emt_{uv}X^u\right)r^2(u,\bar{v}) d\bar{v} \le C_{l,g,M}\fl(u,v).
	\end{equation}
	As the integrand for the other surface is of the form
	\begin{equation}
	\left(\emt_{uu} X^u + \emt_{uv}X^v\right)r^2 = \f{-r_v}{\Omega^2r^{n-2}}\psi_u^2 +\f{a}{2l^2r^{n-2}}r_u\psi^2,
	\end{equation}
	it may be treated in the same way.\\
\end{proof}~\\
\begin{lem}\label{lwe}
	In $\mathscr{R}_\mh$, for $\kappa < \hf$, we have the following global integrated decay estimate
	\begin{equation}
	\mathbb{I}_{deg}[\psi] := \int_D r^{-6}\left(\f{1}{\gamma^2}\psi_u^2+\f{1}{\chi^2}\psi_v^2\right) \f{\Omega^2r^2}{2}dudv+\int_D \left(\f{1-2\kappa}{r}\right) \psi^2\f{\Omega^2}{2}r^2dudv \le  C_{l,g,M}\fl(u,v).
	\end{equation}
\end{lem}
\begin{proof}
	\textit{Bulk terms}~\\
We now look at the bulk term
\begin{equation}
\begin{split}
K^{X}_{main} =& (2+n)r^{-n-1}\left( \f{r_v}{\Omega^2}\psi_u+\f{r_u}{\Omega^2}\psi_v\right)^2\\&-4(2-n)\f{r_ur_v}{\Omega^4}r^{-n-1}\psi_u\psi_v +\f{a}{l^2}\psi^2r^{-n}\left(2(n-2)\f{r_ur_v}{r\Omega^2}-2\f{r_{uv}}{\Omega^2} \right) \\
=&(2+n)r^{-n-1}\left( \f{r_v}{\Omega^2}\psi_u+\f{r_u}{\Omega^2}\psi_v\right)^2 + (n-2)r^{-n-1}\left( \f{r_v}{\Omega^2}\psi_u-\f{r_u}{\Omega^2}\psi_v\right)^2\\&+\f{a}{l^2}\psi^2r^{-n}\left(2(n-2)\f{r_ur_v}{r\Omega^2}-2\f{r_{uv}}{\Omega^2} \right) \\
=&(2+n)r^{-n-1}\left( \mathcal{R}\psi\right)^2 + (n-2)r^{-n-1}\left( \mathcal{T}\psi\right)^2+\f{a}{l^2}\psi^2r^{-n}\left(2(n-2)\f{r_ur_v}{r\Omega^2}-2\f{r_{uv}}{\Omega^2} \right),
\end{split}
\end{equation}
so the first order terms have a good sign for $n\ge 2$.\\
As for the zeroth order terms
\begin{equation}
\underbrace{\f{a}{l^2}\psi^2r^{-n-1}\f{2(n-2)r_ur_v}{\Omega^2}}_{\ge 0}-\f{2a}{l^2}\psi^2r^{-n}\f{r_{uv}}{\Omega^2}.
\end{equation}
The first term has a good sign for ($n\ge 2$). For the second term we analyse through the $r_{uv}$ equation \eqref{EKG3}. We recall this may be expressed as
\begin{equation}
r_{uv} = -\f{\Omega^2}{2}\left(\f{\varpi_1}{r^2}+\f{r}{l^2}-\f{4\pi ar}{l^2}\psi^2 \right). 
\end{equation}
The term thus has the form
\begin{equation}\label{BWZO}
-\f{a}{l^2}r^{-n}\psi^2\left(2\f{r_{uv}}{\Omega^2}\right) = \f{a}{l^2r^n}\underbrace{\left(\f{\varpi_1}{r^2}+\f{r}{l^2}-\f{4\pi a r}{l^2}\psi^2 \right)}_{\ge 0}\psi^2, 
\end{equation}
so it has a negative sign. We proceed by proving a Hardy type inequality to absorb it.\\
First note that from the ($\varpi_1$) Hawking mass equations we can show that 
\begin{equation}
\f{1}{2r_u}\pa_u\mu_1 - \f{8\pi rr_v}{r_u\Omega^2}(\pa_u\psi)^2 = \f{r}{l^2}+\f{\varpi_1}{r^2}-\f{4\pi a r}{l^2}\psi^2,
\end{equation}
and
\begin{equation}
\f{1}{2r_v}\pa_v\mu_1 - \f{8\pi rr_u}{r_v\Omega^2}(\pa_v\psi)^2 = \f{r}{l^2}+\f{\varpi_1}{r^2}-\f{4\pi a r}{l^2}\psi^2.
\end{equation}
We now integrate \eqref{BWZO}  (recall $\sqrt{g} = \f{\Omega^2}{2}r^2$)
\begin{equation}
\begin{split}
&\int_{D(u,v)} \f{a}{l^2r^n}\left(\f{\varpi_1}{r^2}+\f{r}{l^2}-\f{4\pi a r}{l^2}\psi^2 \right)\psi^2 \f{\Omega^2}{2}r^2dudv\\ =& \hf\int_D \f{a}{l^2r^{n-2}}\psi^2\left(-\chi\cdot\mu_{1,u}- \f{4\pi rr_v}{r_u}\left(\pa_u\psi\right)^2\right)  dudv\\
&+ \hf\int_D \f{a}{l^2r^{n-2}}\psi^2\left(\gamma\cdot\mu_{1,v}- \f{4\pi rr_u}{r_v}\left(\pa_v\psi\right)^2\right)  dudv.
\end{split}
\end{equation}
Integrating the first term by parts and recalling the following relation
\begin{equation}
\mu_1 = \f{-4r_ur_v}{\Omega^2},
\end{equation}
and from the Raychaudhuri equations
\begin{equation}
\pa_u\chi = -\f{\Omega^2}{r_u^2}r\pi\psi_u^2,
\end{equation}
and
\begin{equation}
\pa_v\gamma = \f{\Omega^2}{r_v^2}r\pi\psi_v^2.
\end{equation}
We can then compute
\begin{equation}
\mu_1\pa_u\left(\chi r^{-n+2}\psi^2\right) = r^{-n+2}\left( \f{(2-n)r_ur_v}{r}\psi^2+2r_v\psi\psi_u + \f{4\pi rr_v}{r_u}\psi_u^2\right),
\end{equation}
and
\begin{equation}
-\mu_1\pa_v\left(\gamma r^{-3}\psi^2\right) = r^{-n+2}\left( \f{(2-n)r_ur_v}{r}\psi^2+2r_u\psi\psi_v + \f{4\pi rr_u}{r_v}\psi_v^2\right). 
\end{equation}
We see that
\begin{equation}
\begin{split}
&\int_{D(u,v)} \f{a}{l^2r^{n}}\left(\f{\varpi_1}{r^2}+\f{r}{l^2}-\f{4\pi a r}{l^2}\psi^2 \right)\psi^2 \f{\Omega^2}{2}r^2dudv  =\\& \int_{D}\f{4a}{l^2r^{n-2}}\gamma\chi\mu_1\psi\left(\f{1}{4\gamma}\psi_u-\f{1}{4\chi}\psi_v \right)+\f{(2-n)a}{l^2r^{n-1}}r_ur_v\psi^2 dudv - \int_{u_0}^u \f{a}{l^2r^{n-2}}\psi^2(-r_u)(\bar{u},v_0)d\bar{u}\\& + \int_{u_\mi}^u \f{a}{l^2r^{n-2}}\psi^2(-r_u)(\bar{u},v)d\bar{u} - \int_{v_0}^v \f{a}{l^2r^{n-2}}\psi^2r_v(u,\bar{v})d\bar{v}.
\end{split}
\end{equation}
We rewrite as 
\begin{equation}\label{HEp1}
\begin{split}
&\int_{D(u,v)} \f{-a}{l^2r^{n}}\left(\f{\varpi_1}{r^2}+\f{r}{l^2}-\f{4\pi a r}{l^2}\psi^2 \right)\psi^2 \f{\Omega^2}{2}r^2dudv  =\\& \int_{D}\f{-4a}{l^2r^{n-2}}\gamma\chi\mu_1\psi\left(\f{1}{4\gamma}\psi_u-\f{1}{4\chi}\psi_v \right)+\f{(n-2)a}{l^2r^{n-1}}r_ur_v\psi^2 dudv\\ & - \int_{u_0}^u \f{-a}{l^2r^{n-2}}\psi^2(-r_u)(\bar{u},v_0)d\bar{u} + \int_{u_\mi}^u \f{-a}{l^2r^{n-2}}\psi^2(-r_u)(\bar{u},v)d\bar{u} - \int_{v_0}^v \f{-a}{l^2r^{n-2}}\psi^2r_v(u,\bar{v})d\bar{v}. 
\end{split}
\end{equation}
Now using Young's inequality we see that
\begin{equation}
\begin{split}
&\int_{D}\f{4\abs{a}}{l^2r^{n-2}}\gamma\chi\mu_1\psi\left(\f{1}{4\gamma}\psi_u-\f{1}{4\chi}\psi_v \right)dudv\\ \le&  \int_D \f{\abs{a}}{2l^2r^n}\psi^2\left(\f{r}{l^2}+\f{\varpi_1}{r^2} -\f{4\pi r a\psi^2}{l^2} \right) \f{r^2\Omega^2}{2}dudv\\
&+\int_{D}\f{32\abs{a}}{l^2r^{n-2}}\f{\gamma^2\chi^2\mu_1^2}{\Omega^4}\left(\f{r}{l^2}+\f{\varpi_1}{r^2} -\f{4\pi r a\psi^2}{l^2} \right)^{-1}\left(\f{\psi_u}{4\gamma}-\f{\psi_v}{4\chi} \right)^2 \f{\Omega^2}{2}dudv.
\end{split}
\end{equation}
Noting the relation
\begin{equation}
\f{\gamma^2\chi^2\mu_1^2}{\Omega^4} =\f{1}{16},
\end{equation}
we see that
\begin{equation}\label{HEp2}
\begin{split}
&\int_{D}\f{4\abs{a}}{l^2r^{n-2}}\gamma\chi\mu\psi\left(\f{1}{4\gamma}\psi_u-\f{1}{4\chi}\psi_v \right)dudv\\ \le & \int_D \f{\abs{a}}{2l^2r^{n}}\psi^2\left(\f{r}{l^2}+\f{\varpi_1}{r^2} -\f{4\pi r a\psi^2}{l^2} \right) \f{r^2\Omega^2}{2}dudv\\
&+\int_{D}\f{2\abs{a}}{l^2r^{n-1}}\left(\f{1}{l^2}+\f{\varpi_1}{r} -\f{4\pi a\psi^2}{l^2} \right)^{-1}\left(\f{\psi_u}{4\gamma}-\f{\psi_v}{4\chi} \right)^2 \f{\Omega^2}{2}dudv.
\end{split}
\end{equation}
For ease we now define
\begin{equation}
\alpha := \left(1+\f{\varpi_1l^2}{r^{3}}-4\pi a \psi^2 \right)^{-1},
\end{equation}
substituting \eqref{HEp2} into \eqref{HEp1} gives
\begin{equation}
\begin{split}
-\int_{D} \f{2(n-2)a}{l^2r^{n-1}}r_ur_v\psi^2 dudv&+\int_D \f{-a}{l^2r^{n}}\psi^2\left(\f{r}{l^2}+\f{\varpi_1}{r^2}-\f{4\pi ra\psi^2}{l^2} \right)\f{\Omega^2}{2}r^2dudv \le C_{X,Y,l,g,M}\fl(u,v)\\ &+\int_D \f{4\abs{a}\alpha}{r^{n-1}}\left(\f{1}{4\gamma}\psi_u - \f{1}{4\chi}\psi_v \right)^2 \f{\Omega^2}{2}dudv.
\end{split}
\end{equation}
Now
\begin{equation}
\int_{D} -\f{2(n-2)a}{l^2r^{n-1}}r_ur_v\psi^2 dudv = \f{(n-2)a}{l^2}\int_D \f{1}{r^n}\left(\f{-2\varpi_1}{r^2}+\f{r}{l^2} \right)\psi^2 \f{\Omega^2r^2}{2}dudv.
\end{equation}
We then rewrite \eqref{HEp2} as the following Hardy estimate
\begin{equation}
\begin{split}
&\int_{D} \f{2(n-2)a}{l^2r^{n-1}}r_ur_v\psi^2 dudv-\int_D \f{-a}{l^2r^{n}}\psi^2\left(\f{r}{l^2}+\f{\varpi_1}{r^2}-\f{4\pi ra\psi^2}{l^2} \right)\f{\Omega^2}{2}r^2dudv\\ =
&\int_D \f{-a}{l^2r^n}\psi^2\left(\f{r(n-1)}{l^2}+\f{\varpi_1(5-2n)}{r^2}-\f{4\pi ra\psi^2}{l^2} \right)\f{\Omega^2}{2}r^2dudv \le C_{l,g,M}\fl(u,v)\\ &+\int_D \f{4\abs{a}\alpha}{r^{n+1}}\left(\mathcal{R}\psi\right)^2 \f{\Omega^2r^2}{2}dudv.
\end{split}
\end{equation}
The LHS is clearly non negative for $1\le n \le \f{7}{2}$. Examining the integral of $K^{X}$ we have
\begin{equation}
\begin{split}
\int_D K^{X}dudv\f{\Omega^2r^2}{2}dudv =& \int_D(2+n)r^{-n+1}\left( \mathcal{R}\psi\right)^2\f{\Omega^2}{2}dudv + \int_D(n-2)r^{-n-1}\left( \mathcal{T}\psi\right)^2\f{\Omega^2r^2}{2}dudv\\&+\int_D\f{a}{l^2}\psi^2r^{1-n}(n-2)r_ur_vdudv\\&+\int_{D(u,v)} \f{a}{l^2r^n}\left(\f{\varpi_1}{r^2}+\f{r}{l^2}-\f{2\pi a r}{l^2}\psi^2 \right)\psi^2 \f{\Omega^2}{2}r^2dudv\\
\ge &\int_D\left( n+2 -4\abs{a}\alpha\right) r^{-n+1}\left( \mathcal{R}\psi\right)^2\f{\Omega^2}{2}dudv\\&+\int_D(n-2)r^{-n-1}\left( \mathcal{T}\psi\right)^2\f{\Omega^2r^2}{2}dudv\\&-\int_D\f{a}{l^2}\psi^2r^{1-n}(n-2)r_ur_vdudv- C_{l,g,M}\fl(u,v)\\
=&\int_D\left( n+2 -4\abs{a}\alpha\right) r^{-n+1}\left( \mathcal{R}\psi\right)^2\f{\Omega^2}{2}dudv\\&+\int_D(n-2)r^{-n-1}\left( \mathcal{T}\psi\right)^2\f{\Omega^2r^2}{2}dudv\\&+\f{(n-2)a}{2l^2}\int_D \f{1}{r^n}\left(\f{-2\varpi_1}{r^2}+\f{r}{l^2} \right)\psi^2 \f{\Omega^2r^2}{2}dudv - C_{l,g,M}\fl(u,v).\\
\end{split}
\end{equation}
That is
\begin{equation}\label{bulkest}
\begin{split}
&C_{l,g,M}\fl(u,v)+\int_D K^{X}dudv\f{\Omega^2r^2}{2}dudv \ge \int_D\left( n+2 -4\abs{a}\alpha\right) r^{-n-1}\left( \mathcal{R}\psi\right)^2\f{\Omega^2r^2}{2}dudv\\&+\int_D(n-2)r^{-n-1}\left( \mathcal{T}\psi\right)^2\f{\Omega^2r^2}{2}dudv+\f{(n-2)a}{2l^2}\int_D \f{1}{r^n}\left(\f{-2\varpi_1}{r^2}+\f{r}{l^2} \right)\psi^2 \f{\Omega^2r^2}{2}dudv.
\end{split}
\end{equation}
An inspection of the terms implies that if we choose $n=2$ we will lose the clearly negatively signed $\psi^2$ bulk terms. We however need an estimate on the quantity $1-\alpha$. 
\begin{equation}
\begin{split}
1-\alpha =1-\abs{a}\left(1+\f{\varpi_1 l^2}{r^3}-4\pi a \phi^2 \right)^{-1} &= \f{1-\abs{a}}{1+\f{\varpi_1 l^2}{r^3}-4\pi a \phi^2 } +\f{\f{\varpi l^2}{r^3}-4\pi a \phi^2}{1+\f{\varpi l^2}{r^3}-4\pi a \phi^2 }\\
&\ge \f{1}{1+\f{M l^2+cb^2}{r^3}} \left( 1-\abs{a} + \f{\varpi_1 l^2}{r^3}\right) \\
&\ge \f{1}{1+\f{M l^2}{r_{min}^3}+\f{cb^2}{r_{min}^3}} \left( 1-\abs{a} + \f{\varpi_1 l^2}{r^3}\right). \\
\end{split}
\end{equation}
Then recalling $r_+ = 2Ml^2$ using the estimate
\begin{equation}
\abs{r_+^3-r_{min}^3} \le Cb^2,
\end{equation}
we see
\begin{equation}
\begin{split}
1-\abs{a}\left(1+\f{\varpi l^2}{r^3}-4\pi a \phi^2 \right)^{-1} 
&\ge \f{1}{1+\f{M l^2}{r_{min}^3}+\f{cb^2}{r_{min}^3}} \left( 1-\abs{a} + \f{\varpi l^2}{r^3}\right) \\
&\ge \f{1}{1+\hf+\f{cb^2}{r_{min}^3}} \left( 1-\abs{a} + \f{\varpi l^2}{r^3}\right)\\
&\ge \hf \left( 1-\abs{a} + \f{M l^2}{4r^3}\right).
\end{split}
\end{equation}
Choose $n=2$ and restricting to $\abs{a}<1$,  we have
\begin{equation}
\int_D\left( 1 -\abs{a}\right) r^{-3}\left( \mathcal{R}\psi\right)^2\f{\Omega^2r^2}{2}dudv\le C_{l,g,M}\fl(u,v),
\end{equation}
and the $L^2$ estimate
\begin{equation}\label{l2m}
\int_D \left(1-\abs{a}\right)r^{-1} \psi^2\f{\Omega^2}{2}r^2dudv \le  C_{l,g,M}\fl(u,v).
\end{equation}
We now need to recover the $\left( \mc{T}\psi\right) ^2$ terms. To do this we choose $n=5$ in \eqref{bulkest}
\begin{equation}
\begin{split}
&C_{l,g,M}\fl(u,v)+\int_D K^{X}dudv\f{\Omega^2r^2}{2}dudv \ge \int_D\left( 7 -4\abs{a}\right) r^{-6}\left( \mathcal{R}\psi\right)^2\f{\Omega^2r^2}{2}dudv\\&+\int_D3r^{-6}\left( \mathcal{T}\psi\right)^2\f{\Omega^2r^2}{2}dudv+\f{3a}{2l^2}\int_D \f{1}{r^5}\left(\f{-2\varpi_1}{r^2}+\f{r}{l^2} \right)\psi^2 \f{\Omega^2r^2}{2}dudv. 
\end{split}
\end{equation}
Then using estimate \eqref{l2m}
\begin{equation}
\begin{split}
&C_{l,g,M}\fl(u,v)+\int_D K^{X}dudv\f{\Omega^2r^2}{2}dudv\\ \ge &\int_D r^{-6}\left( \mathcal{R}\psi\right)^2\f{\Omega^2r^2}{2}dudv+\int_Dr^{-6}\left( \mathcal{T}\psi\right)^2\f{\Omega^2r^2}{2}dudv+\int_D \f{(1-\abs{a})}{r}\psi^2\f{\Omega^2}{2}r^2dudv,
\end{split}
\end{equation}
allows us to recover the $\mc{T}\psi$ term. Expressing in terms of $u$ and $v$ derivatives
\begin{equation}
\int_D r^{-6}\left(\f{1}{\gamma^2}\psi_u^2+\f{1}{\chi^2}\psi_v^2\right) \f{\Omega^2r^2}{2}dudv+\int_{D}\left(\f{1-\abs{a}}{r}\right) \psi^2\f{\Omega^2}{2}r^2dudv \le  C_{l,g,M}\fl(u,v).
\end{equation}
\end{proof}

\begin{rem}
	This theorem is where the restriction $\kappa<\hf$ comes from. This is due to the choice of $n=2$, it produces terms of the form $r^{n}\psi^2$ which will only decay for $\kappa<\hf$ for the boundary conditions we are considering. If one where to choose $n<2$ to remedy this problem, the issue of positivity from \eqref{bulkest} arises. The bulk term can't be seen to be positive. However as in \cite{holzegel_stability_2013} we expect faster decay for Dirichlet boundary conditions. If one were to follow that scheme in the toroidal setting at the $H^2$ level, one would expect to extend the result to $\kappa=\hf$. Beyond this value seems to be out of reach technically. It is expected to require more sophisticated multipliers.    
\end{rem}~\\
\subsection{The redshift vector field}~\\ \\
We now seeking to remove the degeneration in the estimates due to the $\left( \f{1}{\gamma^2}\right) $ factor. To do this we localise a vector field to a region near the horizon, and exploit a red shift effect. The result of lemma \ref{C3VFRS} is an adapted version of \cite{holzegel_stability_2013} to this setting.
\begin{lem}\label{C3VFRS}
	In $\mathscr{R}_\mh$ we have the following estimate for $\kappa\in(0,\hf)$
	\begin{equation}
	\begin{split}
	&\int_{D}  \f{1}{r^{7}}\left( \f{r^4}{-r_u}(\grt_u\psi)^2+\f{-r_ur^2}{\Omega^2}(\grt_v\psi)^2\right)  d\bar{u}d\bar{v}\\+&\int_D \left(\left( 1-2\kappa\right) r\right)(-r_u) \psi^2+ \left(\left( 1-2\kappa\right) r\right)r_v \psi^2d\bar{u}d\bar{v}  \le C_{l,g,M}\fl(u,v).
	\end{split}
	\end{equation}
\end{lem}
\begin{proof}
	Firstly fix $r_W > r_X$ and define
	\begin{equation}
	Y = (-r_u)^{-1}\eta(r)\pa_u,
	\end{equation}
	where 
	\begin{equation}
	\eta(r)=
	\begin{cases}
	0 & \text{ for } r_W\le r,\\
	\text{smooth, monotonic with bounded derivative} & \text{ for } r_X \le r \le r_W,\\
	1 & \text{ for } r \le r_X.
	\end{cases}
	\end{equation}
	We see that
	\begin{equation}
	\begin{split}
	K^{Y}=\left( -\f{2}{\Omega^2}\pa_v\left( \f{-1}{r_u}\right) \left( \pa_u\psi\right)^2+(\pa_u\psi)(\pa_v\psi)\f{-4}{\Omega^2r}+\f{2a}{l^2}\psi^2+\f{4a\pi}{r_u^2}r\psi^2\psi_u^2\right)\eta(r) + \left(\f{2}{\Omega^2}\psi_u^2+\f{a}{l^2}\psi^2 \right)\eta'(r). 
	\end{split}
	\end{equation}
	The first term is equal to
	\begin{equation}\label{redshiftbulk}
	\left( \f{\psi_u^2}{2r_u^2}\left(\f{\varpi_1}{r^2}+\f{2r}{l^2} \right) + \f{\psi_u}{r_u}\f{1}{r\chi}\psi_v + \f{a}{l^2}\psi^2\f{2}{r}\right) \eta(r). 
	\end{equation}
	We also can quickly compute that
	\begin{equation}
	J^{Y,0}[\psi](Y,\pa_u) = \f{\psi_u^2}{-r_u}\eta, \quad J^{Y,0}[\psi](Y,\pa_v)= \f{a\chi}{l^2}\eta\psi^2, 
	\end{equation}
	asymptotically it is clear that these terms are $0$ at $\mi$. Estimating \eqref{redshiftbulk}, we see
	\begin{equation}
	\begin{split}
	&\left( \f{\psi_u^2}{2r_u^2}\left(\f{\varpi_1}{r^2}+\f{2r}{l^2} \right) + \f{\psi_u}{r_u}\f{1}{r\chi}\psi_v + \f{a}{l^2}\psi^2\f{2}{r}\right) \eta(r)\\
	\ge &\left( \f{3r\psi_u^2}{4l^2r_u^2} + \f{r^\hf\psi_u}{lr_u}\f{l}{r^{\f{3}{2}}\chi}\psi_v + \f{a}{l^2}\psi^2\f{2}{r}\right) \eta(r)\\
	\ge&\left( \f{r\psi_u^2}{4l^2r_u^2}-\f{l^2}{2}\f{1}{r^3\chi^2}\psi_v^2 - \f{2\abs{a}}{r}\psi^2\right) \eta(r).
	\end{split}
	\end{equation}
	Integrating $K^{Y}$ this gives the estimate 
	\begin{equation}\label{redcorest}
	\begin{split}
	&\int_{u_\mi}^{u}\f{\psi_u^2}{r_u^2}r^2(-r_u)(\bar{u},v)\eta d\bar{u} + \int_D\f{r\psi_u^2}{r_u^2}\Omega^2r^2\eta d\bar{u}d\bar{v}\\ 
	\le &C_{M,l,a}\int_{D}\left(\f{\psi_v^2}{r^3\chi^2}+\f{\psi^2}{r} \right)\eta(r)\Omega^2r^2d\bar{u}d\bar{v}+\int_{D}\left(\f{2}{\Omega^2}\psi_u^2+\f{a}{l^2}\psi^2 \right)\eta'(r)\Omega^2r^2d\bar{u}d\bar{v}\\&  + C_{M,l,a}\int_{v_0}^v\f{1}{l^2}\chi\psi^2r^2(u,\bar{v})\eta d\bar{v}
	+\int_{u_\mi}^u \f{\psi_u^2}{r_u^2}r^2(-r_u)\eta d\bar{u}. 
	\end{split} 
	\end{equation}
	It is clear from corollary \ref{SEst} that
	\begin{equation}
	\begin{split}
	\int_{D}\left(\f{2}{\Omega^2}\psi_u^2+\f{a}{l^2}\psi^2 \right)\eta'(r)\Omega^2r^2d\bar{u}d\bar{v} \le C_{M,l,g,Y}\mathbb{I}_{deg}[\psi](u,v).
	\end{split}
	\end{equation}
	Now the first term on the right hand side of \eqref{redcorest}, we already control from lemma \ref{lwe}, the last term is the an initial data norm that we control after a few trivial estimates. We are left to deal with the third term
	\begin{equation}
	\int_{v_0}^v \f{1}{l^2}\chi\psi^2r^2(u,\bar{v})\eta d\bar{v}= \int_{v_0}^ud\bar{u}\pa_u \int_{v_0}^{v^*(u)}\f{1}{l^2}\chi\psi^2r^2\eta d\bar{v}.
	\end{equation}
	Where $v^*(u)$ is the $v$-value where the ray of constant $u$ intersects either $\mi$, or the constant $v$ ray. We pass the derivative through
	\begin{equation}
	\int_{v_0}^v \f{1}{l^2}\chi\psi^2r^2(u,\bar{v})\eta d\bar{v}=\int_D\left(-r\pi\f{\psi_u^2}{r_u^2}\psi^2-\hf\psi\f{\psi_u}{r_u}-\f{1}{2r}\psi^2 \right)\eta\Omega^2r^2 d\bar{u}d\bar{v} + \int_D \f{1}{l^2}\chi\psi^2r^2\eta'r_u d\bar{u}d\bar{v},
	\end{equation}
	The second integrand is estimated using corollary \ref{GloEst}.The $\eta'$ allows us to disregard the $r$ weights.
	The first integrand may be estimated by dropping the negative terms, and a Young inequality 
	\begin{equation}
	\int_{v_0}^v \f{1}{l^2}\chi\psi^2r^2(u,\bar{v})\eta d\bar{v}\le\int_D\left(\epsilon\f{r\psi_u^2}{r_u^2}+\f{1}{16\epsilon}\f{\psi^2}{r} \right) \eta\Omega^2r^2 d\bar{u}d\bar{v} + C_{M,l,g,Y}\mathbb{I}_{deg} [\psi](u,v).
	\end{equation}
	We can an absorb an $\epsilon$ amount of the the derivative terms with the left hand side. Control of the zeroth order terms has already been established.\\
	This leaves us with
	\begin{equation}
	\begin{split}
	&\int_{u_\mi}^{u}\f{\psi_u^2}{r_u^2}r^2(-r_u)(\bar{u},v)\eta d\bar{u} + \int_D\f{r\psi_u^2}{r_u^2}\eta\Omega^2r^2dudv\\ 
	\le &C_{M,l,a}\mathbb{I}_{deg}[\psi](u,v)
	+\int_{u_\mi}^u \f{\psi_u^2}{r_u^2}r^2(-r_u)\eta d\bar{u}.
	\end{split}
	\end{equation}
	Combining this with the global estimate we have that 
	\begin{equation}\label{UTILED}
	\int_D r^{-6}\left(r^4\f{\psi_u^2}{r_u^2}+\f{1}{\chi^2}\psi_v^2\right) \f{\Omega^2r^2}{2}dudv+\int_D \f{(1-\abs{a})}{r} \psi^2\f{\Omega^2}{2}r^2d\bar{u}d\bar{v} \le C_{X,Y,l,g,M}\fl(u,v). 
	\end{equation}
	We now show this holds in a twisted setting. 
	We estimate
	\begin{equation}
	r^4\f{\psi_u^2}{r_u^2} \ge C\left( r\f{\psi_u^2}{r_u^2}\right)  = C \left( \f{r}{r_u^2}(\grt_u\psi)^2 + \f{2g}{r_u}\psi\grt_u\psi +\f{1}{r}g^2\psi^2 \right).  
	\end{equation}
	We apply Young's inequality to get
	\begin{equation}
	r^4\f{\psi_u^2}{r_u^2} \ge   C\left( \f{r}{2r_u^2}(\grt_u\psi)^2 -\f{1}{r}g^2\psi^2 \right). 
	\end{equation}
	We thus have
	\begin{equation}
	\int_D \f{1}{r^6}\left(\f{r}{2r_u^2}(\grt_u\psi)^2 -\f{1}{r}g^2\psi^2 \right)\Omega^2r^2(\bar{u},\bar{v})d\bar{u}d\bar{v} \le C_{l,g,M}\fl(u,v). 
	\end{equation}
	Adding a multiple of the zeroth order terms of \eqref{UTILED} to get
	\begin{equation}
	\int_D \f{1}{r^6}\left(\f{r}{r_u^2}(\grt_u\psi)^2+\f{\psi_v^2}{\chi^2}\right)\Omega^2r^2(\bar{u},\bar{v})d\bar{u}d\bar{v}+\int_D \left(\f{1-\abs{a}}{r}\right) \psi^2\f{\Omega^2}{2}r^2d\bar{u}d\bar{v}
	\le C_{X,Y,l,g,M}\fl(u,v).
	\end{equation}
	With the use of corollary \ref{GloEst} we estimate the $\psi_v$ terms by
	\begin{equation}
	\begin{split}
	\f{\psi_v^2}{\chi^2} &\ge \f{C}{r^3} \f{16r^2_u}{\Omega^4}\left( (\grt_v\psi)^2+\f{2gr_v}{r}\psi\grt_v\psi+ \f{r_v^2}{r^2}\psi^2\right) \\
	&\ge C_l\f{-r_u}{\Omega^2}\left(\f{1}{r^3}(\grt_v\psi)^2+\f{2gr_v}{r^4}\psi\grt_v\psi+ \f{r_v^2}{r^5}\psi^2 \right)\\
	&\ge C_l\f{-r_u}{\Omega^2}\left(\f{1}{2r^3}(\grt_v\psi)^2  - \f{r_v^2}{r^5}\psi^2\right)\\
	&\ge C_{l}\f{-r_u}{\Omega^2}\left(\f{1}{2r^3}(\grt_v\psi)^2  - \f{1}{r}\psi^2\right).    
	\end{split}
	\end{equation}
	This allows us to conclude that
	\begin{equation}
	\begin{split}
	&\int_{D}\f{1}{r^{11}}\left(  \f{r^4}{-r_u}(\grt_u\psi)^2+\f{-r_ur^2}{\Omega^2}(\grt_v\psi)^2\right) \Omega^2r^2(\bar{u},\bar{v})d\bar{u}d\bar{v}\\+&\int_D\left(\f{1-\abs{a}}{r}\right) \psi^2\f{\Omega^2}{2}r^2d\bar{u}d\bar{v}  \le C_{l,g,M}\fl(u,v),
	\end{split}
	\end{equation}
	or as
	\begin{equation}
	\begin{split}
	&\int_{D}  \f{1}{r^{7}}\left( \f{r^4}{-r_u}(\grt_u\psi)^2+\f{-r_ur^2}{\Omega^2}(\grt_v\psi)^2\right)  d\bar{u}d\bar{v}\\+&\int_D \left( 1-\abs{a}\right) r(-r_u) \psi^2+ \left( 1-\abs{a}\right) rr_v \psi^2d\bar{u}d\bar{v}  \le C_{l,g,M}\fl(u,v).
	\end{split}
	\end{equation}
\end{proof}
\subsection{Morawetz estimate}
\begin{thm}
	In $\mathscr{R}_\mh$, for $\kappa < \hf$, the following Morawetz estimate holds
	\begin{equation}\label{ME}
	\int_D \f{r^4}{-r_u}(\grt_u\psi)^2+(-r_u)\psi^2 +\f{-r_ur^2}{\Omega^2}(\grt_v\psi)^2 + \f{r_v}{r}\psi^2  dudv \le C_{l,g,M}\mathbb{F}(u,v).
	\end{equation} 
\end{thm}
\begin{proof}
	Fix $\infty >r_M>r_Z>r_Y$, (we will specify the conditions they need to satisfy later) and let $\eta(r)$ be the cut off function defined by
	\begin{equation}
	\eta(r)=
	\begin{cases}
	0 & \text{ for } r_Z\ge r,\\
	\text{smooth with bounded derivative} & \text{ for } r_M \ge r \ge r_Z,\\
	1 & \text{ for } r \ge r_M.
	\end{cases}
	\end{equation}	
	From equations \eqref{REKG4} and \eqref{REKG5} we derive
	\begin{equation}
	\begin{split}
	\hf \pa_u\left(h(\grt_v\psi)^2  \right) +\hf \pa_v\left(f(\grt_u\psi)^2  \right)&=   \left( \kappa-\hf\right)\grt_u\psi\grt_v\psi\left(\f{r_v}{r}h+\f{r_u}{r}f \right) \\ &+\left(\hf h_u-\f{r_u}{r}h \right)(\grt_v\psi)^2 +\left(\hf f_v-\f{r_v}{r}f \right)(\grt_u\psi)^2\\&- \f{\Omega^2}{4}V\psi\left(h\grt_v\psi+f\grt_u\psi \right).
	\end{split}
	\end{equation}	
	We now make the choice
	\begin{equation}
	f= -\f{r_v}{\Omega^2}{r}\eta, \quad h = -\f{r_u}{\Omega^2}{r}\eta .
	\end{equation}
	\subsubsection{Flux estimates}~\\
	Applying the divergence theorem and examining the flux terms 
	\begin{equation}
	\begin{split}
	&\int_D\hf \pa_u\left(h(\grt_v\psi)^2  \right) +\hf \pa_v\left(f(\grt_u\psi)^2  \right)du'dv'\\
	= &\hf\int_{v_0}^{v}-\f{r_u}{\Omega^2}\eta{r} (\grt_v\psi)^2(u,v')dv'+ \hf\int_{u_\mi}^{u}\f{r_v}{\Omega^2}\eta{r} (\grt_u\psi)^2(u',v)du'\\ &+ \hf\int_{u_0}^{u}\f{r_v}{\Omega^2}\eta{r} (\grt_u\psi)^2(u',v_0)du'+ \int_{\mi}\hf\f{\Omega^2}{-r_u}r\eta\left(\f{r_v^2}{\Omega^4}\left(\grt_u\psi \right)^2 +\f{r_u^2}{\Omega^2}\left(\grt_v\psi \right)^2   \right)  \\ \le& C_{M,l,g}\mathbb{F}(u,v)
	+ \f{1}{12}\int_{\mi}\f{\Omega^2}{-r_u}r\left( \tilde{\mc{R}}\psi\right)^2  .
	\end{split}
	\end{equation}
	The latter term then vanishes due to the boundary conditions (it decays like $r^{-2\kappa})$, (consistent with the linear problem of \cite{dunn_kleingordon_2016}).\\
	\subsubsection{Bulk terms}~\\ \\
	Defining $\hat{f}$ and $\hat{h}$ through
	\begin{equation}
	f= \hat{f}\eta, \quad h = \hat{h}\eta,
	\end{equation}
	and noting the identity
	\begin{equation}
	\begin{split}
	\hf \pa_u\left(\eta \hat{h}(\grt_v\psi)^2  \right) +\hf \pa_v\left(\eta \hat{f}(\grt_u\psi)^2  \right) &= \eta\left( \hf \pa_u\left(\hat{h}(\grt_v\psi)^2  \right) +\hf \pa_v\left(\hat{f}(\grt_u\psi)^2  \right)\right) \\
	&+\eta'\left( r_u \hat{h}(\grt_v\psi)^2+r_v\hat{f}(\grt_u\psi)^2\right).
	\end{split}
	\end{equation}
	We see there are two regions of interest $r\ge r_M$, and $r_M \ge r \ge r_Z$. We deal with the former first	\\	
	\begin{equation}
	\begin{split}
	&\eta\left( \hf \pa_u\left(\hat{h}(\grt_v\psi)^2  \right) +\hf \pa_v\left(\hat{f}(\grt_u\psi)^2  \right)\right)\\=&   -2\left( \kappa-\hf\right)\grt_u\psi\grt_v\psi\left(\f{r_vr_u}{\Omega^2} \right) \\ &+\left(\hf\f{r_u^2}{\Omega^2} +\f{2\pi r^2}{\Omega^2}\psi_u^2\right)(\grt_v\psi)^2 +\left(\hf\f{r_v^2}{\Omega^2} +\f{2\pi r^2}{\Omega^2}\psi_v^2 \right)(\grt_u\psi)^2\\&+ \f{1}{4}rV\psi\left(r_u\grt_v\psi+r_v\grt_u\psi \right) \\
	\ge& 
	-2\left( \kappa-\hf\right)\grt_u\psi\grt_v\psi\left(\f{r_vr_u}{\Omega^2} \right) +\left(\hf\f{r_u^2}{\Omega^2} \right)(\grt_v\psi)^2 +\left(\hf\f{r_v^2}{\Omega^2} \right)(\grt_u\psi)^2\\&+ \f{1}{4}rV\psi\left(r_u\grt_v\psi+r_v\grt_u\psi \right).
	\end{split}
	\end{equation}
	We apply Young's inequality to see 
	\begin{equation}
	\begin{split}
	\hf \pa_u\left(\hat{h}(\grt_v\psi)^2  \right) +\hf \pa_v\left(\hat{f}(\grt_u\psi)^2  \right)& \ge 
	\left(\left( \hf-\abs{\kappa-\hf}\right) \f{r_u^2}{\Omega^2} \right)(\grt_v\psi)^2 +\left(\left( \hf-\abs{\kappa-\hf}\right)\f{r_v^2}{\Omega^2} \right)(\grt_u\psi)^2\\&+ \f{1}{4}rV\psi\left(r_u\grt_v\psi+r_v\grt_u\psi \right).
	\end{split}
	\end{equation}
	So the first row terms are manifestly positive. We then estimate
	\begin{equation}
	\begin{split}
	\f{1}{4}rV\psi\left(r_u\grt_v\psi+r_v\grt_u\psi \right) &\le \f{1}{8} r^2\psi^2 + \f{1}{8}V^2\left(r_u^2\left( \grt_v\psi\right)^2+r_u^2\left( \grt_u\psi\right)^2  \right)\\ &\le C_{Y,l,M,g}\left( r^2\psi^2+ r^{-2+4\kappa}\left( \left( \grt_u\psi\right)^2 +\left( \grt_v\psi\right)^2 \right)  \right) .
	\end{split}
	\end{equation}
	Using the estimates in corollary \eqref{SEst}, we have in this region for $r_M$ chosen large enough
	\begin{equation}
	\int_D \f{r^4}{-r_u}(\grt_u\psi)^2+(-r_u)\psi^2 +\f{-r_ur^2}{\Omega^2}(\grt_v\psi)^2 + \f{r_v}{r}\psi^2  dudv \le C_{l,g,M}\mathbb{F}(u,v).
	\end{equation} 
	Then in the latter region  
	\begin{equation}
	\begin{split}
	\int_D\eta\left( \hf \pa_u\left(\hat{h}(\grt_v\psi)^2  \right) +\hf \pa_v\left(\hat{f}(\grt_u\psi)^2  \right)\right)
	+\eta'\left( r_u \hat{h}(\grt_v\psi)^2+r_v\hat{f}(\grt_u\psi)^2\right)dv.
	\end{split}
	\end{equation} 
	We see that as the derivative of $\eta$ is bounded, and $r$ is bounded above and below in the region where the stronger estimates \eqref{SEst} hold. We can then trivially bound these terms above by the global estimate \eqref{lwe}.
	We then combine this higher weighted estimate with the global one to see 
	\begin{equation}
	\int_D \f{r^4}{-r_u}(\grt_u\psi)^2+(-r_u)\psi^2 +\f{-r_ur^2}{\Omega^2}(\grt_v\psi)^2 + \f{r_v}{r}\psi^2  dudv \le C_{l,g,M}\mathbb{F}(u,v).\qedhere
	\end{equation}
\end{proof}
\subsection{Exponential decay}
\begin{thm}\label{expoDecay}
	Defining the $v$-flux,
	\begin{equation}
	\mc{F}(v) = \int_{u_\mi}^{u_\mh}  \left( \f{r^4}{-r_u}(\grt_u\psi)^2+(-r_u)\psi^2 \right) (\bar{u},v)d\bar{u},
	\end{equation}
	and the region $\tilde{D}(v_1,v_2) = D(u_\mh,v_2)\cap \{v\ge v_1\}$. Then for $\kappa< \hf$,
	\begin{equation}\label{expoDecay2}
	\mc{F}(v) \le \tilde{C}_{M,l,g}\mc{F}(v_0)e^{-\alpha v}. 
	\end{equation}
	for some uniform $\alpha>0.$
\end{thm}
\begin{proof}
	Applying the estimate \eqref{ME} on the region $\tilde{D}(v_{n},v_{n+1})$ yields the estimate
	\begin{eqnarray}
	\mc{F}(v_{n+1}) + c_{M,l}\int_{v_n}^{v_{n+1}}\mc{F}(\bar{v})d\bar{v} \le C_{M,l,g} \mc{F}(v_{n}).
	\end{eqnarray}
	From which the result follows from the standard pigeon hole principle arguments.
\end{proof}
\begin{coro}\label{expodec}
	We have in $\mathscr{R}_\mh$ for $\kappa<\hf$, that
	\begin{equation}
	\sup_{u}\abs{2\chi(u,v)-1}+\sup_{u}\abs{\varpi(u,v)-M} \le \tilde{C}_{M,l,g}\exp\left(-\alpha\cdot v \right), 
	\end{equation}
	and	
	\begin{equation}
	\abs{\psi(u,v)} \le \tilde{C}_{M,l,g}r^{-\f{3}{2}+\kappa}\exp\left(-\alpha\cdot v \right).
	\end{equation}
\end{coro}
\begin{proof}
	This follows from \eqref{OSE} and \eqref{expoDecay2}.
\end{proof}
\begin{rem}
	It is in this sense that we say the metric is converging to a Toroidal Schwarzschild-AdS metric of mass $M$, exponentially in $v$, in the Eddington Finkelstein gauge. 
\end{rem}
\begin{coro}\label{expohm1}
	In $\mathscr{R}_\mh$, for $\kappa < \hf$, we have
	\begin{equation}
	\abs{\f{\varpi_1(u,v)-M}{r^{2\kappa}}}\le \tilde{C}_{M,l,g}\exp\left(-\alpha\cdot v \right).  
	\end{equation}
\end{coro}
\begin{proof}
	Write 
	\begin{equation}
	\begin{split}
	\f{\varpi_1(u,v)-M}{r^{2\kappa}} =& r^{-2\kappa}\left(\varpi(u,v)-M\right) \\&+r^{-2\kappa}\left(e^{-4\pi g\psi^2}-1 \right)\varpi - \f{r^{3-2\kappa}}{2l^2}\left(e^{-4\pi g\psi^2}-1 \right). 
	\end{split}
	\end{equation}
	Taking absolute values the result then follows from theorem \ref{basicests} and corollary \ref{expodec}.
\end{proof}
\begin{coro}
	We have that for $\kappa < \hf$, the Lorentzian Penrose inequality
	\begin{equation}
	\sup_{\mh}r \le r_+.
	\end{equation}
	Furthermore we have along $\mh$ that $r$ converges to $r_+$ exponentially in $v$.
\end{coro}
\begin{proof}
	This is an adaptation of the proof in \cite{holzegel_stability_2013} to this setting.\\	
	Assume for contradiction that $r\ge r_+ + \delta$ for some $\delta>0,$ held along $\mh$. Then by corollary \ref{expohm1} we have the existence of a $v_i>v_0$ such that
	\begin{equation}\label{C3rehb}
	\f{\mu_1}{r^2} \ge \left(\f{1}{l^2}-\f{2M}{r^3} \right)-\f{2\abs{M-\varpi_1}}{r^3}\ge c_{M,l}\delta, 
	\end{equation}
	holds on $\mh\cap\{v\ge v_i\}$. Now integrating $r_vr^{-2}$ we see
	\begin{equation}
	\int_{v_i}^{v}\f{r_v}{r^2}d\bar{v} = - \f{1}{r(v)}+\f{1}{r(v_i)}\le C,
	\end{equation}
	holds for some uniform $C>0$. However from \eqref{C3rehb} we have 
	\begin{equation}
	\int_{v_i}^{r_v}\f{r_v}{r^2}d\bar{v} = \int_{v_i}^{v} \f{\mu}{r^2}\chi d\bar{v} \ge C_{M,l}\delta\cdot \left(v-v_i \right) ,
	\end{equation}
	which is clearly a contradiction.
	\\ Now that we have seen $r$ is bounded along $\mh$, we can prove the exponential decay through an integrated decay estimate.
	\begin{equation}
	\begin{split}
	\int_{v}^\infty \left( r_+-r(u_\mh,\bar{v})\right)  d\bar{v} &\le C_{M,l}\int_{v}^\infty\mu_1d\bar{v} +\tilde{C}_{M,l,g}e^{\left( -B_{M,l,g}v\right)} 	\\
	&\le C_{M,l}\int_{v}^\infty r_v d\bar{v}+ \tilde{C}_{M,l,g}e^{ -B_{M,l,g}v}\\	
	&\le C_{M,l}\left( r_+-r(u_\mh,v)\right)+ \tilde{C}_{M,l,g}e^{ -B_{M,l,g}v}.
	\end{split}
	\end{equation}
	From the positivity of $r_+-r(u_\mh,v)$ we derive the integrated decay statement
	\begin{equation}
	\int_{v_1}^{v_2} \left( r_+-r(u_\mh,\bar{v})\right)  d\bar{v} \le  C_{M,l}\left( r_+-r(u_\mh,v_1)\right)+ \tilde{C}_{M,l,g}e^{-C_{M,l,g}v_1}.
	\end{equation}
	Exponential decay follows in a similar manner to theorem \ref{expoDecay}.
\end{proof}
\subsection{The main theorem}
\begin{thm}\label{MT}
	For a weak solution arising from small initial data of the Einstein--Klein-Gordon system, within the class of square flat toroidal symmetries satisfying: Dirichlet or Neumann boundary conditions and a Klein-Gordon mass bound $ \kappa \le \hf$. The associated maximal development of the Lorentzian manifold is a black hole spacetime, with a regular future horizon and a complete null infinity. Furthermore for $\kappa< \hf$ the estimates of \eqref{OSE} and \eqref{asymstab} hold for any $(u,v)$ in the regular region exterior to the apparent horizon. This implies that $\psi$ decays exponentially in $v$ on this region.     
\end{thm}
We may remark that we can use these techniques to study toroidally symmetric solutions of the Klein-Gordon equation on a fixed toroidal AdS Schwarzschild background. In this decoupled setting we have $\chi=\gamma=\hf$ and $\mc{T} = \pa_t$ and the following corollary.
\begin{coro}
	Let $(\mc{M},g)$ be a fixed toroidal Schwarzschild AdS spacetime with Eddington Finkelstein coordinate system $(u,v)$. Let $\kappa < \hf$ then the toroidally symmetric solutions of the Klein-Gordon equation decay exponentially in the $v$ coordinate on the black hole exterior.
\end{coro}
It is worth contrasting this with the non symmetric results of \cite{dunn_kleingordon_2016}, where only polynomial decay can be established. 
\section{Vacuum result}
So far we have restricted to square flat toroidal symmetry in order to emulate the situation familiar in spherical symmetry where a Birkhoff theorem holds. There are more degrees of freedom in putting a flat metric on a torus, in contrast to the round metric on a sphere, and we can evade trivial vacuum dynamics within a rectangular flat toroidal symmetry class by making the following metric ansatz
\begin{equation}\label{metans}
g = -\Omega^2(u,v)dudv + r^{2}(u,v)\left(e^{-\sqrt{8\pi}B(u,v)}dx^2+e^{\sqrt{8\pi}B(u,v)}dy^2 \right). 
\end{equation}
Here the ratio of the periods of the tori are allow to vary, parameterised in effect by a scalar field $B$. At different points $(u,v)$, we get rectangular tori which, unlike the case $B=const$, we cannot scale back to a unit torus through coordinate transformations of $x$ and $y$. If we study the, now dynamical, vacuum equations
\begin{equation}
R_{\mu\nu}-\hf g_{\mu\nu}R-\f{3}{l^2}g_{\mu\nu}=0,
\end{equation} 
the symmetry reduction leads to the system
\begin{equation}\label{EV1}
\pa_u\left(\frac{r_u}{\Omega^2} \right) = -4\pi r\frac{\left(B_u\right)^2 }{\Omega^2},
\end{equation}
\begin{equation}\label{EV2}
\pa_v\left(\frac{r_v}{\Omega^2} \right) = -4\pi r\frac{\left(B_v\right)^2}{\Omega^2},
\end{equation}
\begin{equation}\label{EV3}
r_{uv} = - \frac{r_ur_v}{r} - \frac{3}{4}\frac{r}{l^2}\Omega^2,
\end{equation}
\begin{equation}\label{EV4}
\left(\log\Omega\right)_{uv} = \frac{r_ur_v}{r^2}- 4\pi B_uB_v,
\end{equation}
\begin{equation}\label{EV5}
B_{uv} =-\frac{r_u}{r}B_v-\frac{r_v}{r}B_u.
\end{equation}
We notice this system is equivalent to \eqref{EKG1} - \eqref{EKG5} where the Klein-Gordon field is massless ($a=0$). In contrast to the Bianchi IX system as studied in \cite{dold_global_2017}, the scalar curvature of the group orbits is $0$. Consequently \eqref{EV5} is a linear wave equation, making the analysis much simpler. However as $a=0$ corresponds to $\kappa = \f{3}{2}$, we cannot currently hope to pose any other boundary conditions other than Dirichlet. Intuitively this makes sense: imposing Dirichlet boundary conditions would mean fixing the periods of the torus at null infinity. 
The main results of this paper do not directly apply for this value of $\kappa$, however, as discussed elsewhere there are only very minor differences between the spherical and toroidal systems (at the level of the reduced equations). It is clear that combining the results of \cite{holzegel_self-gravitating_2012}, and \cite{holzegel_stability_2013} with the observations above we have:
\begin{thm}
	Consider an initial free data set $(\bar{r},\bar{B})$, obeying the conditions of an $\epsilon$-perturbed Schwarzschild-AdS data set (as defined in \cite{holzegel_stability_2013}, taking $a=0)$ on a null ray $N(v_0)$, and Dirichlet boundary conditions. The associated maximal development is a black hole spacetime with a regular future horizon, and a complete null infinity. Furthermore the estimate
	\begin{equation}
	\abs{r^{\frac{3}{2}-\kappa}B(u,v)}\le D\exp(-Cv),
	\end{equation}
	holds on the intersection of the regular region of the spacetime and the exterior of the black hole.
	From which we may deduce that the metric is converging exponentially in $v$, uniformly in $u$, to a toroidal AdS Schwarzschild solution with mass $M$, in the Eddington-Finkelstein gauge.
\end{thm}
Thus the toroidal AdS black hole is indeed a stable solution to the vacuum equations within the symmetry class imposed by the metric ansatz (\ref{metans}).
\section{Acknowledgements}
Claude Warnick acknowledges support from EPSRC through the Cambridge Centre for Analysis.

	\bibliographystyle{alpha}
	\bibliography{MyLibrary2}
\end{document}